\documentclass[a4paper, 11pt]{article}

\usepackage{amsmath}
\usepackage{amssymb}
\usepackage{amsfonts}
\usepackage{amsthm}
\usepackage{bm}
\usepackage{caption}
\usepackage{subcaption}
\usepackage{dsfont}
\usepackage{setspace}
\usepackage{geometry}
\usepackage[hidelinks]{hyperref}
\usepackage[inline]{enumitem}
\usepackage{rotating}

\usepackage{accents}

\usepackage{algorithm}
\usepackage{algpseudocode}
\makeatletter
\newlength{\trianglerightwidth}
\algnewcommand{\LineComment}[1]{\Statex \hskip\ALG@thistlm%
	\parbox[t]{\dimexpr\linewidth-\ALG@thistlm}{\hangindent=\trianglerightwidth \hangafter=1 \strut$\triangleright$ #1\strut}}
\makeatother

\usepackage{multirow}

\usepackage{tikz}
\usepackage{pgfplots}

\usepackage{authblk}

\newcommand{\rbracket}[1]{\ensuremath{\left( #1\right)}}
\newcommand{\sbracket}[1]{\ensuremath{\left\lbrack #1\right\rbrack}}
\newcommand{\cbracket}[1]{\ensuremath{\left\lbrace #1\right\rbrace}}
\newcommand{\set}[2]{\ensuremath{\cbracket{#1:#2}}}

\newcommand{\indicator}[1]{\ensuremath{\mathds{1}_{\lbrace{#1}\rbrace}}}

\newcommand{\converge}[1][]{\ensuremath{\xrightarrow{#1}}}

\newcommand{\prob}[2][\rbracket]{\ensuremath{\mathbb{P}#1{#2}}}

\newcommand{\dd}[1]{\ensuremath{\mathrm{d}{#1}}}

\newtheorem{prop}{Proposition}
\newtheorem{lemma}{Lemma}
\newtheorem{theorem}{Theorem}

\usepackage{makecell}
\newcommand{\e}{\textcolor{black}}

\usepackage[american]{babel}
\usepackage{csquotes}
\usepackage[backend=biber,natbib=true,style=apa]{biblatex}
\title{Testing for sufficient follow-up in survival data \\with {a cure fraction}}
\author[a]{Tsz Pang Yuen\thanks{t.p.yuen@uva.nl}}
\author[a]{Eni Musta\thanks{e.musta@uva.nl}}
\affil[a]{Korteweg-de Vries Institute for Mathematics, University of Amsterdam, Netherlands}
\date{}

\begin{document}
	\maketitle
	\begin{abstract}
		In order to estimate the proportion of `immune' or `cured' subjects who will never experience failure, a sufficiently long follow-up period is required. Several statistical tests have been proposed in the literature for assessing the assumption of sufficient follow-up, meaning that the study duration is longer than the support of the survival times for the uncured subjects. These tests do not perform satisfactorily, especially in terms of Type I error. In addition, they are constructed based on the assumption that the survival time for the uncured subjects has a compact support, i.e. the existence of a `cure time'.
		However, for practical purposes, the assumption of `cure time' is not realistic and the follow-up would be considered sufficiently long if the probability for the event to happen after the end of the study is very small. Based on this observation, we formulate a more relaxed notion of `practically' sufficient follow-up characterized by the quantiles of the distribution and develop  a novel nonparametric  statistical test. The proposed method  relies mainly on the assumption of a  non-increasing density function in the tail of the distribution.
		The test is then based on a shape constrained density estimator such as the Grenander or the kernel smoothed Grenander estimator and a bootstrap procedure is used for computation of the critical values. The performance of the test is investigated through an extensive simulation study, and the method is illustrated on breast cancer data.
	\end{abstract}
	
	\begin{refsection}
\section{Introduction}\label{sec:intro}

Cure models for analysis of time-to-event data in the presence of subjects who will never experience the event of interest has recently attracted increasing attention from both methodological and application perspectives. Survival data with a cure fraction are nowadays frequently encountered in oncology since advances in cancer treatments have resulted in a larger proportion of patients recovering from their illnesses and not experiencing relapse or cancer-related death \citep{LB2019}. Among other fields, cure models have also been applied to analyze time to pregnancy and 
the default time of a loan applicant \citep{VVBHZM2012,ZYS2019}. Independent of the application of interest,
we refer to the subjects
immune to the event of interest as `cured' and to the susceptible ones as `uncured'. For a comprehensive review on cure models, we refer the reader to \textcite{MZ1996}, \textcite{AK2018}, \textcite{PY2022} and \textcite{MRSZ2024}.

In the presence of censoring, it is not possible to distinguish the cured subjects from the censored uncured ones. The presence of a cure fraction is however indicated by a Kaplan--Meier  estimator (KME) of the survival function  \citep{KM1958} that reaches a  plateau at a level greater than zero.
The level of such plateau can be considered as an estimate of the cure fraction, provided that the follow-up is sufficiently long to assume that the survival function would remain constant even after the end of the study. In practice, it is however unclear what is the minimum required follow-up to accurately estimate the cure fraction and a KME with a prolonged plateau, 
containing
many censored observations, is
an indication of sufficient follow-up. Nevertheless, such visual inspection might
be ambiguous and inadequate for assessing sufficient follow-up, leading to an overestimation of the cure fraction. The crucial nature of this assumption underscores the necessity for a reliable statistical test.

The notion of `sufficient follow-up' was first characterized as the setting
where
the support of the event times for the uncured is included in the support of the censoring times in \textcite{MZ1992,MZ1994}. A procedure for testing the \textit{null hypothesis of insufficient follow-up}, based on the length between the maximum observed  time and the maximum uncensored event time, was introduced.
\textcite{MZ1996} proposed an alternative, which relies on simulation to approximate the distribution of the test statistics, by assuming the survival time of the susceptible and the censoring time have an exponential and a uniform distribution, respectively. To circumvent making parametric assumptions, while still ensuring control over the level, \textcite{S2000} introduced an alternative method based on the ratio of the two maximal event times, but the practical behavior of the test is still unsatisfactory. Recently, \textcite{MRS2022} studied the finite sample and the asymptotic distributions of the two maximal censored and uncensored times. Such results were then used in \textcite{MRS2023} to establish the asymptotic distribution of the test statistic in \textcite{MZ1996}. On the other hand, \textcite{XEK2023} considered testing the \textit{null hypothesis of sufficient follow-up} and proposed a new test statistics based on extreme value theory. 

The existing characterization of sufficient follow-up requires that there exists a cutoff or `cure time', such that surviving beyond such time is equivalent of being cured, and it is not possible for an event to happen after the end of the study. This is however not realistic and in practice follow-up would be considered sufficient if the chance of the event happening after the end of the study is negligible. For example, depending on the cancer type, relapse after 10 or 20 years is very rare but not impossible. In such cases, considering patients that survive 10 or 20 years relapse-free  as cured and using the height of the plateau of the  KME  as estimator of the cure proportion would be adequate.
Therefore we propose a more relaxed notion of sufficient follow-up which means that the probability for the event to happen after the end of the study is smaller than a prespecified threshold, e.g. 1\%, chosen by the user. This would be more realistic in practice and still guarantee good identification of the cure proportion.
\textcite{SO2023} also point out that, when the proportion of uncured subjects remaining at the end of the study is very small, cure models are still appropriate. Motivated by this, they propose a parametric procedure to evaluate a similar relaxed notion of sufficient follow-up. In this paper we introduce a nonparametric statistical test based on the following idea.
Under the reasonable assumption that the survival time of the susceptible subjects has a non-increasing density in the tail region, one intuitively expects a small value of the density function at the end point of the support of the observed survival times, given the follow-up is sufficient. Therefore we utilize the Grenander or the kernel smoothed Grenander estimator of the density as our test statistic. The asymptotic properties of the test are studied, and its finite sample performance is investigated through an extensive simulation study. In terms of the level control and test power, the simulation study shows that the proposed method performs better compared to the existing tests for the null hypothesis of insufficient follow-up.

The article is organized as follows. Section~\ref{sec:problem} contains a description of existing formulations of testing sufficient follow-up and a discussion on their differences. In Section~\ref{sec:proposed_method}, we propose a novel notion of `practically' sufficient follow-up and introduce a procedure to test the null hypothesis of insufficient follow-up under
the new notion.
The finite sample performance of the proposed method is investigated through a comprehensive simulation study in Section~\ref{sec:sim}. To illustrate the use of the method, we analyze two breast cancer datasets in Section~\ref{sec:app}.

\section{Problem formulation and discussion}\label{sec:problem}

Suppose $T$ is a nonnegative random variable denoting the event time of a subject with distribution function $F$ and density $f$. Let $C$ be the random right censoring time with distribution function $G$ and assume that $C$ and $T$ are independent. Under the random right censoring assumption, we observe the pair $\rbracket{Y, \Delta}$ where $Y = \min{\rbracket{T, C}}$ is the observed survival time and $\Delta = \indicator{T \leq C}$ is the censoring indicator. Let $H(t)=\prob[\sbracket]{Y\leq t}$ be the distribution of the observed survival time. In the presence of cured subjects in the population, the mixture cure model (MCM) assumes that a subject  is susceptible with probability $p$ and is cured with probability $1-p$. With the convention that $\infty\cdot 0 = 0$, the survival time $T$ can be decomposed as $T=(1-B)\cdot\infty + B\cdot T_{u}$, where $B$ is a latent binary variable  indicating the uncure status of a subject ($B=1$ if the subject is uncured) and $T_{u}$ is the survival time for the uncured subject with a proper distribution function $F_{u}(t)=\prob[\sbracket]{T\leq t\ |\ B=1}$. The subscript `u' will be used to denote quantities that correspond to the uncured subpopulation. Then the distribution function of $T$ under the MCM is given by
$
F(t)=
pF_{u}(t),
$
which is an improper distribution when $p<1$.
Throughout this paper, we assume that $0 < p < 1$, indicating the presence of cured subjects in the population. The testing for $p=1$ against $p<1$ has been explored in the literature \citep{MZ1996,PY2022}.

Let  $\tau_{G}=\sup\set{t\geq0}{G(t)<1}$ and $\tau_{F_{u}}=\sup\set{t\geq0}{F_{u}(t)<1}$ be the right extremes of $G$ and $F_u$, respectively.  Suppose that $n$ i.i.d. realizations $(Y_{i}, \Delta_{i}), i=1,\dots,n$ of $(Y, \Delta)$ are observed.
We can estimate nonparametrically the distribution $F$ by the KME
\citep{KM1958} 
defined as
\[
\hat{F}_n(t)=1-\prod_{i: Y_{(i)}\leq t}\left(1-\frac{\#\{j: Y_j=Y_{(i)}, \Delta_j=1\} }{\#\{j: Y_j\geq Y_{(i)} \}}\right), 
\]
where $Y_{(i)}$ is the $i$-th order statistic.
\textcite{MZ1992} showed that, if $\tau_G\geq \tau_{F_u}$, the uncure fraction $p$ can be estimated consistently by $\hat{p}_n=\hat{F}_n(Y_{(n)})$, where $Y_{(n)}$ is the largest observed survival time. Otherwise, $\hat{p}_n$ would underestimate $p$. This gives rise to the crucial assumption of sufficient follow-up, i.e. $\tau_G\geq \tau_{F_u}$. In practice $\tau_{F_u}$ is not known and a visual inspection of the KME is usually used to assess sufficient follow-up.
Under such assumption one expects to see a Kaplan--Meier estimate with a long plateau, containing many censored observations.
This paper aims to introduce a new statistical test for the sufficient follow-up assumption. First we discuss extensively on the existing approaches mentioned previously.

\subsection{Testing insufficient versus sufficient follow-up}
\textcite{MZ1992,MZ1994} were the first to propose a nonparametric test for testing
\begin{equation}
	\label{eq:insuff_vs_suff}
	H_{0}: \tau_{G} \leq \tau_{F_{u}}
	\quad\text{versus}\quad
	H_{a}: \tau_{G} > \tau_{F_{u}},
\end{equation}
relying on the magnitude of $\tau_{G} - \tau_{F_{u}}$. Let $Y_{(n)}$ be the largest observed survival time and $\tilde{Y}_{(n)}$ be the largest event  time (uncensored survival time). They 
showed that $Y_{(n)} - \tilde{Y}_{(n)}$ converges almost surely to $\tau_{G} - \tau_{F_{u}}$ if $\tau_{F_{u}} \leq \tau_{G}$ and to $0$ if $\tau_{F_{u}} > \tau_{G}$, and suggested rejecting the null hypothesis $H_{0}$ if the $p$-value 
$\mathbb{P}[{Y_{(n)} - \tilde{Y}_{(n)} \geq y_{(n)} - \tilde{y}_{(n)}} ] < \alpha$, where $\alpha$ is a pre-specified significance level,  $y_{(n)}$ and $\tilde{y}_{(n)}$ are the observed values of $Y_{(n)}$ and $\tilde{Y}_{(n)}$, respectively. 
Under the null hypothesis $H_{0}$ and with a large sample size, the previous probability can be approximated reasonably by $(1-q_{n})^{n}$, where $q_{n} = \mathbb{P}[{C \geq {T >  2\tilde{y}_{(n)} - {y}_{(n)}} }]$. Without knowing the distribution functions $F$ and $G$, one can estimate $q_{n}$ using
$\hat{q}_{n} = {
	\sum_{i=1}^{n}\Delta_{i}\indicator{
		2\tilde{y}_{(n)}-y_{(n)} < Y_{i} \leq \tilde{y}_{(n)}
	}
}/{n}$.
Then $H_0$ is rejected if the estimated $p$-value $\alpha_{n}=\rbracket{1 - \hat{q}_{n}}^{n}$ is smaller than the nominal level.
We refer to this test procedure as the $\alpha_{n}$ test.
\textcite{S2000} introduced an alternative, called the $\tilde{\alpha}_{n}$ test, which uses a test statistic that based on the ratio $Y_{(n)} / \tilde{Y}_{(n)}$ instead of the difference $Y_{(n)}-\tilde{Y}_{(n)}$. The Type I error and power of the test were studied and its finite sample performance was investigated through simulations, with comparisons made to the $\alpha_{n}$ test.

As mentioned in \citet[Page 85]{MZ1996}, it is more appropriate to compare $\alpha_n$ with a nominated quantile of its finite sample or its asymptotic distribution instead of the nominal level.
\textcite{MRS2022} studied both finite sample and asymptotic distributions for the largest observed survival time $Y_{(n)}$ and the largest  event time $\tilde{Y}_{(n)}$, and emphasized {that such results can be used to  determine} the critical values of the previous tests. Based on the established results, \textcite{MRS2023} further derived the finite sample and asymptotic distributions of the test statistics $n\hat{q}_{n}$, where $\hat{q}_{n}$ is defined previously.
The asymptotic distribution of $n\hat{q}_n$ when the follow-up is sufficient/insufficient is obtained under some tail behavior assumptions on $F_u$ and $G$ given in the Supplementary Material. The asymptotic distribution can be used to determine the critical value for testing $H_{0}$.
We refer to such test procedure as the $Q_{n}$ test. However, simulation studies show that the empirical level of the $Q_{n}$ test is still larger than the nominal level and deviates more when the follow-up time increases \citep{XEK2023}. Nevertheless, the $Q_n$ test has the advantage of its simplicity and its intuitive idea.

\subsection{Testing sufficient versus insufficient follow-up}
\textcite{XEK2023} considered testing the null hypothesis of sufficient follow-up
\begin{equation}
	\label{eq:suff_vs_insuff}
	\check{H}_{0}: \tau_{F_{u}} \leq \tau_{G}
	\quad\text{versus}\quad
	\check{H}_{a}: \tau_{F_{u}} > \tau_{G}.
\end{equation}
The authors proposed a test statistic $T_{n}$ given by the difference between an estimator of $p$ computed as if the follow-up was sufficient $\hat{F}_{n}(y_{(n)})$ and an estimator computed via extrapolation of $\hat{F}_{n}$ beyond $y_{(n)}$ based on extreme value theory. The extrapolation corrects the underestimation of $\hat{F}_{n}(y_{(n)})$ for $p$ under the insufficient follow-up setting. The asymptotic normality of
$T_{n}$ was established and a bootstrap procedure was introduced,
resulting in better approximation to the critical values of the test compared to the asymptotic results. The required assumptions for this test are given in the Supplementary Material.

\subsection{Flipping the hypotheses or not?}
The difference between the two formulations of the hypotheses discussed previously is the type of error that one aims to control in practice. For the $\alpha_{n}$, $\tilde{\alpha}_{n}$ and $Q_{n}$ tests one aims to control the probability of concluding that the follow-up is sufficient when it is actually not, while $T_n$ controls the probability of deciding that the follow-up is insufficient when it is actually sufficient.
From the practical point of view, if based on the plateau of the KME and medical knowledge, one expects that the follow-up is sufficient, the $T_n$ test for the null hypothesis of sufficient follow-up can be used. In this case sufficient follow-up would be rejected if there is evidence against it and otherwise one can proceed with estimation assuming sufficient follow-up. On the other hand, if one is uncertain about the sufficient follow-up assumption based on KME and medical knowledge, it is safer to test the null hypothesis  of insufficient follow-up and conclude sufficient follow-up only if there is strong evidence in favor of it. This would be a safer choice since the consequences of assuming sufficient follow-up when it is actually not true are more serious than vice versa because it would lead to wrong conclusions regarding the cure fraction.

\section{Testing procedure}\label{sec:proposed_method}

In this paper we focus on testing the null hypothesis of insufficient follow-up. However our goal is to consider a more relaxed formulation of sufficient follow-up compared to the one in \eqref{eq:insuff_vs_suff}, which would be more realistic and satisfactory in practice.  Note that in practice one always has $\tau_G<\infty$ because of the finite length of studies and as a result the characterization in \eqref{eq:insuff_vs_suff} would mean that the follow-up is sufficient only if there exists a finite cure time $\tau_{F_u}$ such that it is impossible for an event to happen after $\tau_{F_u}$ and the study is longer than that. This is hardly ever the case in practical applications. For example, it is known that cancer relapse after 5, 10 or 20 years (depending on the cancer type) is very rare but yet not impossible and clinical trials are usually not very long to have the certainty that no events will happen after the end of the follow-up period. Hence, in practice it would be 
more meaningful to consider follow-up as sufficient when $F_{u}(\tau_{G}) > 1 - \epsilon$ for some small $\epsilon > 0$ {(e.g. $\epsilon=0.01$)}, or equivalently, when $\tau_{G} > q_{1 - \epsilon}$, with $q_{1 - \epsilon} = \inf\set{t \geq 0}{F_{u}(t) \geq 1 - \epsilon}$. Specifically, we introduce the following hypotheses:
\begin{equation}
	\label{eq:insuff_vs_suff_new}
	\tilde{H}_{0}: q_{1 - \epsilon} \geq \tau_{G}
	\quad\text{versus}\quad
	\tilde{H}_{a}: q_{1 - \epsilon} < \tau_{G}.
\end{equation}
Note that under $\tilde{H}_a$, we have $F(\tau_G)=pF_u(\tau_G)\in[p-\epsilon p,p]$. Hence, the amount of underestimation of $p$, resulting from using $\hat{F}_n(\tau_{G})$ as an estimator, is very small (less than $\epsilon$).
Such underestimation of $p$ is small relative to the estimation uncertainty for finite sample sizes and therefore negligible for practical purposes.
\e{The hypotheses considered here are referred to as practically sufficient or insufficient follow-up, depending on the choice of $\epsilon$. In practice, using $\epsilon=0.01$ is reasonable. However, we emphasize that this is not a universal choice and the practically sufficient follow-up does not necessarily imply the traditional sufficient follow-up defined in \eqref{eq:insuff_vs_suff}.}

\subsection{{Idea of the test}}
\label{seq:idea_test}
Our test relies on the assumption that the density $f_{u}$ of $F_{u}$ is non-increasing  in the tail region and at least continuous. Then the hypothesis in \eqref{eq:insuff_vs_suff_new} can equivalently be written as 
\begin{equation}
	\label{eq:insuff_vs_suff_new*}
	\tilde{H}_{0}:  f(\tau_{G})\geq f(q_{1 - \epsilon}) 
	\quad\text{versus}\quad
	\tilde{H}_{a}:  f(\tau_{G})<f(q_{1 - \epsilon}),
\end{equation}
since $f(t)=pf_u(t)$. Then the idea is to estimate $f(\tau_G)$ from the sample with standard methods and to reject $\tilde{H}_0$ if such estimator is smaller than a critical value. Note however that $f(q_{1 - \epsilon})$ is unknown and we instead try to find a lower bound for that, which would guarantee the level of the test in the worse case scenario.  The reasoning is as follows. Let $\eta>0$ be a very small number compared to $\epsilon$ and $\tau > q_{1 - \epsilon}$ be such that $\prob{T_{u}>\tau} < \eta$ or essentially $\prob{T_{u}>\tau}$ is negligible. By the monotonicity and the continuity of $f_{u}$, we have
\begin{equation*}
	\label{eq:eps_f_u}
	\epsilon = \int_{q_{1 - \epsilon}}^{\tau_{F_{u}}}f_{u}(t)\dd{t}
	\leq \int_{q_{1 - \epsilon}}^{\tau}f_{u}(t)\dd{t} + \eta
	\leq f_{u}(q_{1 - \epsilon})(\tau - q_{1 - \epsilon}) + \eta,
\end{equation*} 
meaning that under $\tilde{H}_0$
\begin{equation*}
	\label{eq:h0_f_u_impl}
	f(\tau_{G}) \geq f(q_{1 - \epsilon}) 
	\geq \frac{(\epsilon - \eta)p}{\tau - \tau_{G}}\approx \frac{\epsilon p}{\tau - \tau_{G}}.
\end{equation*} Since $p\geq F(\tau_G)$, the idea is to reject $\tilde{H}_0$ if an estimator $\hat{f}_n(\tau_G)$ is smaller than $\epsilon \hat{F}_n(\tau_{G})/(\tau-\tau_G)+\delta_n$, where $\delta_n$ is determined based on the desired level and the distribution of the estimator $\hat{f}_n(\tau_G)$. More details are given below for two specific density estimators.
The choice of $\tau$ can be based on some prior knowledge, for example, one can take $\tau>\tau_G$ such that it is almost impossible for the event to happen after $\tau$, i.e. if one had a follow-up of length $\tau$ one would consider it sufficient. As we illustrate in the simulation study and real data application below, one can also do a sensitivity analysis with respect to $\tau$ and, if in doubt, we suggest taking a larger $\tau$ for a more conservative test. We note that $\tau$ is not data dependent. Therefore, the asymptotic results are derived for the fixed value of $\tau$ later in this section. For simplicity, we assume that $\tau_G$ is known, for example the study duration in case of administrative censoring. However, all the results can be extended to the case where $\tau_G$ is unknown and
estimated by $Y_{(n) }= \max_i Y_i$. Below we describe two possible approaches of estimation of $f(\tau_G)$, leading to two different statistical tests. For simplicity we consider a non-increasing density on the whole support but the results hold more in general. In practice one would just need to restrict the shape constrained estimation on a subset of the support based on visual inspection of the KME  ($F$ is concave where $f$ is non-increasing) or a statistical test for monotonicity.

\subsection{Grenander estimator of $f$}
A natural nonparametric estimator of $f$ under monotonicity constraints  would be the Grenander estimator \citep{G1956}
denoted by $\hat{f}_{n}^{G}$, which is defined as the left derivative of the least concave majorant (LCM)  of the
KME.
We refer the reader to \textcite{GJ2014} for an overview of statistical inference techniques under shape-constraints. We are interested in estimating the density function at the boundary of the support of the observed data $\tau_G$ but it is well-known that $\hat{f}_{n}^{G}$ is not consistent at the boundaries \citep{WS1993}. For the standard setting without censoring, \textcite{KL2006} investigated the behavior of $\hat{f}_{n}^{G}$ near the boundary of the support of $f$ and proposed using $\hat{f}_{n}^{G}(\tau_{G} - cn^{-a})$ as a consistent estimator of $f(\tau_{G})$. Here $c > 0$ is a constant, and $0 <  a < 1$ controls the rate of convergence and the limiting distribution of the estimator. 

In the presence of random right censoring one can show as in \textcite{KL2006} that the following holds. Here $W(t)$ denotes a Brownian motion and $D_R[Z(t)](a)$ is the right derivative of the LCM on $[0,\infty)$ of the process $Z(t)$ at the point $t=a$. 
\begin{theorem}
	\label{theo:grenander}
	Assume that $\tau_G<\tau_{F_{u}}$, and in addition $f_u$ is nonincreasing and differentiable with bounded derivative on $[0,\tau_G]$, $\tau_G<\infty$. If $f(\tau_G)>0$, $|f^\prime(\tau_G)|>0$, $G(\tau_G-)<1$ and $c>0$ is a fixed constant, we have
	\begin{itemize}[leftmargin=*]
		\item[i)] for $a\in(1/3,1)$, the sequence
		$
		A_1n^{(1-a)/2}\{f(\tau_G)-\hat{f}_n^G(\tau_G-cn^{-a})\}
		$
		converges in distribution to $D_R[W(t)](1)$ as $n\to \infty$, where $A_1=\sqrt{c[1-G(\tau_G-)]/f(\tau_G)}$;
		\item[ii)] the sequence
		$
		A_2n^{1/3}\{f(\tau_G)-\hat{f}_n^G(\tau_G-cB_2n^{-1/3})\}
		$
		converges in distribution to $D_R[W(t)-t^2](c)$ as $n\to \infty$, where 
	\end{itemize}
	\[
	{B_2}=\left(\frac{2\sqrt{f(\tau_G)}}{\left|f'(\tau_G)\right|\sqrt{1-G(\tau_{G}-)}}\right)^{\frac{2}{3}}\text{,}~ {A_2}=\sqrt{\frac{B_2[1-G(\tau_G-)]}{f(\tau_G)}}.
	\] 
\end{theorem}
Note that the assumption $G(\tau_G-)<1$ indicates that the censoring distribution has a positive mass at $\tau_G$ and is required in order to have a strong approximation of the KME by a Brownian motion appearing in the limiting distribution. As this mass becomes smaller, the constants $A_2$ and $B_2$ get closer to zero and infinity respectively. 
For such assumption, one can think of all subjects enter the study at the same time and consider the censoring variable as $C=\min(\tilde{C},\tau_G)$, where $\tau_G$ represents the study duration (administrative censoring), while $\tilde{C}$ is the censoring time due to loss of follow-up for other reasons and it would take values larger than $\tau_G$ if the study continued for longer.
The assumption can be relaxed, for example by applying the result in \textcite{Y1989} to obtain a weak convergence of the KME to a Gaussian process under the assumption $\int_{0}^{\tau_G}\frac{\dd{F(t)}}{1-G(t-)}<\infty$ or applying the result in \textcite{CL1997} to obtain a strong approximation of the KME under the assumption  $\int_{0}^{\tau_G}\frac{\dd{F(t)}}{(1-G(t-))^{p/(1-p)}}<\infty$ for some $p\in(0,1/2)$. Nevertheless, we chose to maintain this assumption, which is simpler, and instead investigate the performance of the test also in scenarios where it is not satisfied.
The previous theorem shows that the $\hat{f}_n^G(\tau_G-cB_2n^{-1/3})$ converges at rate $n^{1/3}$ to $f(\tau_G)$ but computation of $B_2$ requires also estimation of the derivative of the density at $\tau_G$. To avoid that, for the test statistics, we will instead use an estimator $\hat{f}_n^G(\tau_G-cn^{-a})$ for some $a$ larger than but close to $1/3$ as in  statement  (i)  of the theorem. Such choice, despite being sub-optimal in terms of rate of convergence, behaves better in practice.
The unknown quantities in $A_1$, $A_2$ and $B_2$ can be estimated using their corresponding consistent estimators. For example, $f(\tau_G)$ in $A_1$ can be estimated by $\hat{f}_n^G(\tau_G-cn^{-a})$, which is consistent for $f(\tau_G)$, as shown in \textcite{KL2006} for the estimator in the setting without censoring. Also the censoring distribution $G$ can be estimated using the KME.

\subsection{Smoothed Grenander estimator}
Alternatively, we can consider the kernel smoothed Grenander-type estimator.
Let $k$ be a symmetric twice continuously differentiable kernel with support $[-1, 1]$ such that $\int k(u)\dd{u} = 1$ and its first derivative is bounded uniformly.
To circumvent the inconsistency issues of the standard kernel density estimator at the boundary,
we utilize boundary kernels \citep{DGL2013}.
Denoting the boundary kernel as $k_{B,t}(v)$, which is defined by
\[
k_{B,t}(v)=
\begin{cases}
	\phi\rbracket{\frac{t}{h}}k(v) + \psi\rbracket{\frac{t}{h}}vk(v)&\quad t\in[0,h],\\
	k(v)&\quad t\in(h,\tau_{G}-h),\\
	\phi\rbracket{\frac{\tau_G - t}{h}}k(v) - \psi\rbracket{\frac{\tau_G-t}{h}}vk(v)&\quad t\in[\tau_{G}-h,\tau_{G}],
\end{cases}
\]
for $v\in[-1,1]$, where $h$ is the bandwidth parameter. Here, for $s\in[-1,1]$, the coefficients $\phi(s)$ and $\psi(s)$ are determined by
\begin{align*}
	\begin{split}
		\phi(s)\int_{-1}^{s}k(v)\dd{v} &+ \psi(s)\int_{-1}^{s}vk(v)\dd{v}=1,\\
		\phi(s)\int_{-1}^{s}vk(v)\dd{v} &+ \psi(s)\int_{-1}^{s}v^{2}k(v)\dd{v}=0.
	\end{split}
\end{align*}
The smoothed Grenander-type estimator with boundary correction of $f(\tau_{G})$ is given by:
\begin{equation*}
	\hat{f}_{nh}^{SG}(\tau_{G}) = 
	\int_{\tau_{G} - h}^{\tau_{G}}
	\frac{1}{h}k_{B,\tau_{G}}\rbracket{
		\frac{\tau_{G} - v}{h}
	}\hat{f}_{n}^{G}(v)
	\dd{v}.
\end{equation*}
Such estimator and its asymptotic distribution in the interior of the support has been studied in \textcite{LM2017}. Here we extend such results to the estimation at the right boundary $\tau_G$.
\begin{theorem}
	\label{thm:sg_normality}
	Assume that $\tau_G<\tau_{F_u}$, and in addition $f_{u}$ is nonincreasing and twice continuously differentiable with $f_{u}$ and $|f_{u}^{\prime}|$ being strictly positive on $[0,\tau_G]$. If $G(\tau_G-)<1$ and $hn^{1/5}\converge c\in(0, \infty)$ as $n\converge\infty$, then
	\begin{equation*}
		n^{2 / 5}\lbrace{
			\hat{f}_{nh}^{SG}(\tau_{G}) - f(\tau_{G}) 
		}\rbrace
		\converge[d]
		N(\mu, \sigma^2),
	\end{equation*}
	where \[\mu = \frac{1}{2}c^{2}f^{\prime\prime}(\tau_{G})\int_{0}^{1}v^{2}k_{B,\tau_{G}}(v)\dd{v},~
	\sigma^2 = \frac{f(\tau_{G})}{c\sbracket{1 - G(\tau_{G}-)}}\int_{0}^{1}k_{B,\tau_{G}}(v)^{2}\dd{v}.\]
\end{theorem}
The assumption $\tau_G<\tau_{F_u}$ of Theorem~\ref{thm:sg_normality} is the situation of most interest when testing $\tilde{H}_0: q_{1-\epsilon}\geq\tau_G$. We note that the rate of convergence of the smooth Grenander estimator is faster than the non-smoothed one and the limiting variance increases as the mass
of the censoring distribution at $\tau_G$ 
decreases. The limiting bias depends on the second derivative of the density which is difficult to estimate so in practice we will use
bootstrapping
as described later. 

\subsection{Test procedure based on asymptotic results}

Based on the idea described in Section~\ref{seq:idea_test} and the two estimators of a decreasing density, we propose two statistical tests for the assumption of `practically' sufficient follow-up formulated in \eqref{eq:insuff_vs_suff_new}. The first test rejects the null hypothesis of insufficient follow-up $\tilde{H}_0$ if 
\begin{equation}
	\label{eqn:test_grenander}
	\hat{f}_{n}^{G}(\tau_{G} - cn^{-a}) \leq \frac{\epsilon\hat{F}_{n}(\tau_{G})}{\tau - \tau_{G}} - A_{1}^{-1}n^{-(1-a)/2}Q_{1 - \alpha}^{G},
\end{equation} 
where $a\in(1/3,1)$ (we choose $a=0.34$ in the simulation study, which is very close to 1/3), $c=\tau_G$ (so that the amount of deviation from the right boundary scales with the length of the support), $A_1$ is defined in Theorem~\ref{theo:grenander}  and is in practice replaced by a plug-in estimator $\hat{A}_1$, with $f(\tau_G)$ and $G$ estimated by $\hat{f}_{n}^{G}(\tau_{G} - cn^{-a})$ and the KME for $G$, respectively. $Q_{ \alpha}^{G}$ is the $\alpha$-quantile of the distribution $D_R[W(t)](1)$, and it is approximated through computer simulation.
Below we show that the level of the test is asymptotically bounded by $\alpha$. 
\begin{prop}
	\label{prop:level_grenander}
	Suppose $0<\eta<\epsilon$ is such that $n^{(1-a)/2}\eta\to 0$ as $n\to\infty$. Then, for any $\tau_G\leq q_{1-\epsilon}$, 
	\[
	\lim_{n\to\infty}\prob[\sbracket]{\hat{f}_{n}^{G}(\tau_{G} - cn^{-a}) \leq \frac{\epsilon\hat{F}_{n}(\tau_{G})}{\tau - \tau_{G}} - A_{1}^{-1}n^{-\frac{1-a}{2}}Q_{1 - \alpha}^{G}}\leq\alpha.
	\]
\end{prop}
The advantage of this test is that it does not require smoothing and twice differentiability of the true density function. However, it requires a plug-in estimate for $A_1$ and a choice of the constant $c$. As noted in the proof of the proposition, the asymptotic level of the test remains bounded by $\alpha$ as long as we use a consistent estimate for $A_1$. In practice we observe that very large sample sizes would be required to have a good performance of this test.

The second test we propose rejects the null hypothesis of insufficient follow-up $\tilde{H}_0$ if 
\begin{equation}
	\label{eqn:test_SG}
	\hat{f}_{nh}^{SG}(\tau_{G}) \leq \frac{\epsilon \hat{F}_{n}(\tau_{G})}{\tau - \tau_{G}} + n^{-2/5}Q_{\alpha}^{SG},
\end{equation}
where $Q_{\alpha}^{SG}$ denotes the $\alpha$-quantile of the distribution $N(\mu,\sigma^2)$ in Theorem~\ref{thm:sg_normality}. Such test is also guaranteed to be asymptotically correct.
\begin{prop}
	\label{prop:level_SG}
	Suppose $0<\eta<\epsilon$ is such that $n^{2/5}\eta\to 0$ as $n\to\infty$. Then, for any $\tau_G\leq q_{1-\epsilon}$, 
	\[
	\lim_{n\to\infty}\prob[\sbracket]{\hat{f}_{nh}^{SG}(\tau_{G}) \leq \frac{\epsilon\hat{F}_{n}(\tau_{G})}{\tau - \tau_{G}} +n^{-2/5}Q_{ \alpha}^{SG}}\leq\alpha.
	\]
\end{prop}		
Computation of $Q_{ \alpha}^{SG}$ requires plug-in estimates for the asymptotic mean and variance of the estimator, which do not behave well for small sample sizes. For this reason, we suggest a bootstrap procedure for approximating the critical value of the test as described in the next subsection. We note again that the asymptotic level of the test remains bounded by $\alpha$ if the unknown quantile $Q_{ \alpha}^{SG}$ is replaced by a consistent estimator. 

Recall that we consider $\tau_G$ as known and equal to the study duration under administrative censoring. However, all the results can be extended to the case where $\tau_G$ is unknown and estimated by $Y_{(n)}=\max_iY_i$, which is a consistent estimator for $\tau_G$ if $\tau_G\leq\tau_{F_u}$. We establish the asymptotic normality of the estimator $\hat{f}_{nh}^{SG}(Y_{(n)})$ and show the asymptotic level of the test is bounded by $\alpha$, under the conditional setting with categorical covariate information, in a subsequent work \citep{YMVK2024}.

\subsection{Bootstrap procedure}
\label{sec:bootstrap}
It is worthy to mention that the naive bootstrap, by resampling with replacement from the pairs $(y_{i}, \delta_{i})$ and computing the Grenander estimator $\hat{f}_{n}^{G}(\tau_{G})$ from the bootstrap samples, is inconsistent, since the true density of the bootstrap sample is not continuous \citep{K2008}. Therefore a smoothed bootstrap procedure, based on the smoothed Grenander estimator $\hat{f}_{n}^{SG}$, is introduced to approximate the critical value of the test
(see Algorithm \ref*{algo:bootstrap_algo} in the Supplementary Material).
The idea of the smoothed bootstrap procedure is that a modified kernel density estimator $\tilde{f}_{nh_{0}}$ is estimated from the sample data, followed by drawing bootstrap samples from $\tilde{f}_{nh_{0}}$. Specifically, $\tilde{f}_{nh_{0}}(t)$ is the derivative of the kernel smoothed least concave majorant $\hat{F}^G_n$ of the KME, i.e. the derivative of $\int\frac{1}{h_{0}}k_{B,t}({\frac{t-u}{h_{0}}})\hat{F}_{n}^{G}(u)\dd{u}$. If one would use the bandwidth $h$ of order $n^{-2/5}$ as in the computation of $\hat{f}_{nh}^{SG}$, an explicit estimation of $f^{\prime\prime}(\tau_{G})$ would be required to adjust the asymptotic bias of the bootstrap estimates \citep{GH2018}. Hence obtaining an accurate estimate for $f^{\prime\prime}(\tau_{G})$ would be crucial
but it is problematic especially near the end point $\tau_{G}$. Instead, as proposed by \textcite{GJ2023} in the monotone regression setting, we draw bootstrap samples from $\tilde{f}_{nh_{0}}$ with $h_{0} = \tau_{G}(0.7n^{-1/9} \wedge 0.5)$,
which leads to an over-smoothed estimate. Oversmoothing circumvents the need to correct for the asymptotic bias. Then we replace $Q^{SG}_\alpha$ in the critical value of the test in \eqref{eqn:test_SG} by the $\alpha$-quantile of the bootstrap estimates $n^{2/5}\{\hat{f}_{nh}^{SG^{\ast}}(\tau_{G})-\tilde{f}_{nh_{0}}(\tau_{G})\}$.

\section{Simulation study}\label{sec:sim}

In this section, we study the finite sample performance of the testing procedure (using $\hat{f}_{n}^{G}$ and $\hat{f}_{nh}^{SG}$) described in Section~\ref{sec:proposed_method} for testing $\tilde{H}_{0}: q_{1-\epsilon} \leq \tau_{G}$ and compare it with the $\alpha_{n}$, $\tilde{\alpha}_{n}$ and $Q_{n}$ tests.
In order to cover various scenarios with different uncured survival time and censoring time distributions, we consider six settings as described below. In Settings \ref*{enum:sim_exp_unif}--\ref*{enum:sim_wb_unif} the survival time of the uncured has unbounded support, i.e. $\tau_{F_{u}}=\infty$, while that of the censoring time is bounded, i.e. $\tau_{G}<\infty$. Such a scenario is the main focus of our methodology but it means that the follow-up is never sufficient based on the characterization of sufficient follow-up in \eqref{eq:insuff_vs_suff} used by the existing tests. Settings \ref*{enum:sim_texp_unif}--\ref*{enum:sim_texp_exp_trunc} correspond to having $\tau_{F_{u}}<\infty$ and $\tau_{G}<\infty$, indicating that the follow-up is sufficient, under the notion in \eqref{eq:insuff_vs_suff}, when $\tau_{F_{u}}<\tau_{G}$. Setting \ref*{enum:sim_exp_exp} corresponds to having $\tau_{F_{u}}=\infty$ and $\tau_{G}=\infty$. Such case was considered by \textcite{XEK2023} and is regarded as sufficient follow-up under the 
notion
in \eqref{eq:suff_vs_insuff}. The $T_n$ test of \textcite{XEK2023} considers a null hypothesis of sufficient follow-up and as a result cannot be compared with ours. However, for this last setting we also present the results of the $T_n$ test. In addition to the above settings, we include a small simulation with a non-monotone density for the uncured survival time to study the performance of the proposed method when the montonicity assumption of $f_u$ is not satisfied. 
A comparison with the parametric method RECeUS proposed by \textcite{SO2023} is also made through a small simulation under Setting~\ref*{enum:sim_exp_unif}.

\subsection{Simulation settings}
For Settings \ref*{enum:sim_exp_unif}--\ref*{enum:sim_wb_unif} and \ref*{enum:sim_exp_exp}, three different uncured fractions $p$ are considered, namely 0.2, 0.6 and 0.8, to study the effect of $p$ on the testing performance. Note that the censoring rate is at least $1-p$ since all the cured subjects are censored. For Settings \ref*{enum:sim_texp_unif} and \ref*{enum:sim_texp_exp_trunc}, only $p=0.6$ is considered. We consider distribution that has a decreasing density for the uncured event time $T_{u}$. We generate the censoring time $C$ as $C=\min(\tilde{C},\tau_G)$, where the support of $\tilde{C}$ includes $[0,\tau_G]$.
In Section~\ref{sec:proposed_method}, the asymptotic properties of both the Grenander $\hat{f}_{n}^{G}$ and the smooth Grenander $\hat{f}_{nh}^{SG}$ estimators are affected by the mass of the censoring distribution at $\tau_{G}$, $\Delta G(\tau_{G})=1-G(\tau_{G}-)$. To investigate the effect of $\Delta G(\tau_{G})$ on the test performance, we consider $\tilde{C}$ with a uniform distribution on $[0,\zeta]$ for $\zeta\geq\tau_{G}$ chosen in such a way that $\Delta G(\tau_{G})\in\cbracket{0,0.02,0.05,0.2}$, in Settings \ref*{enum:sim_exp_unif}, \ref*{enum:sim_wb_unif} and \ref*{enum:sim_texp_unif}. For such settings, we have that the higher the mass $\Delta G(\tau_{G})$, the lower the censoring rate. For Settings \ref*{enum:sim_exp_unif}--\ref*{enum:sim_texp_exp_trunc}, 12 different $\tau_{G}$'s are considered in order to 
study the level and power
of the test. We choose 
$\tau_G$
in terms of quantiles of $F_u$ because, in practice, the probability of the event to happen after the end of the study is a more important measure of deviation from sufficient follow-up compared to the value of $\tau_G$. Specifically, we choose $q_1,q_2,q_3,q_4, q_6,q_{12}$ to be the 90, 92.5, 95, 97.5, 99 and 99.9\% quantiles of $F_{u}$. In addition we consider $q_5$ to be the mid-point between the 97.5 and 99\% quantiles of $F_{u}$ (i.e. $q_5=(q_4+q_6)/2$) and $q_7-q_{11}$ to be five evenly separated points between the 99 and 99.9\% quantiles of $F_{u}$ base on the event time scale (e.g. $q_7=q_6+(q_{12}-q_6)/6$). 
Table \ref*{supp_tab:actual_tauc} in the Supplementary Material reports their 
values for Settings \ref*{enum:sim_exp_unif}--\ref*{enum:sim_texp_exp_trunc}.
\begin{enumerate}[label={\textit{Setting \arabic*.}},align=left,wide,ref=\arabic*]
	\item\label{enum:sim_exp_unif} The uncured subjects have an exponential distribution with rate parameter~1. The censoring times are generated using the aforementioned uniform distribution for $\tilde{C}$.
	
	\item\label{enum:sim_exp_exp_trunc} The uncured have an exponential distribution with rate parameter $\lambda\in\{0.4,1,5\}$. The censoring times are generated using $\tilde{C}$ having exponential distribution with rate parameter 0.5. In this case 
	the higher the rate parameter $\lambda$, the lower the censoring rate.
	
	\item\label{enum:sim_wb_unif} The uncured subjects have a Weibull distribution with shape and scale parameters 0.5 and 1.5, respectively. Such distribution for the uncured has a density function that decreases faster comparing to the exponential distribution in Setting \ref*{enum:sim_exp_unif}. The censoring times are generated using the aforementioned uniform distribution for $\tilde{C}$.
	
	\item\label{enum:sim_texp_unif} The uncured have a truncated exponential distribution with rate parameter~1 and endpoint $\tau_{F_{u}}$, where $\tau_{F_{u}}$ is the 99\% quantile of the exponential distribution with rate parameter 1, i.e. $\tau_{F_{u}}\approx4.6$. The censoring times are generated using the aforementioned uniform distribution for $\tilde{C}$. To study the testing performance when $\tau_{F_{u}}\leq\tau_{G}$, we consider few extra $\tau_{G}$'s, which are greater than $\tau_{F_{u}}$, in addition to the $\tau_{G}$'s mentioned previously.
	
	\item\label{enum:sim_texp_exp_trunc} The uncured have a truncated exponential distribution with rate parameter 5 and endpoint $\tau_{F_{u}}$, where $\tau_{F_{u}}$ is the 99\% quantile of the exponential distribution with rate parameter 5, i.e. $\tau_{F_{u}}\approx 0.92$. The censoring times are generated using $\tilde{C}$ having an exponential distribution with rate parameter $\lambda_{C}\in\{0.5,3\}$. In this case the higher the rate parameter for the censoring distribution $\lambda_{C}$, the higher the censoring rate. As in Setting \ref*{enum:sim_texp_unif}, we include extra $\tau_{G}$'s
	to investigate the testing performance when $\tau_{F_{u}}\leq\tau_{G}$.
	
	\item\label{enum:sim_exp_exp} The uncured have an exponential distribution with rate parameter 1. The censoring time follows an exponential distribution with rate parameter 0.5.
	\item\label{enum:sim_lnorm_unif} The uncured have a mixture distribution of two log-normal distributions with weights 0.7 and 0.3, denoted by $F_u = 0.7F_{u1}+0.3F_{u2}$, where the parameter of $F_{u1}$ is $(\mu_1,\sigma_1)=(0,1)$ and that of $F_{u2}$ is $(\mu_2,\sigma_2)=(\log8,0.3)$. The censoring time follows the uniform distribution as in Setting~1 with $\tau_G=6.5$ and $\Delta G(\tau_G-)=0.02$. The uncured fraction $p$ is 0.6.
\end{enumerate}
The range of the censoring rates over different values of $\tau_G$ for all the considered settings are reported in Tables \ref*{supp_tab:exp_unif_cens_rate}--\ref*{supp_tab:texp_exp_trunc_cens_rate} in the Supplementary Material. The censoring rate depends on $p$, $\Delta G(\tau_{G})$ in Settings~\ref*{enum:sim_exp_unif}, \ref*{enum:sim_wb_unif} and \ref*{enum:sim_texp_unif}, $\lambda$ in Setting~\ref*{enum:sim_exp_exp_trunc}, and $\lambda_{C}$ in Setting~\ref*{enum:sim_texp_exp_trunc} as mentioned previously.

For the proposed method, $\epsilon$ is set at 0.01, meaning that we consider the follow-up as sufficient when $q_{0.99}<\tau_{G}$. In Setting~\ref*{enum:sim_wb_unif}, we also investigate how the method performs when $\epsilon= 0.005$, resulting in a more conservative hypothesis of `practically' sufficient follow-up.  The parameter $\tau$ is set to the 99.95\% quantile of 
$F_{u}$ for Settings~\ref*{enum:sim_exp_unif}--{\ref*{enum:sim_texp_exp_trunc}. A different choice is made for Setting~\ref*{enum:sim_exp_exp} since $\tau_G=\infty$.} We also studied the effect of the choice of $\tau$ on the rejection rate, using Settings~\ref*{enum:sim_exp_unif} and \ref*{enum:sim_texp_unif}. The parameters $a$ and $c$ of 
$\hat{f}_{n}^{G}$
are set to 0.34 and $\tau_{G}$ ($y_{(n)}$ for Setting~\ref*{enum:sim_exp_exp}), respectively. For
$\hat{f}_{nh}^{SG}$,
the tri-weight kernel is used together with a bandwidth of 
$h=\tau_{G}(n^{-1/5}\wedge0.5)$
($\tau_G$ is replaced by $y_{(n)}$ in Setting~\ref*{enum:sim_exp_exp}). We use 500 bootstrap samples to compute the critical value of the test based on $\hat{f}_{nh}^{SG}$.
The parameter $\gamma$
for $Q_{n}$ test is set to 1, which covers 
various
censoring distributions including those in the current simulation settings. We use the same setting as described in 
\textcite{XEK2023} for $T_{n}$ test with 500 bootstrap samples. The significance level of the tests is set at $0.05$.

\subsection{Simulation results}
For each simulation setting, three sample sizes are considered, namely 250, 500 and 1000, and for each of them we consider 500 replications. For Settings~\ref*{enum:sim_exp_unif}--\ref*{enum:sim_texp_exp_trunc}, we plot the rejection rate of insufficient follow-up against $\tau_{G}$ for
different methods.
For
Setting~\ref*{enum:sim_exp_exp}, which corresponds to only one choice $\tau_G=\infty$, the rejection rate is reported in a table rather than a plot.
In this section we only  present the results for Settings~\ref*{enum:sim_exp_unif} and \ref*{enum:sim_texp_unif} and compare the two proposed tests to the $Q_{n}$~test, which is the best of the existing tests for the null hypothesis of insufficient follow-up. The rest of the results and comparison also with the other methods can be found in the Supplementary Material.
\e{Recall that the hypotheses of practically sufficient or insufficient follow-up in \eqref{eq:insuff_vs_suff_new} depend on the choice of $\epsilon$, thus both the level and power also depend on $\epsilon$. In contrast, the comparator tests assess the null hypothesis of the traditional insufficient follow-up $H_0$, making direct comparison less straightforward. For the consistency of comparison, we fix $\epsilon$ and evaluate the level and power of all tests for testing $\tilde{H}_0$ instead of $H_0$. We note that, in all settings except Settings~\ref*{enum:sim_texp_unif}, \ref*{enum:sim_texp_exp_trunc} and \ref*{enum:sim_exp_exp}, follow-up is insufficient under the traditional notion in \eqref{eq:insuff_vs_suff}, since $\tau_{F_u}=\infty$ in such settings. From the perspective of testing $H_0$, the rejection rates of the comparator tests would be interpreted as measures of the level but not the power, and we would expect them are bounded by the nominal significance level.}

\begin{figure}%
	\centering
	\begin{subfigure}[b]{0.325\textwidth}
		\centering
		\includegraphics[width=\textwidth]{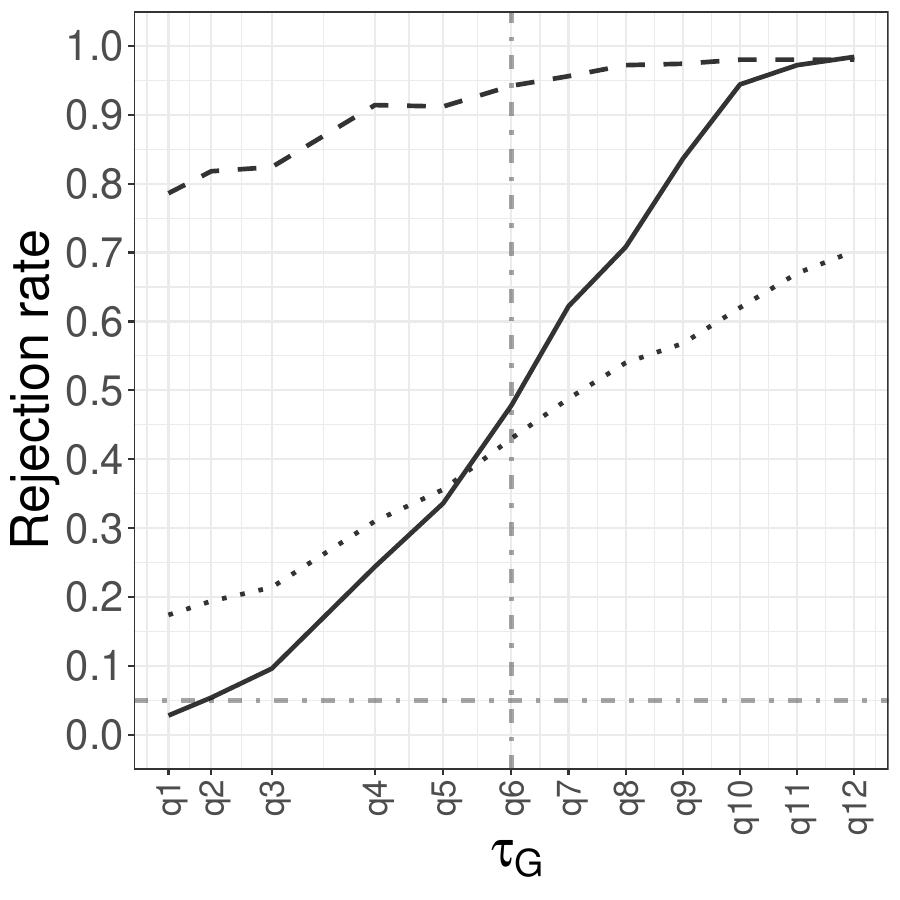}
	\end{subfigure}
	\hfill
	\begin{subfigure}[b]{0.325\textwidth}
		\centering
		\includegraphics[width=\textwidth]{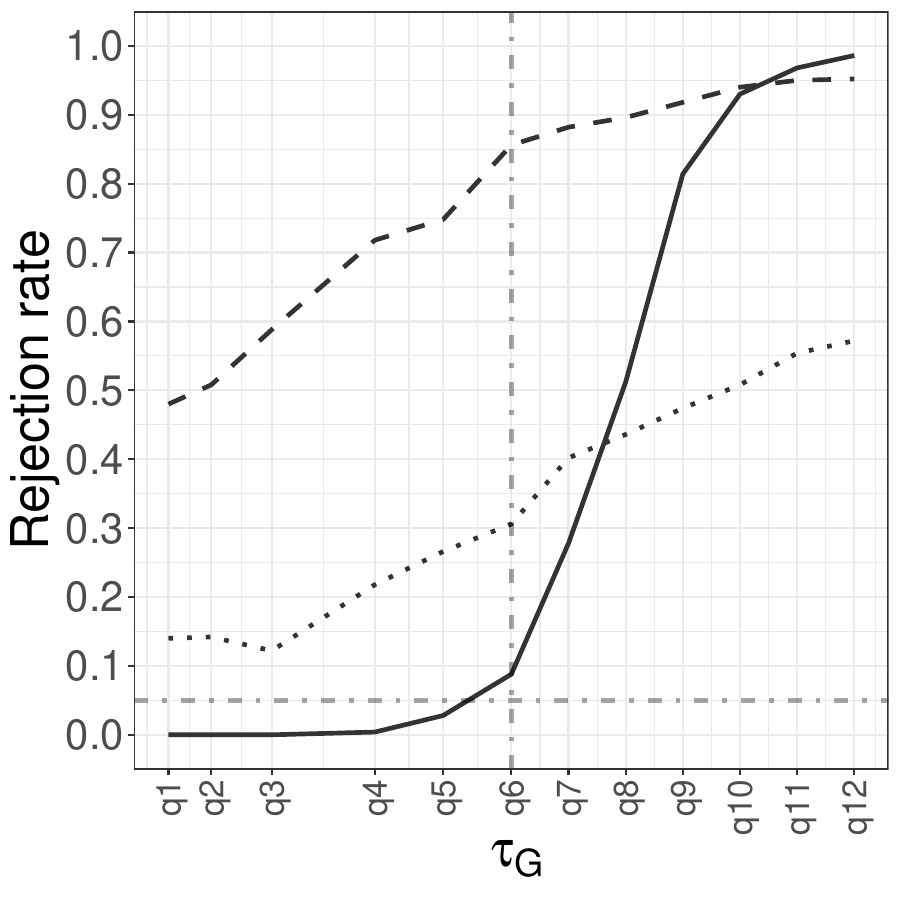}
	\end{subfigure}
	\hfill
	\begin{subfigure}[b]{0.325\textwidth}
		\centering
		\includegraphics[width=\textwidth]{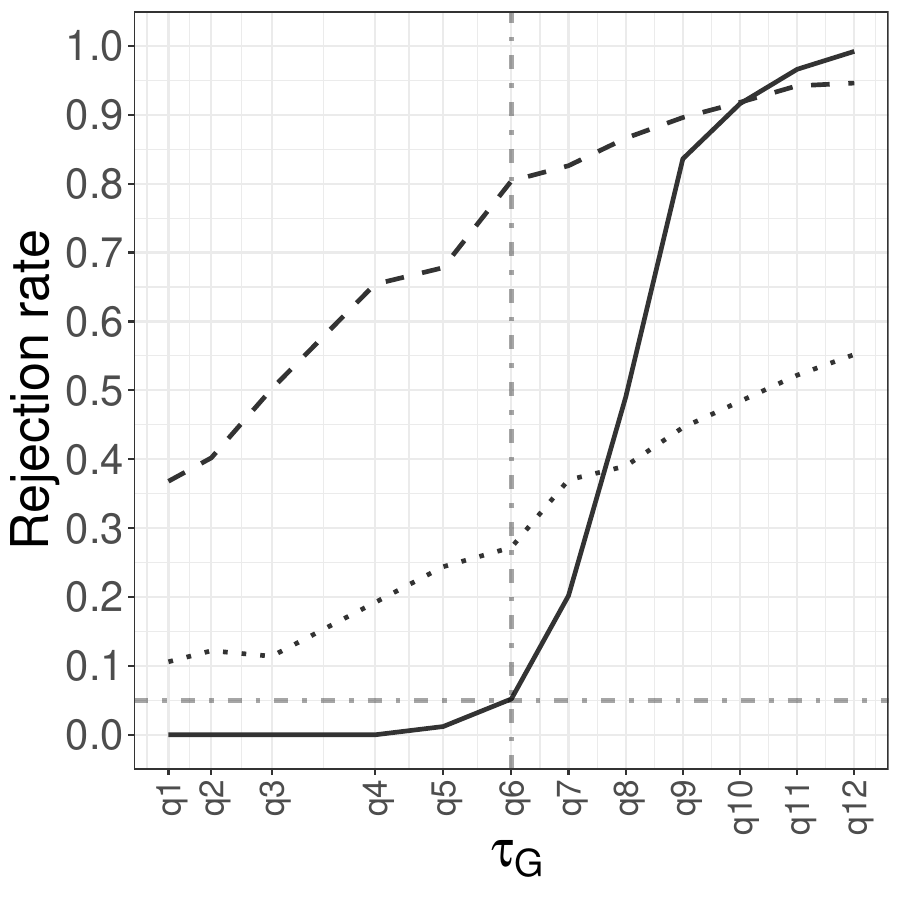}
	\end{subfigure}
	\caption{Rejection rate of the null hypothesis of insufficient follow-up for different methods (solid: $\hat{f}_{nh}^{SG}$, dashed: $\hat{f}_{n}^{G}$, dotted: $Q_{n}$) in Setting~\ref*{enum:sim_exp_unif} with $n=500$, $\Delta G(\tau_{G})=0.02$ and $p=0.2$ (left), $p=0.6$ (center), $p=0.8$ (right).\label{fig:rej_prop_setting_1_by_p}}
\end{figure}%

Figure~\ref{fig:rej_prop_setting_1_by_p} depicts the rejection rate of the null hypothesis of insufficient follow-up against $\tau_{G}$ for Setting~\ref*{enum:sim_exp_unif} when $n=500$, $\Delta G(\tau_{G})=0.02$ and for different choices of $p$. In each subplot, 
$q_{1}-q_{12}$
refer to the 12 different $\tau_{G}$'s considered as described before.
The horizontal dash-dotted line at 0.05 indicates the significance level, and the vertical one at $q_{6}$ indicates $\tau_{G}=q_{0.99}$ (99\% quantile of $F_{u}$). Recall that we are testing
$\tilde{H}_{0}: q_{1-\epsilon} \geq \tau_{G}$ with $\epsilon=0.01$ and an ideal testing procedure is expected to reject the null less when $\tau_{G} < q_{0.99}$ (left-side of the vertical dash-dotted line) while to reject more when $\tau_{G} > q_{0.99}$ (right-side of the vertical dash-dotted line). The rejection rates of $\hat{f}_{nh}^{SG}$, $\hat{f}_{n}^{G}$, and $Q_{n}$ tests are represented by solid, dashed, and dotted lines, respectively, in each subplot. Each of these lines is constructed by interpolating the rejection rates between two consecutive $q$'s. 

From Figure~\ref{fig:rej_prop_setting_1_by_p} we see that the proposed method using $\hat{f}_{nh}^{SG}$ performs better, in terms of level empirically, among the three procedures for different $p$'s, while the test using $\hat{f}_{n}^{G}$ is the worst. The poor performance of $\hat{f}_{n}^{G}$ in controlling the level probably results from the poor approximation of the critical value using the asymptotic distribution at a relatively small sample size. In terms of the empirical power, $\hat{f}_{nh}^{SG}$ behaves better than the $Q_{n}$ test in general, although there is a region for $\tau_{G}$ near $q_{6}$ where the $Q_{n}$ test has higher power. For the effect of $p$, we observe that as $p$ increases the performance improves, i.e. better control of the level and higher power. Note that when $p$ is small the censoring rate is very high (around 85\% for $p=0.2$) and it becomes more difficult to detect that the follow-up is not sufficient. In such cases, larger sample sizes and $\Delta G(\tau_G)$ are required in order to control the level of the test (see Figure~\ref*{supp_fig:rej_prop_exp_unif_p_02} in the Supplementary Material). We also point out that for the formulation of sufficient follow-up in \eqref{eq:insuff_vs_suff} considered by the $Q_n$ test, the null hypothesis is always true for this setting,  but the observed rejection rate is far larger than the nominal level of $5\%$.

\begin{figure}%
	\centering
	\begin{subfigure}[b]{0.325\textwidth}
		\centering
		\includegraphics[width=\textwidth]{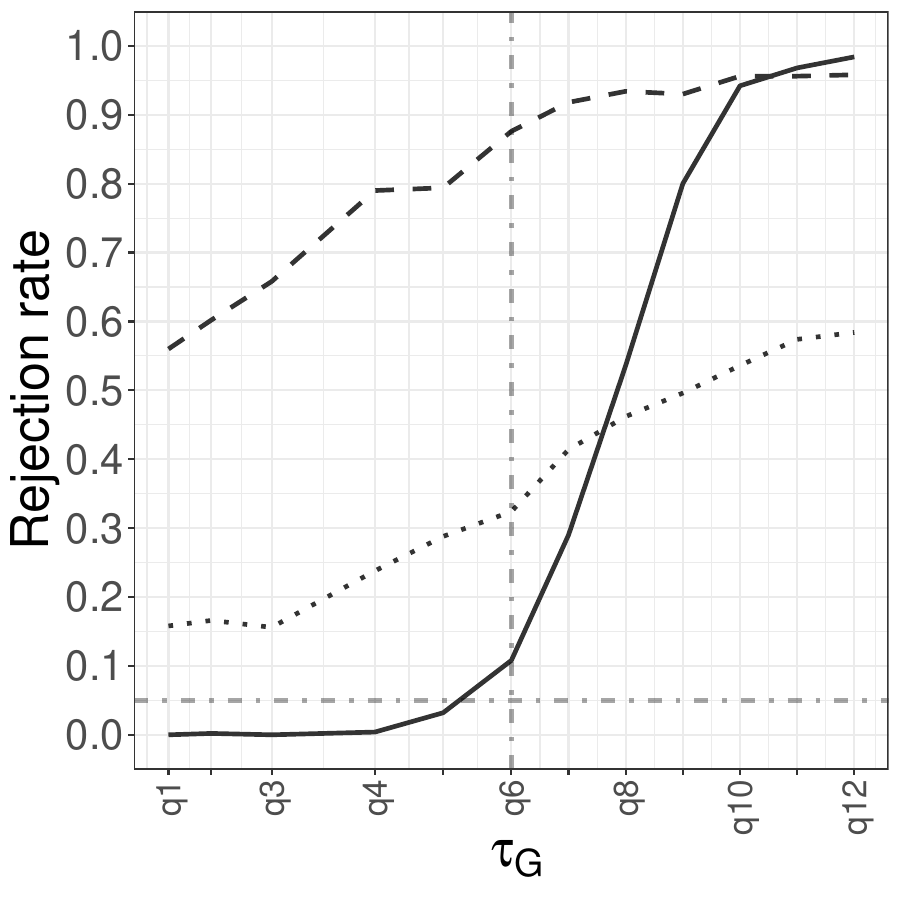}
	\end{subfigure}
	\hfill
	\begin{subfigure}[b]{0.325\textwidth}
		\centering
		\includegraphics[width=\textwidth]{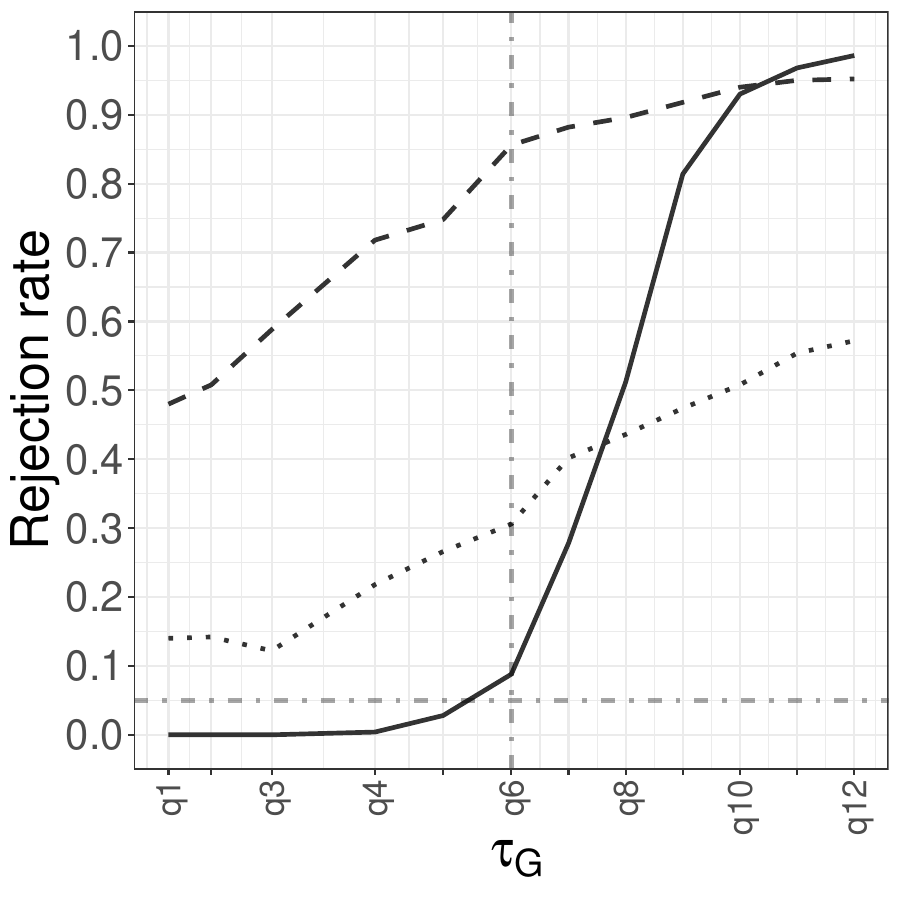}
	\end{subfigure}
	\hfill
	\begin{subfigure}[b]{0.325\textwidth}
		\centering
		\includegraphics[width=\textwidth]{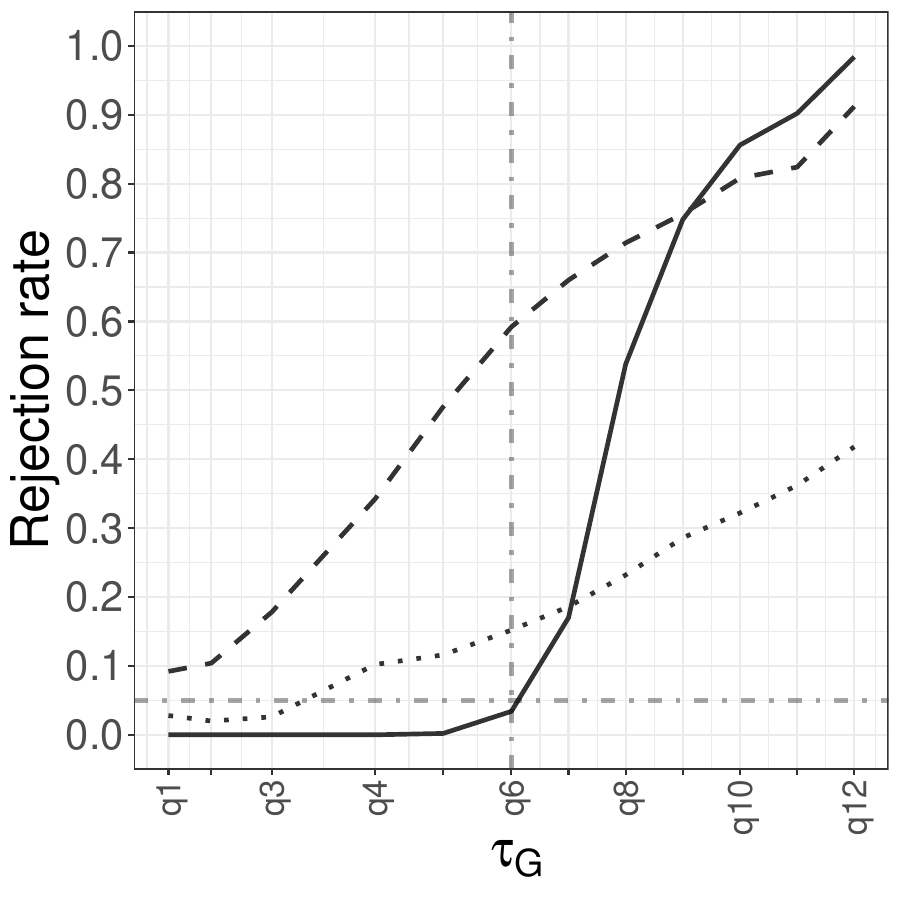}
	\end{subfigure}
	\caption{Rejection rate of the null hypothesis of insufficient follow-up for different methods (solid: $\hat{f}_{nh}^{SG}$, dashed: $\hat{f}_{n}^{G}$, dotted: $Q_{n}$) in Setting~\ref*{enum:sim_exp_unif} with $n=500$, $p=0.6$ and  $\Delta G(\tau_{G})=0$ (left), 0.02 (center), and 0.2 (right).\label{fig:rej_prop_setting_1_by_jumpsize}}
\end{figure}%

Figure~\ref{fig:rej_prop_setting_1_by_jumpsize} shows the rejection rate against $\tau_{G}$ when $n=500$ and $p=0.6$,
while changing $\Delta G(\tau_{G})$. The proposed procedure using $\hat{f}_{nh}^{SG}$ possesses a better control on the empirical level at each $\Delta G(\tau_{G})$, among the three methods, while the Type I error probability of $\hat{f}_{n}^{G}$ remains higher than the others empirically. {The test based on} $\hat{f}_{nh}^{SG}$ also has a higher rejection rate under the alternative 
$\tilde{H}_{a}: \tau_{G} < q_{1-\epsilon}$,
apart from in a small region near $q_{1-\epsilon}$. The performance of all tests improves, in terms of empirical level, as $\Delta G(\tau_{G})$ increases. We note that the asymptotic distributions of the estimators $\hat{f}_{n}^{G}$ and $\hat{f}_{nh}^{SG}$ rely on the assumption that $\Delta G(\tau_{G})>0$. The simulation results suggest the proposed method still works well when such assumption is violated, i.e. $\Delta G(\tau_{G})=0$. 

\begin{figure}%
	\centering
	\begin{subfigure}[b]{0.325\textwidth}
		\centering
		\includegraphics[width=\textwidth]{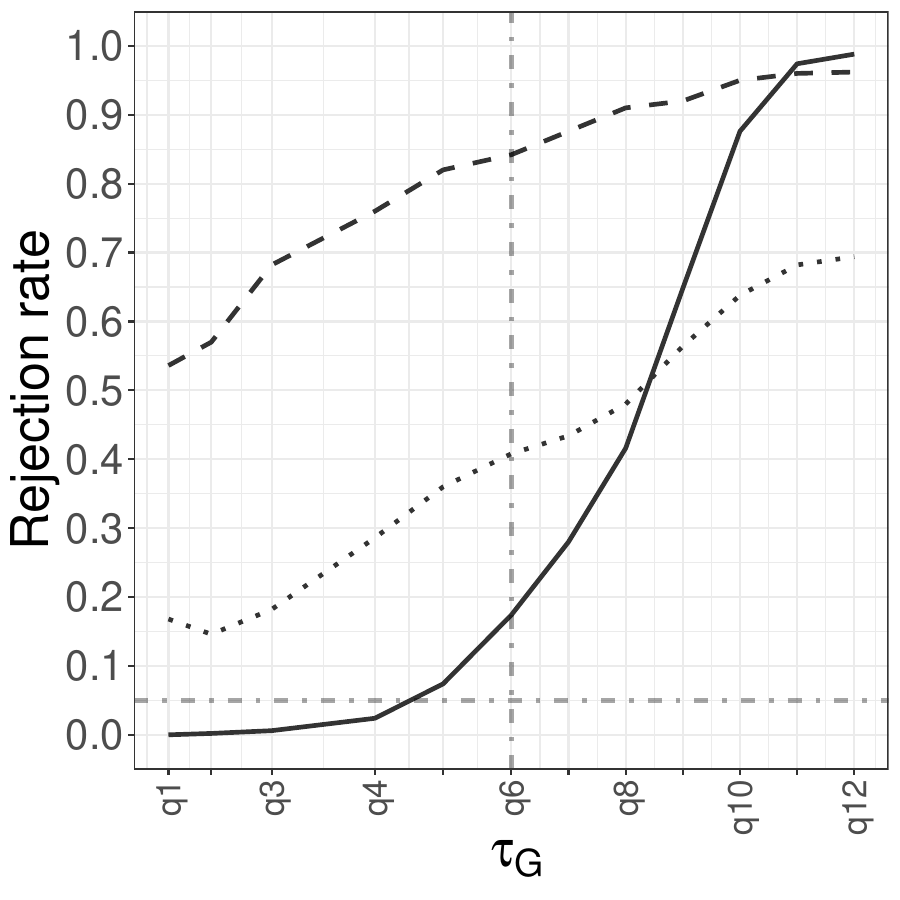}
	\end{subfigure}
	\hfill
	\begin{subfigure}[b]{0.325\textwidth}
		\centering
		\includegraphics[width=\textwidth]{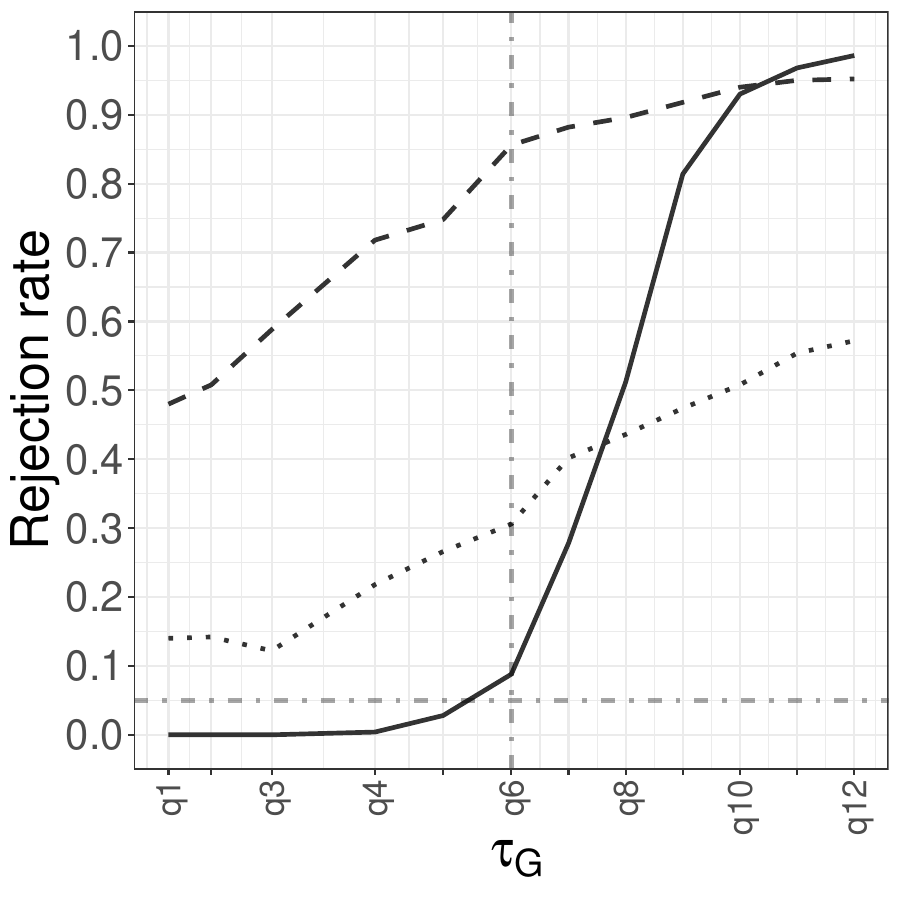}
	\end{subfigure}
	\hfill
	\begin{subfigure}[b]{0.325\textwidth}
		\centering
		\includegraphics[width=\textwidth]{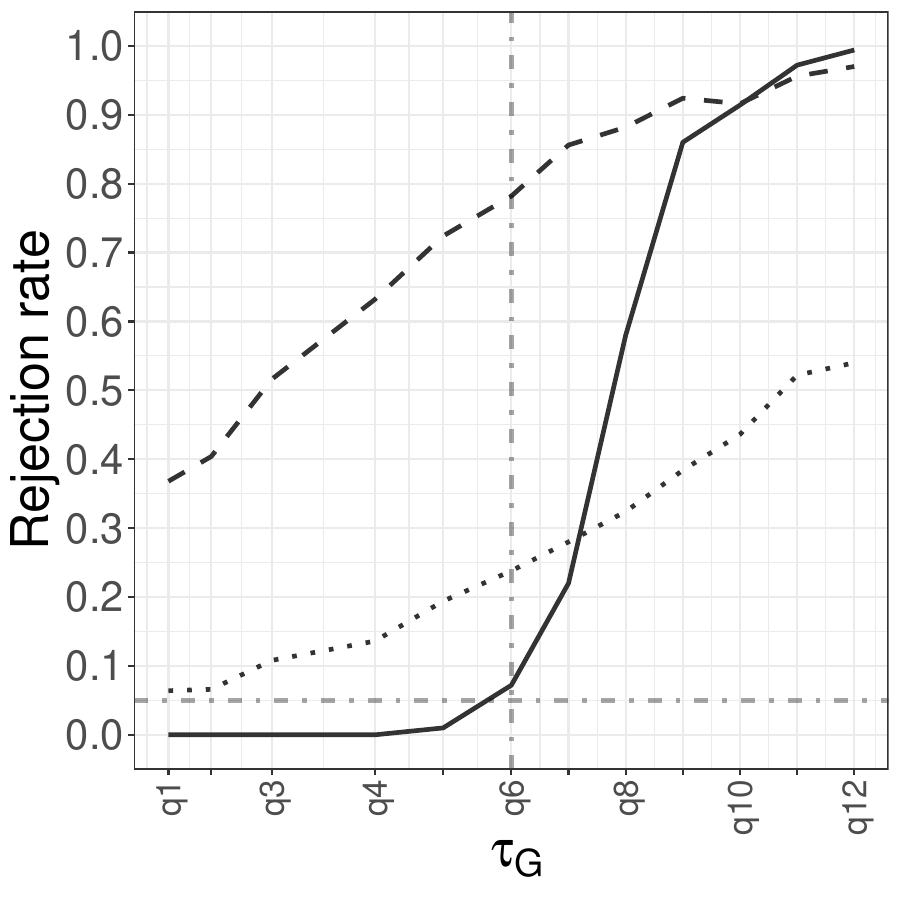}
	\end{subfigure}
	\caption{Rejection rate of the null hypothesis of insufficient follow-up for different methods (solid: $\hat{f}_{nh}^{SG}$, dashed: $\hat{f}_{n}^{G}$, dotted: $Q_{n}$) in Setting~\ref*{enum:sim_exp_unif} with $p=0.6$, $\Delta G(\tau_{G})=0.02$ and a sample size of 200 (left), 500 (center), and 1000 (right).\label{fig:rej_prop_setting_1_by_size}}
\end{figure}%

In 
Figure~\ref{fig:rej_prop_setting_1_by_size}
we investigate the effect of the sample size,  for fixed $p=0.6$ and $\Delta G(\tau_{G})=0.02$.
As expected, we observe a better control of the level and a steeper power curve when the sample size is larger. The proposed procedure using $\hat{f}_{nh}^{SG}$ performs better than others for all three different sample sizes. 
Complete results for all combinations of $p$, $\Delta G(\tau_G)$ and $n$ can be found in Figures~\ref*{supp_fig:rej_prop_exp_unif_p_02}--\ref*{supp_fig:rej_prop_exp_unif_p_08} in the Supplementary Material. In terms of sensitivity of the test with respect to the choice of $\tau$, we observe in Figure~\ref*{supp_fig:rej_prop_exp_unif_p_06_tau} that as $\tau$ increases, the rejection rate decreases, resulting in a more conservative test. The effect is more visible in terms of power than level. 
For the comparison with RECeUS-AIC by \textcite{SO2023}, we observe from Figure~\ref*{supp_fig:rej_prop_exp_unif_p_06_RECeUS} that $\hat{f}_{nh}^{SG}$ performs better in terms of empirical level, while RECeUS-AIC demonstrates higher power in the region where $\tau_G$ is close to $q_{0.99}$. When comparing the $Q_n$ test with RECeUS-AIC, the latter controls the level more effectively when $\tau_G$ is much smaller than $q_{0.99}$, but not in the region near $q_{0.99}$, although RECeUS-AIC exhibits higher empirical power than the $Q_n$ test.
\e{We note that the hypotheses considered by \textcite{SO2023} do not exactly correspond to our hypotheses in \eqref{eq:insuff_vs_suff_new}. 
	Despite this, for consistency, the interpretation of level and power here is made with respect to our hypotheses for $\epsilon=0.01$.}

Similar conclusions can be drawn from the simulation results in Settings~\ref*{enum:sim_exp_exp_trunc}--\ref*{enum:sim_wb_unif} reported in the Supplementary Material. We point out that when 
censoring is heavy and $p$ is small,
all methods exhibited poor performance in terms of empirical level. Problems occur also when the density $f_u$ decreases faster, as in Setting~\ref*{enum:sim_wb_unif}, for which the tests have a better control of the level when $p$ and $n$ are larger. For a more conservative test, one can consider a smaller $\epsilon$, which results in lower rejection rates as 
shown
in Figure~\ref*{supp_fig:rej_prop_wb_unif_trunc_p_06_eps}, but the difference is not large.

\begin{figure}%
	\centering
	\begin{subfigure}[b]{0.325\textwidth}
		\centering
		\includegraphics[width=\textwidth]{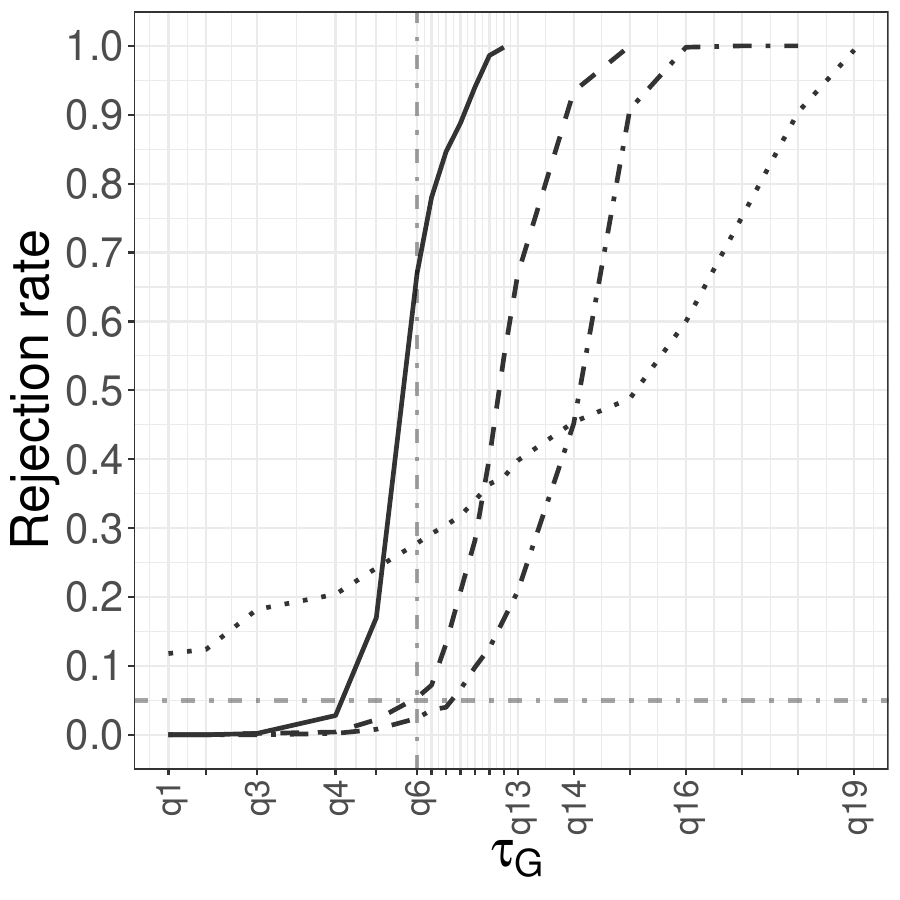}
	\end{subfigure}
	\hfill
	\begin{subfigure}[b]{0.325\textwidth}
		\centering
		\includegraphics[width=\textwidth]{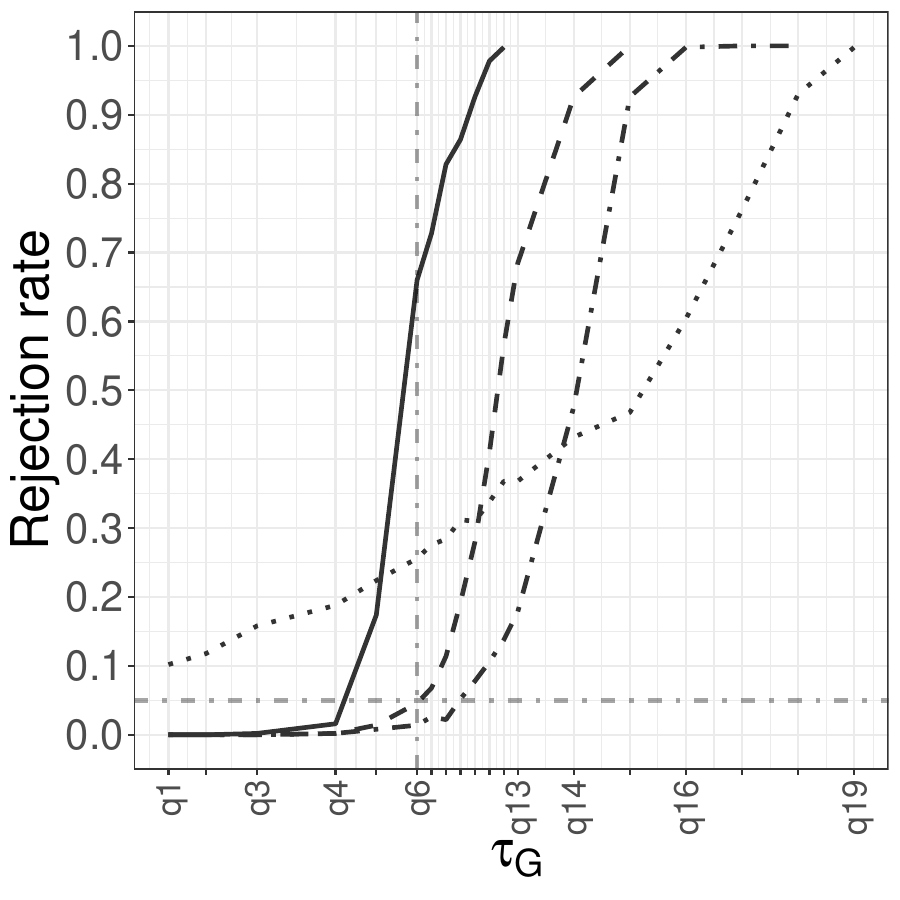}
	\end{subfigure}
	\hfill
	\begin{subfigure}[b]{0.325\textwidth}
		\centering
		\includegraphics[width=\textwidth]{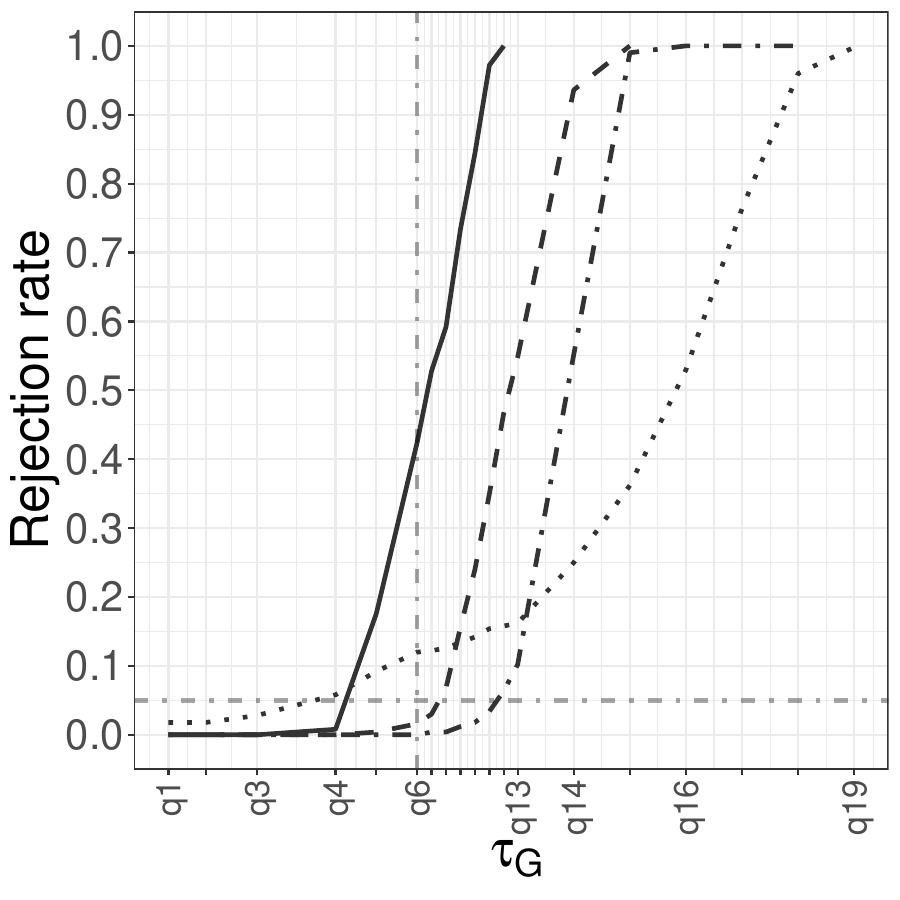}
	\end{subfigure}
	\caption{Rejection rate of the null hypothesis of insufficient follow-up for different methods (solid: $\hat{f}_{nh}^{SG}$ with $\tau_{1}\approx4.5569$, dashed: $\hat{f}_{nh}^{SG}$ with $\tau_{2}=1.25\tau_{1}$, dash-dotted: $\hat{f}_{nh}^{SG}$ with $\tau_{3}=1.5\tau_{1}$, dotted: $Q_{n}$) in Setting~\ref*{enum:sim_texp_unif} with $n=500$, $p=0.6$ and $\Delta G(\tau_{G})=0$ (left), $\Delta G(\tau_{G})=0.02$ (center), and $\Delta G(\tau_{G})=0.2$ (right).\label{fig:rej_prop_setting_4_by_jumpsize}}
\end{figure}%

Next we consider Setting~\ref*{enum:sim_texp_unif}, in which $f_u$ has compact support.
Figure~\ref{fig:rej_prop_setting_4_by_jumpsize} depicts the rejection rate of insufficient follow-up against $\tau_{G}$ when $n=500$ and $p=0.6$,
while varying $\Delta G(\tau_{G})$. In each subplot, 
we introduce 7 extra $\tau_{G}$'s ($q_{13}-q_{19}$ ranging from $\tau_{F_{u}}\approx4.6051$ to 6.9078), in addition to $q_{1}-q_{12}$, to investigate the power of the test. The vertical dash-dotted line indicates $\tau_{G}=q_{0.99}$ (99\% quantile of $F_{u}$), and the horizontal one indicates the significance level. We considered 3 different $\tau$'s, namely $\tau_{1}$ being the 99.95\% quantile of $F_{u}$ as described
previously,
$\tau_{2}=1.25\tau_{1}$, and $\tau_{3}=1.5\tau_{1}$. The rejection rates of $\hat{f}_{nh}^{SG}$ using $\tau_{1}$, $\tau_{2}$, $\tau_{3}$, and the $Q_{n}$ test are rendered by solid, dashed, dash-dotted, and dotted lines, respectively, in each subplot. 
We again observe that as $\tau$ increases, the test becomes more conservative. For $\tau=\tau_1$, the rejection rate 
starts being 
higher than the nominal level of $5\%$ when $\tau_G>q_4=q_{0.975}$, which is still reasonable in practice, while the $Q_n$ test has higher rejection rate even for shorter follow-up. For $\tau_2$ and $\tau_3$ the test behaves well in terms of level and has most of the time has also higher power compared to $Q_n$. Overall, we advice to take $\tau$ possibly larger than $\tau_{F_u}$ when it is believed that the $F_u$ has compact support based on practical knowledge of an 
approximation
for $\tau_{F_u}$.
Similar results are observed in Setting~\ref*{enum:sim_texp_exp_trunc}, see Figure~\ref*{supp_fig:rej_prop_texp_exp_trunc_p_06}. However, in that case the test $Q_n$ exhibits a good control of the level but has almost no power.

\section{Real data application}\label{sec:app}

In this section we illustrate the practical application of the proposed method and compare it with
the existing tests using
two breast cancer datasets, one with a sample size of 286 and another with 1233. Besides using the data with a follow-up cutoff at the end of the study, we construct `what-if' scenarios
in which the follow-up cutoff is earlier than the actual one. In particular, the subjects with the 
follow-up time greater than the hypothetical cutoff are considered as censored with a follow-up time equal to the hypothetical cutoff. Such scenarios help explore whether a test possesses a monotonicity behavior of deciding sufficient follow-up if the follow-up period is lengthened \citep{KY2007}. \textcite{MRSZ2024} also studied the monotonicity behavior of the $Q_n$ statistic when the maximum survival time is fixed while varying the maximum uncensored survival time.

For the proposed procedure ($\hat{f}_{n}^{G}$ and $\hat{f}_{nh}^{SG}$), 
$\tau_{G}$ is approximated by the maximum observed survival time $y_{(n)}$. We used two choices of $\epsilon$, namely 0.01 and 0.025, indicating the follow-up is considered as sufficient if $\tau_{G}>q_{0.99}$ 
when $\epsilon=0.01$.
The distribution of $\hat{f}_{nh}^{SG}$ is approximated using the 
bootstrap procedure described in Section~\ref{sec:bootstrap}
with 1000 bootstrap samples and bandwidths as mentioned in Section~\ref{sec:sim}.
We set $a$ and $c$ to 0.34 and $y_{(n)}$, respectively, for the test using
$\hat{f}_{n}^{G}$. The parameter $\gamma$ for $Q_{n}$ test is set at $1$. The critical value of $T_{n}$ test is computed using 1000 bootstrap samples.

\subsection{Breast cancer study I}
We analyze a dataset of an observational study consisting of 286 lymph-node-negative breast cancer patients. The patients received treatment in 1980--1995, whose age at treatment ranges from 26 to 83 years. The relapse-free survival time (in months), i.e. the time until death or occurrence of distant metastases, is considered. 
The uncensored survival time ranges from 2 to 80 months, and the censoring rate is around 62.59\%. The same dataset has been studied by \textcite{AKL2019} and \textcite{DM2024} under the mixture cure model settings. We refer the reader to \textcite{WangEtal2005} for a detailed description of the data.
\begin{figure}%
	\centering
	\includegraphics[scale=0.5]{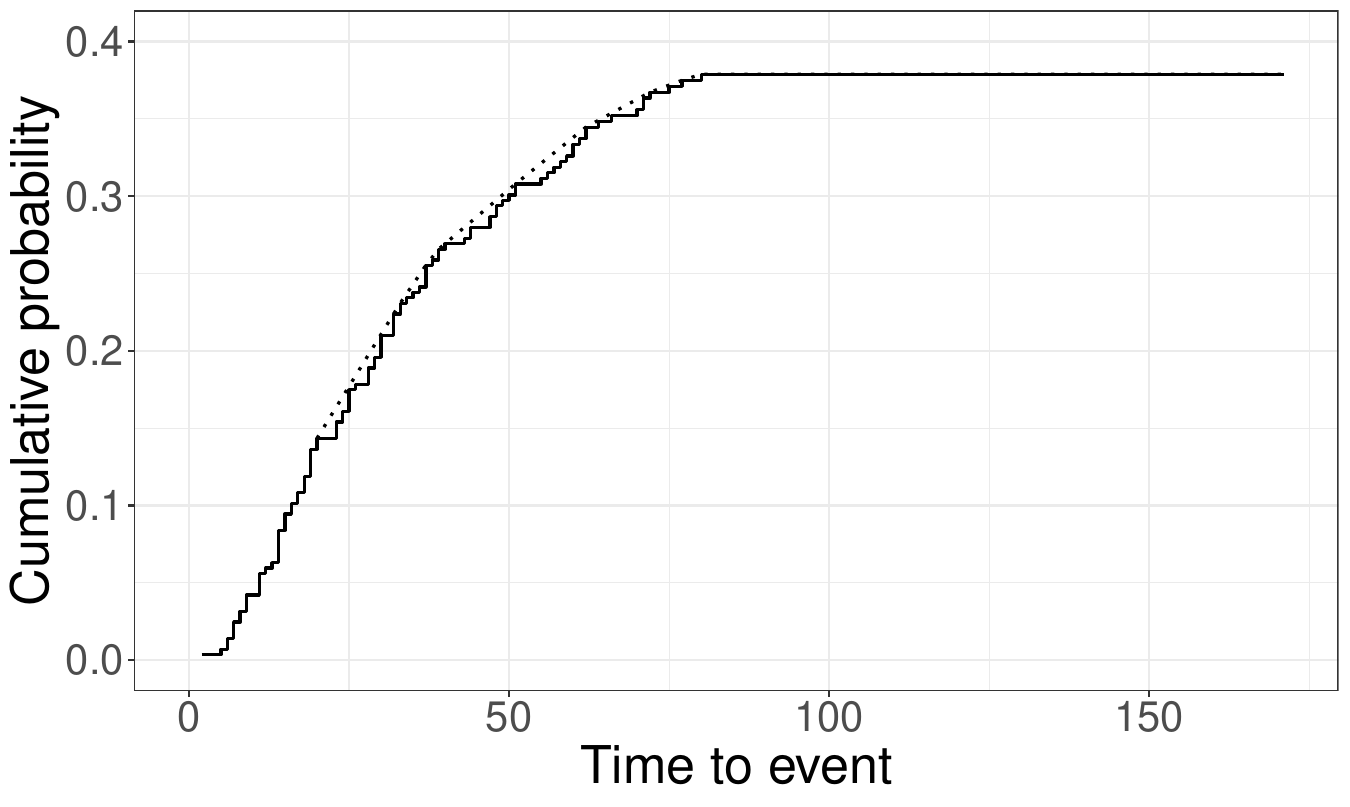}
	\caption{Kaplan--Meier estimate (solid) and its least concave majorant (dotted) for the breast cancer observational study data.\label{fig:kme_lcm_breast}}
\end{figure}%

Figure~\ref{fig:kme_lcm_breast} shows the Kaplan--Meier curve, which exhibits a long plateau. The maximum uncensored survival time $\tilde{y}_{(n)}$ and the maximum observed survival time $y_{(n)}$ are 80 and 171 months, respectively, and there are 152 censored observations between $\tilde{y}_{(n)}$ and $y_{(n)}$ (i.e. on the plateau). A graphical inspection based on the plateau of the Kaplan--Meier curve suggests the follow-up is sufficient. In addition, the Kaplan--Meier curve supports the assumption of decreasing density in the tail (the curve is close to concave at least starting from around $25$ months). One could use some data-driven method to estimate the point at which the density starts being decreasing, but that is not crucial here since it will not affect much the estimator at the end point. We apply the proposed testing procedure ($\hat{f}_{n}^{G}$ and $\hat{f}_{nh}^{SG}$)
, $\alpha_{n}$, $\tilde{\alpha}_{n}$, $Q_{n}$ and $T_n$ tests
for the whole follow-up and the hypothetical scenarios with shorter follow-up.
For the proposed procedure,
we used two choices of $\tau$, namely 240 and 360 months, meaning that we are considering the probability of relapse after 20 or 30 years to be almost zero.
\begin{table}%
	\begin{center}
		\caption{$p$-values of  testing $\tilde{H}_{0}$, $H_{0}$ or $\check{H}_{0}$ at different follow-up cutoffs for the breast cancer observational study (rounded to three decimal places). Bold indicates rejecting the null at the 5\% level or the estimated $p$-value $\alpha_n<0.05$ for $\alpha_n$ test.} \label{tab:breast_rfs_results}
		\renewcommand{\arraystretch}{1.2}
		\setlength{\tabcolsep}{2pt}
		\footnotesize
		\begin{tabular}{@{}c@{\extracolsep{1pt}}cccccccc@{\extracolsep{3pt}}ccc@{\extracolsep{1pt}}c}
	\multirow{4}{*}{\makecell{Cutoff\\(months)}} & \multicolumn{8}{c}{$\tilde{H}_{0}: q_{1-\epsilon}\geq\tau_{G}$}                                                                                                                                                                 & \multicolumn{3}{c}{$H_{0}: \tau_{F_{u}}\geq\tau_{G}$}                                         & $\check{H}_{0}: \tau_{F_{u}}\leq\tau_{G}$ \\
	\cline{2-9}\cline{10-12}\cline{13-13}
	& \multicolumn{2}{c}{\shortstack{\\$\epsilon=0.01$\\$\tau=240$}}                             & \multicolumn{2}{c}{\shortstack{\\$\epsilon=0.01$\\$\tau=360$}}                             & \multicolumn{2}{c}{\shortstack{\\$\epsilon=0.025$\\$\tau=240$}}                             & \multicolumn{2}{c}{\shortstack{\\$\epsilon=0.025$\\$\tau=360$}}        &	\multirow{2}{*}{$\alpha_{n}$} & \multirow{2}{*}{$\tilde{\alpha}_{n}$} & \multirow{2}{*}{$Q_n$} & \multirow{2}{*}{$T_n$}                    \\
	& $\hat{f}_{nh}^{SG}$ & \multicolumn{1}{c}{$\hat{f}_{n}^{G}$} & $\hat{f}_{nh}^{SG}$ & \multicolumn{1}{c}{$\hat{f}_{n}^{G}$} & $\hat{f}_{nh}^{SG}$ & \multicolumn{1}{c}{$\hat{f}_{n}^{G}$} & $\hat{f}_{nh}^{SG}$ & $\hat{f}_{n}^{G}$ &                               &                                       &                        &                                           \\ \hline
	90                      & 0.538              & \multicolumn{1}{c}{0.999}             & 0.542              & \multicolumn{1}{c}{0.999}             & 0.503              & \multicolumn{1}{c}{0.999}             & 0.530              & 0.999             & \textbf{0.002}                & 0.367                                 & 0.134                  & 0.089                                     \\
	110                     & 0.588              & \multicolumn{1}{c}{\textbf{0.000}}             & 0.649              & \multicolumn{1}{c}{\textbf{0.000}}    & 0.380              & \multicolumn{1}{c}{\textbf{0.000}}    & 0.542              & \textbf{0.000}    & \textbf{0.000}                & \textbf{0.049}                        & \textbf{0.002}         & 0.106                                     \\
	130                     & 0.144              & \multicolumn{1}{c}{\textbf{0.000}}    & 0.234              & \multicolumn{1}{c}{\textbf{0.000}}    & \textbf{0.035}     & \multicolumn{1}{c}{\textbf{0.000}}    & 0.131              & \textbf{0.000}    & \textbf{0.000}                & \textbf{0.001}                        & \textbf{0.000}         & 0.167                                     \\
	150                     & \textbf{0.007}     & \multicolumn{1}{c}{\textbf{0.000}}    & \textbf{0.015}     & \multicolumn{1}{c}{\textbf{0.000}}    & \textbf{0.001}     & \multicolumn{1}{c}{\textbf{0.000}}    & \textbf{0.006}     & \textbf{0.000}    & \textbf{0.000}                & \textbf{0.000}                        & \textbf{0.000}         & 0.303                                     \\
	171                     & \textbf{0.000}     & \multicolumn{1}{c}{\textbf{0.000}}    & \textbf{0.000}     & \multicolumn{1}{c}{\textbf{0.000}}    & \textbf{0.000}     & \multicolumn{1}{c}{\textbf{0.000}}    & \textbf{0.000}     & \textbf{0.000}    & \textbf{0.000}                & \textbf{0.000}                        & \textbf{0.000}         & --                \\
	\hline      
\end{tabular}
	\end{center}
\end{table}%

Table~\ref{tab:breast_rfs_results} reports the $p$-values of testing different hypotheses ($\tilde{H}_{0}$, $H_{0}$ or $\check{H}_{0}$) at different cutoffs. With a follow-up cutoff equal to the end of the study
(171 months),
the proposed procedure (using $\hat{f}_{n}^{G}$ and $\hat{f}_{nh}^{SG}$) rejects 
$\tilde{H}_{0}$
at the 5\% level when $\epsilon=0.01$ or 0.025, deciding the follow-up is sufficient for the two different $\tau$'s. The $\alpha_{n}$, $\tilde{\alpha}_{n}$, $Q_{n}$ tests also reject 
$H_{0}$ and deciding sufficient follow-up.
We note that the parameter $\epsilon^{\ast}$ for the $T_{n}$ test is set to $y_{(n)}$ since $2(y_{(n)}-\tilde{y}_{(n)})\geq y_{(n)}$ ($\epsilon^{\ast}$ instead of $\epsilon$ is used to avoid any confusion). If we consider $y_{(n)}$ and $\tilde{y}_{(n)}$ as the approximation of $\tau_{G}$ and $\tau_{F_{u}}$, respectively, then $\epsilon^{\ast}=y_{(n)}=171$ would possibly fall within the interval $(\tau_{G}-\tau_{F_{u}}, 2(\tau_{G}-\tau_{F_{u}})]\approx(91, 182]$. Hence, the asymptotic distribution of the test statistic $T_{n}$ is degenerated as mentioned in Section~3 of \textcite{XEK2023}, while the $T_{n}$ test still concludes the follow-up as sufficient without applying the bootstrap procedure.

By shortening the follow-up period, all methods possess a monotonicity behavior of deciding sufficient follow-up. For example, the proposed procedure using $\hat{f}_{nh}^{SG}$ with $\tau=240$ and $\epsilon=0.01$ decides sufficient follow-up at the 5\% level when the cutoff is at or after 150 months, while changes to conclude insufficient follow-up when the cutoff is less than 150 months. In terms of the cutoff at which the test decision changed, for example 150 for $\hat{f}_{nh}^{SG}$ and 130 for $\hat{f}_{n}^{G}$ when $\epsilon=0.01$, the procedure using $\hat{f}_{n}^{G}$ is less conservative than that using $\hat{f}_{nh}^{SG}$. This behavior was also observed in the simulation. The procedure using $\hat{f}_{n}^{G}$ or $\hat{f}_{nh}^{SG}$ is more conservative than $\alpha_{n}$, $\tilde{\alpha}_{n}$ and $Q_{n}$ tests for this dataset. It was also observed in the simulation, e.g. in Setting~1, that $\alpha_{n}$, $\tilde{\alpha}_{n}$ and $Q_{n}$ tests exhibit higher rejection rate when $\tau_{G}$ is slightly greater than $q_{1-\epsilon}$. For the choice of $\tau$, a larger $\tau$ results in a more conservative test procedure, as shown in the simulation. However, the conclusion does not change significantly when $\tau$ increases from 240 to 360. For $T_{n}$ test, we observe  that {it does not reject sufficient follow-up and} the $p$-values decrease from 0.3 to around 0.09 as the follow-up cutoff reduces.

\subsection{Breast cancer study II}\label{sec:app2}
We consider a dataset of breast cancer patients extracted from the Surveillance, Epidemiology and End Results (SEER) database. The database `Incidence-SEER Research Data, 8 Registries, Nov 2022 Sub (1975-2020)' with follow-up until December 2020 was selected and the breast cancer patients diagnosed in 1992 were extracted. This allows a maximum of 347 months (about 29 years) of follow-up. We excluded the observations with unknown or zero follow-up time, with unstaged or unknown cancer stage, and restricted the dataset to white patients with age less than 60 and with Grade II tumor grade at diagnosis. The event time of interest is the time to death because of breast cancer. This cohort consists of 1233 observations with follow-up ranging from 1 to 347 months and has a censoring rate of 75.91\%. A similar breast cancer dataset extracted from the SEER database was also studied by \textcite{TaiEtal2005}, which suggested the minimum required follow-up time for Grade II breast cancer is 26.3 years if a log-normal distributed uncured survival time is assumed.

\begin{figure}%
	\centering
	\includegraphics[scale=0.5]{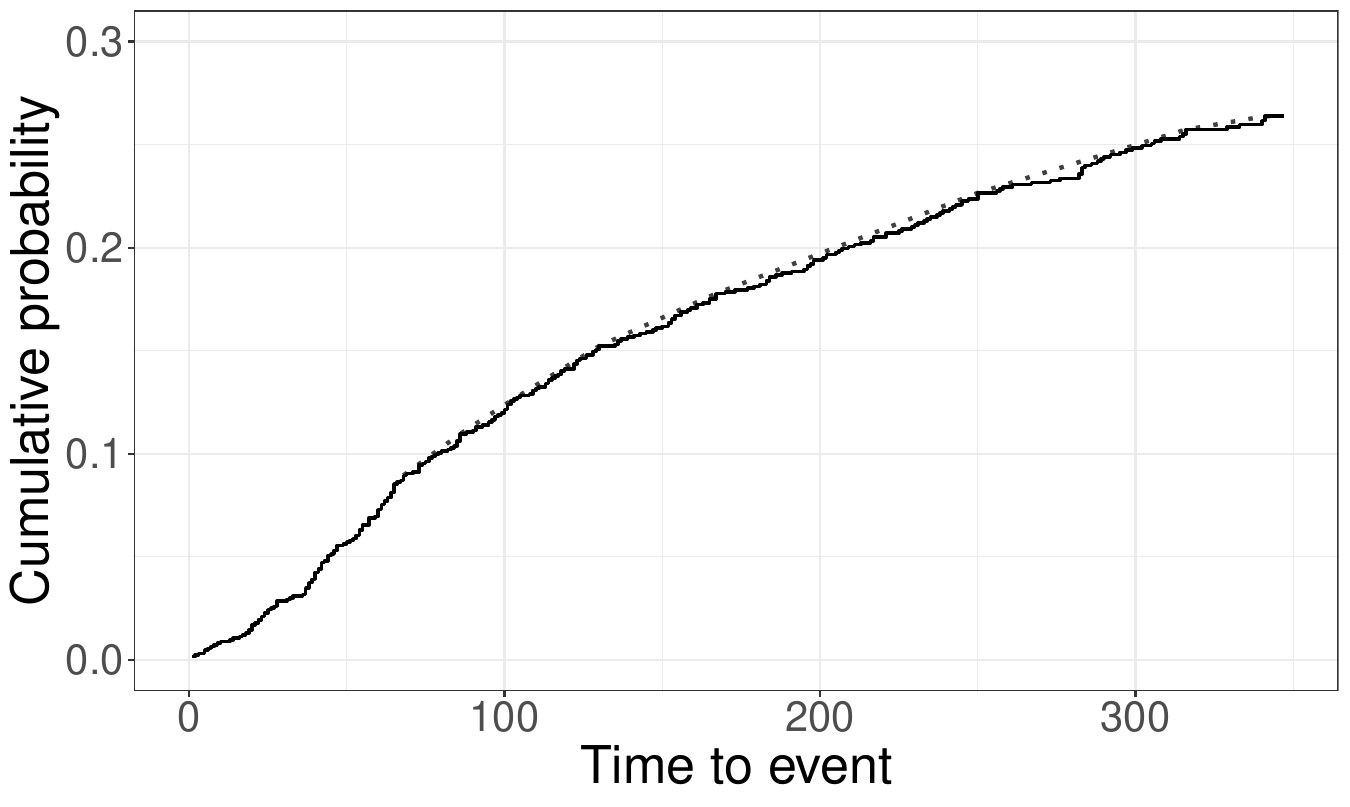}
	\caption{Kaplan--Meier estimate (solid) and its least concave majorant (dotted) for the SEER data.\label{fig:kme_lcm_seer}}
\end{figure}%

Figure~\ref{fig:kme_lcm_seer} depicts the Kaplan--Meier curve for the breast cancer data, which shows a short plateau. The maximum uncensored survival time $\tilde{y}_{(n)}$ and the maximum observed survival time $y_{(n)}$ are 341 and 347 months, respectively, and there are 295 censored observations on the plateau. Such visual inspection based on the Kaplan--Meier curve is not adequate to assess sufficient follow-up. We apply the introduced testing procedure (using $\hat{f}_{n}^{G}$ and $\hat{f}_{nh}^{SG}$),
$\alpha_{n}$, $\tilde{\alpha}_{n}$ and $Q_{n}$ tests to this dataset
for testing the null hypothesis of insufficient follow-up. We do not apply 
the $T_{n}$ test
given that a short plateau is observed from the Kaplan--Meier curve and there is no strong evidence of sufficient follow-up. Besides using the data with a follow-up cutoff in December 2020, we again  construct also hypothetical scenarios with shorter follow-up. 
For the proposed procedure (using $\hat{f}_{n}^{G}$ and $\hat{f}_{nh}^{SG}$), 
the Kaplan--Meier curve in Figure~\ref{fig:kme_lcm_seer} supports the assumption of decreasing density in the tail, at least starting from around 70 months. We again use two choices of $\tau$, namely 360 and 480 months. This means that we can consider the probability of a  cancer related death after 30 or 40 years as negligible.

\begin{table}%
	\begin{center}
		\caption{$p$-values of testing $\tilde{H}_{0}$ or $H_{0}$ at different follow-up cutoffs for the SEER dataset (rounded to three decimal places). Bold indicates rejecting the null at the 5\% level or the estimated $p$-value $\alpha_n<0.05$ for $\alpha_n$ test.}
		\label{tab:seer_results}
		\setlength{\tabcolsep}{5pt}
		\footnotesize
		\begin{tabular}{@{}c@{\extracolsep{3pt}}cccccccc@{\extracolsep{6pt}}ccc}
	\multirow{4}{*}{\makecell{Cutoff\\(months)}} & \multicolumn{8}{c}{$\tilde{H}_{0}: q_{1-\epsilon}\geq\tau_{G}$}                                                                                                                                                                 & \multicolumn{3}{c}{$H_{0}: \tau_{F_{u}}\geq\tau_{G}$}                                          \\
	\cline{2-9}\cline{10-12}
	& \multicolumn{2}{c}{\shortstack{\\$\epsilon=0.01$\\$\tau=360$}}                             & \multicolumn{2}{c}{\shortstack{\\$\epsilon=0.01$\\$\tau=480$}}                             & \multicolumn{2}{c}{\shortstack{\\$\epsilon=0.025$\\$\tau=360$}}                             & \multicolumn{2}{c}{\shortstack{\\$\epsilon=0.025$\\$\tau=480$}}        & \multirow{2}{*}{$\alpha_{n}$} & \multirow{2}{*}{$\tilde{\alpha}_{n}$} & \multirow{2}{*}{$Q_n$} \\
	& $\hat{f}_{nh}^{SG}$ & \multicolumn{1}{c}{$\hat{f}_{n}^{G}$} & $\hat{f}_{nh}^{SG}$ & \multicolumn{1}{c}{$\hat{f}_{n}^{G}$} & $\hat{f}_{nh}^{SG}$ & \multicolumn{1}{c}{$\hat{f}_{n}^{G}$} & $\hat{f}_{nh}^{SG}$ & $\hat{f}_{n}^{G}$ &                               &                                       &                        \\ \hline
	180                     & 0.991              & \multicolumn{1}{c}{0.999}             & 0.992              & \multicolumn{1}{c}{0.999}             & 0.988              & \multicolumn{1}{c}{0.999}             & 0.990              & 0.999             & 0.368                         & 0.368                                 & 0.75                   \\
	200                     & 0.999              & \multicolumn{1}{c}{0.999}             & 0.999              & \multicolumn{1}{c}{0.999}             & 0.999              & \multicolumn{1}{c}{0.999}             & 0.999              & 0.999             & \textbf{0.050}                & 0.135                                 & 0.316                  \\
	220                     & 0.999              & \multicolumn{1}{c}{0.999}             & 0.999              & \multicolumn{1}{c}{0.999}             & 0.999              & \multicolumn{1}{c}{0.999}             & 0.999              & 0.999             & \textbf{0.050}                & 0.135                                 & 0.563                  \\
	240                     & 0.999              & \multicolumn{1}{c}{0.999}             & 0.999              & \multicolumn{1}{c}{0.999}             & 0.999              & \multicolumn{1}{c}{0.999}             & 0.999              & 0.999             & 0.368                         & 0.368                                 & 0.563                  \\
	260                     & 0.999              & \multicolumn{1}{c}{0.999}             & 0.999              & \multicolumn{1}{c}{0.999}             & 0.999              & \multicolumn{1}{c}{0.999}             & 0.999              & 0.999             & 0.135                         & 0.368                                 & 0.422                  \\
	280                     & 0.881              & \multicolumn{1}{c}{0.999}             & 0.906              & \multicolumn{1}{c}{0.999}             & 0.788              & \multicolumn{1}{c}{0.999}             & 0.881              & 0.999             & 0.135                         & 0.368                                 & 0.563                  \\
	300                     & 0.999              & \multicolumn{1}{c}{0.999}             & 0.999              & \multicolumn{1}{c}{0.999}             & 0.999              & \multicolumn{1}{c}{0.999}             & 0.999              & 0.999             & 0.368                         & 0.368                                 & 0.75                   \\
	320                     & 0.999              & \multicolumn{1}{c}{0.999}             & 0.999              & \multicolumn{1}{c}{0.999}             & 0.998              & \multicolumn{1}{c}{0.999}             & 0.999              & 0.999             & \textbf{0.018}                & 0.135                                 & 0.422                  \\
	340                     & 0.875              & \multicolumn{1}{c}{0.999}             & 0.983              & \multicolumn{1}{c}{0.999}             & 0.305              & \multicolumn{1}{c}{0.999}             & 0.969              & 0.999             & 0.999                             & 0.368                                 & 0.75                   \\
	347                     & 0.739              & \multicolumn{1}{c}{0.999}             & 0.947              & \multicolumn{1}{c}{0.999}             & \textbf{0.007}     & \multicolumn{1}{c}{\textbf{0.001}}    & 0.931              & 0.999             & 0.135                         & 0.368                                 & 0.563                 \\
	\hline
\end{tabular}
	\end{center}
\end{table}%
Table~\ref{tab:seer_results}
reports the $p$-values of  testing different hypotheses ($\tilde{H}_{0}$ or $H_{0}$) of sufficient follow-up at different follow-up cutoffs. At 5\% level and with a 
cutoff in December 2020 (347 months), the proposed procedure ($\hat{f}_{n}^{G}$ and $\hat{f}_{nh}^{SG}$) does not reject 
$\tilde{H}_{0}$
when $\epsilon=0.01$, 
deciding
the follow-up as insufficient. However, the follow-up is considered sufficient when considering a more relaxed notion of sufficient follow-up for $\epsilon=0.025$ and $\tau=360$. The $\alpha_{n}$, $\tilde{\alpha}_{n}$ and $Q_{n}$ tests do not reject 
$H_{0}$ and conclude insufficient follow-up.
If the follow-up is shortened, all tests  apart from $\alpha_{n}$ do not reject the hypothesis of insufficient follow-up. Moreover, for the proposed procedure the decision remains the same for all considered choices of $\epsilon$ and $\tau$.

\section{Conclusion and discussion}
In this paper, we proposed a new test for the null hypothesis of insufficient follow-up based on a more relaxed and realistic notion of `practically' sufficient follow-up, meaning that the probability for the event to happen after the end of the study is smaller than a prespecified threshold. The test relies on the assumption that the density function of the survival times is non-increasing in the tail region and two estimators under such shape constraint are considered. Despite guarantees on the asymptotic level of the test, we observe that in practice the test based on the smoothed Grenander estimator and a bootstrap procedure behaves the best. An extensive simulation study showed that the proposed procedure performs well in terms of both level and power, outperforming most of the time  the existing methods. Nevertheless, there is a trade off between the simplicity of the existing methods such as the $Q_n$ test and the accuracy of the proposed method. In addition, we warn that for scenarios with very high censoring rate, low uncured fraction and small sample size the empirical level is larger than the nominal one, but this is a problem even for the other existing tests. This is also observed when the event time follows a distribution that has a density decreases much faster than that of the standard exponential distribution. Under those scenarios with high censoring rate and small uncured fraction, a larger sample size is required for better test performance. 

The relaxed notion of `practically' sufficient follow-up involves the choice of $\epsilon$. 
An overly conservative choice of $\epsilon$ (i.e., when $\epsilon$ is close to 0) may lead to a conclusion of insufficient follow-up, even when a cure model would have been appropriate. Conversely, an anti-conservative choice of $\epsilon$ could result in incorrectly concluding sufficient follow-up.
We emphasize that $\epsilon$ is not chosen based on the application at hand. Although the relaxed notion of sufficient follow-up is defined based on the probability of the event happening after the study ends is smaller than $\epsilon$, which still guarantees good identification of the cure fraction, $\epsilon$ does not depend on the event time distribution or the end of study.
\textcite{SO2023} showed through simulations that improvements in mean squared error and the coverage of the confidence interval for the cure fraction estimate diminish as the follow-up time exceeds the 99\%-quantile of the event time distribution, and suggested that the follow-up can be considered sufficient when the probability of the event occurring after the study ends is less than 1\% (i.e., $\epsilon<0.01$). This motivates our choice of $\epsilon=0.01$.%

The only user-specified parameter $\tau$ for the proposed method, such that it is almost impossible for the event to happen after $\tau$, should rely on medical knowledge. When such information is not available or uncertain, we suggest to use a larger $\tau$ for a more conservative test and  to perform a sensitivity analysis on $\tau$.
Furthermore, an interesting extension of the proposed test would be to consider testing sufficient follow-up in the presence of covariates. 

	\begingroup
	\setlength\bibitemsep{\parsep}
	\setlength{\bibhang}{\leftmargini}
	\printbibliography
	\endgroup
\end{refsection}	

	\begin{refsection}
	\pagebreak
\begin{center}
	\textbf{\large Supplementary Material: \\
		Testing for sufficient follow-up in survival data with a cure fraction\\[5pt]
		Tsz Pang Yuen and Eni Musta}
\end{center}

\setcounter{section}{0}
\setcounter{equation}{0}
\setcounter{figure}{0}
\setcounter{table}{0}
\setcounter{algorithm}{0}
\setcounter{page}{1}
\makeatletter
\renewcommand{\thesection}{S\arabic{section}}
\renewcommand{\theequation}{S\arabic{equation}}
\renewcommand{\thefigure}{S\arabic{figure}}
\renewcommand{\thetable}{S\arabic{table}}
\renewcommand{\thealgorithm}{S\arabic{algorithm}}

\section{Simulation study}
\subsection{Simulation settings}

Tables \ref*{supp_tab:exp_unif_cens_rate}--\ref*{supp_tab:texp_exp_trunc_cens_rate} report the censoring rates over different values of $\tau_G$ for Settings \ref*{enum:sim_exp_unif}--\ref*{enum:sim_texp_exp_trunc}.
\begin{table}[h]
	{
		\centering
		\caption{Censoring rates for Setting 1.\label{supp_tab:exp_unif_cens_rate}}
		{
			\small
			\begin{tabular}{ccc}
	$p$                  & $\Delta G(\tau_{G})$ & Censoring rate     \\ \hline
	\multirow{4}{*}{0.2} & 0                    & 82.91\% -- 87.84\% \\
	& 0.02                 & 82.85\% -- 87.73\% \\
	& 0.05                 & 82.76\% -- 87.56\% \\
	& 0.2                  & 82.34\% -- 86.69\% \\ \hline
	\multirow{4}{*}{0.6} & 0                    & 48.74\% -- 63.52\% \\
	& 0.02                 & 48.56\% -- 63.16\% \\
	& 0.05                 & 48.29\% -- 62.64\% \\
	& 0.2                  & 47.02\% -- 60.04\% \\ \hline
	\multirow{4}{*}{0.8} & 0                    & 31.64\% -- 51.35\% \\
	& 0.02                 & 31.39\% -- 50.87\% \\
	& 0.05                 & 31.05\% -- 50.18\% \\
	& 0.2                  & 29.33\% -- 46.72\% \\ \hline
\end{tabular}
		}\par
	}
\end{table}

\begin{table}[h]
	{
		\centering
		\caption{Censoring rates for Setting 2.\label{supp_tab:exp_exp_trunc_cens_rate}}
		{
			\small
			\begin{tabular}{ccc}
	$p$                  & $\lambda$ & Censoring rate     \\ \hline
	\multirow{3}{*}{0.2} & 0.4       & 91.13\% -- 91.18\% \\
	& 1         & 86.69\% -- 87.12\% \\
	& 5         & 81.86\% -- 83.30\% \\ \hline
	\multirow{3}{*}{0.6} & 0.4       & 73.33\% -- 73.49\% \\
	& 1         & 60.02\% -- 61.29\% \\
	& 5         & 45.56\% -- 49.87\% \\ \hline
	\multirow{3}{*}{0.8} & 0.4       & 64.46\% -- 64.66\% \\
	& 1         & 46.65\% -- 48.36\% \\
	& 5         & 27.34\% -- 33.09\% \\ \hline
\end{tabular}
		}\par
	}
\end{table}

\begin{table}[h]
	{
		\centering
		\caption{Censoring rates for Setting 3.\label{supp_tab:wb_unif_cens_rate}}
		{
			\small
			\begin{tabular}{ccc}
	$p$                  & $\Delta G(\tau_{G})$ & Censoring rate     \\ \hline
	\multirow{4}{*}{0.2} & 0         & 80.83\% -- 85.08\% \\
	& 0.02      & 80.82\% -- 85.02\% \\
	& 0.05      & 80.80\% -- 84.92\% \\
	& 0.2       & 80.68\% -- 84.46\% \\ \hline
	\multirow{4}{*}{0.6} & 0         & 42.57\% -- 55.26\% \\
	& 0.02      & 42.53\% -- 55.08\% \\
	& 0.05      & 42.45\% -- 54.80\% \\
	& 0.2       & 42.08\% -- 53.41\% \\ \hline
	\multirow{4}{*}{0.8} & 0         & 23.39\% -- 40.28\% \\
	& 0.02      & 23.32\% -- 40.03\% \\
	& 0.05      & 23.22\% -- 39.66\% \\
	& 0.2       & 22.73\% -- 37.81\% \\ \hline
\end{tabular}
		}\par
	}
\end{table}

\begin{table}[h]
	{
		\centering
		\caption{Censoring rates for Setting 4.\label{supp_tab:texp_unif_cens_rate}}
		{
			\small
			\begin{tabular}{ccc}
	$p$                  & $\Delta G(\tau_{G})$ & Censoring rate     \\ \hline
	\multirow{4}{*}{0.6} & 0         & 48.38\% -- 63.85\% \\
	& 0.02      & 48.22\% -- 63.48\% \\
	& 0.05      & 47.96\% -- 62.94\% \\
	& 0.2       & 46.71\% -- 60.28\% \\ \hline
\end{tabular}
		}\par
	}
\end{table}

\begin{table}[h]
	{
		\centering
		\caption{Censoring rates for Setting 5.\label{supp_tab:texp_exp_trunc_cens_rate}}
		{
			\small
			\begin{tabular}{ccc}
	$p$                  & $\lambda_{C}$ & Censoring rate     \\ \hline
	\multirow{2}{*}{0.6} & 0.5           & 45.24\% -- 49.75\% \\
	& 3             & 62.14\% -- 63.23\% \\ \hline
\end{tabular}
		}\par
	}
\end{table}

\clearpage
\subsection{Simulation results}
$q_{1}$, $q_{2}$, $q_{3}$, $q_{4}$, $q_{6}$, and $q_{12}$ are the 90\%, 92.5\%, 95\%, 97.5\%, 99\%, and 99.9\% quantiles of $F_{u}$, resp.; $q_{5}$ is the mid-point between $q_{4}$ and $q_{6}$; $q_{7}, \cdots, q_{11}$ are 5 evenly separated points between $q_{6}$ and $q_{12}$ in the figures in this section. Table \ref*{supp_tab:actual_tauc} reports their actual values for Settings 1--5.

\begin{table}[h]
	{
		\centering
		\caption{Actual values of $q_{1}$, ..., $q_{12}$ for Settings 1--5.\label{supp_tab:actual_tauc}}
		{
			\small
			\begin{tabular}{lccccc}
	\multirow{2}{*}{$\tau_{G}$} & \multicolumn{5}{c}{Setting}                 \\
	& 1      & 2      & 3       & 4      & 5      \\ \hline
	$q_{1}$                     & 2.3026 & 0.4605 & 7.9528  & 2.2164 & 0.4433 \\
	$q_{2}$                     & 2.5903 & 0.5181 & 10.0642 & 2.4740 & 0.4948 \\
	$q_{3}$                     & 2.9957 & 0.5991 & 13.4616 & 2.8218 & 0.5644 \\
	$q_{4}$                     & 3.6889 & 0.7378 & 20.4117 & 3.3596 & 0.6719 \\
	$q_{5}$                     & 4.1470 & 0.8294 & 26.1116 & 3.6383 & 0.7277 \\
	$q_{6}$                     & 4.6052 & 0.9210 & 31.8114 & 3.9170 & 0.7834 \\
	$q_{7}$                     & 4.9889 & 0.9978 & 38.4388 & 4.0160 & 0.8032 \\
	$q_{8}$                     & 5.3727 & 1.0745 & 45.0661 & 4.1149 & 0.8230 \\
	$q_{9}$                     & 5.7565 & 1.1513 & 51.6935 & 4.2139 & 0.8428 \\
	$q_{10}$                    & 6.1402 & 1.2280 & 58.3209 & 4.3129 & 0.8626 \\
	$q_{11}$                    & 6.5240 & 1.3048 & 64.9483 & 4.4118 & 0.8824 \\
	$q_{12}$                    & 6.9078 & 1.3816 & 71.5756 & 4.5108 & 0.9022 \\ \hline
\end{tabular}
		}\par
	}
\end{table}

\paragraph{Remark}
The rejection rate for the $\alpha_n$ test is computed using the decision rule: $H_0$ is rejected if $\alpha_n<0.05$ as originally proposed in \textcite{MZ1992,MZ1994} as such decision rule is still being widely used in practice.
As mentioned in \citet[Page 85]{MZ1996}, it is more appropriate to compare $\alpha_n$ with the 5\% quantile of its (limiting) distribution.
Therefore the type I error rate of the $\alpha_n$ test reported in this section is inflated. The results for the $\alpha_n$ test here are provided to increase awareness within the community of the risks associated with the use of this decision rule.

\subsubsection{Setting \ref*{enum:sim_exp_unif}}
Figures \ref*{supp_fig:rej_prop_exp_unif_p_02}--\ref*{supp_fig:rej_prop_exp_unif_p_08} show the rejection rate {of insufficient follow-up} against $\tau_{G}$ for Setting \ref*{enum:sim_exp_unif}, each figure with different $p$. To study the effect of the choice of $\tau$ on the rejection rate, we considered 3 different $\tau$'s, namely $\tau_{1}\approx7.601$ (the 99.95\% quantile of $F_{u}$), $\tau_{2}\approx9.210$ (the 99.99\% quantile of $F_{u}$), and $\tau_{3}=(\tau_{1}+\tau_{2})/2\approx8.406$. The rejection rate of $\hat{f}_{nh}^{SG}$ using $\tau_{1}$, $\tau_{2}$, and $\tau_{3}$ are rendered by solid, long dashed, and dashed lines, respectively, in Figure \ref*{supp_fig:rej_prop_exp_unif_p_06_tau}. In summary, larger $\tau$ results in a more conservative procedure, i.e. better control on the empirical level {but with some loss of} empirical power.
\begin{figure}[H]
	\begin{center}
		\includegraphics[width=\textwidth]{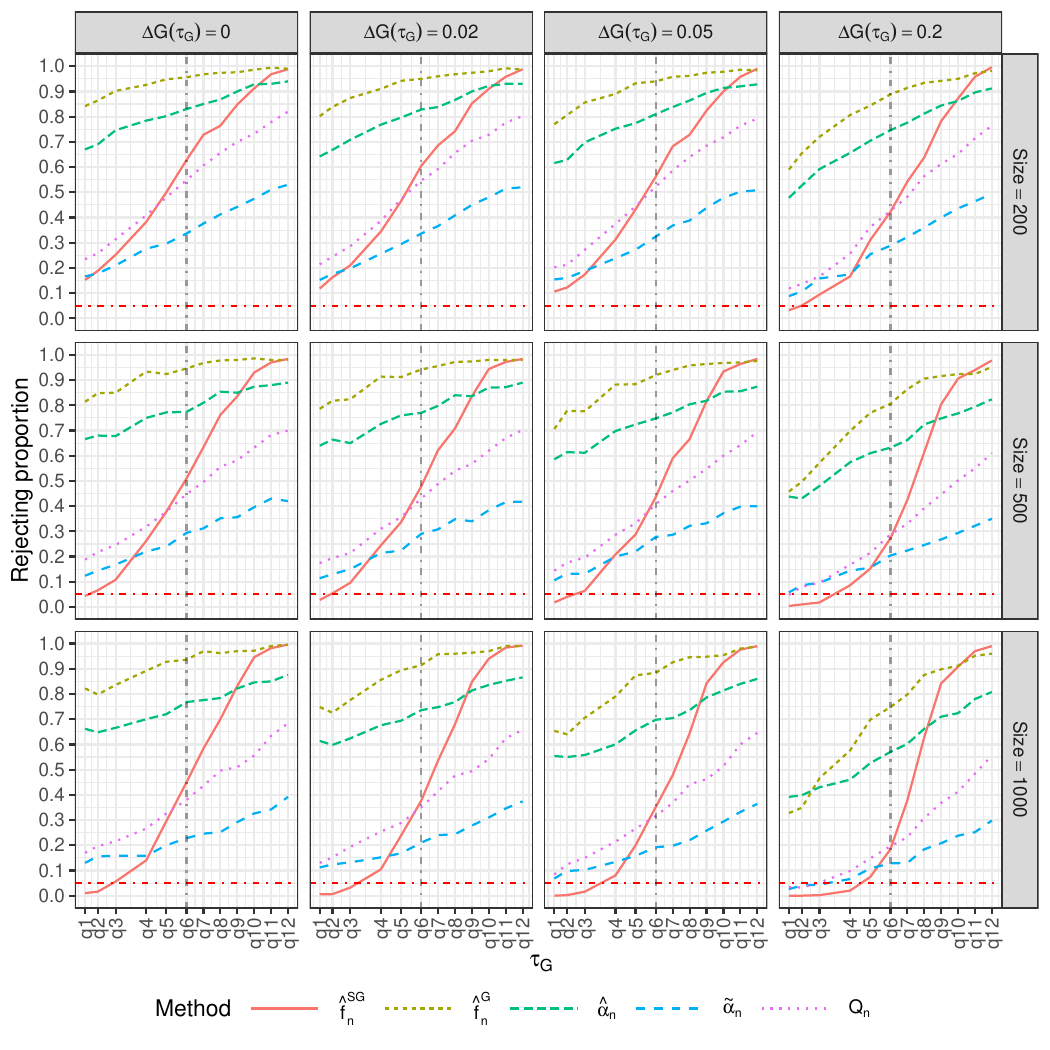}
		\caption{{Rejection rate of the null hypothesis of insufficient follow-up for different methods in} Setting \ref*{enum:sim_exp_unif} when $p=0.2$ (uncured fraction).\label{supp_fig:rej_prop_exp_unif_p_02}}
	\end{center}
\end{figure}

\begin{figure}[H]
	\begin{center}
		\includegraphics[width=\textwidth]{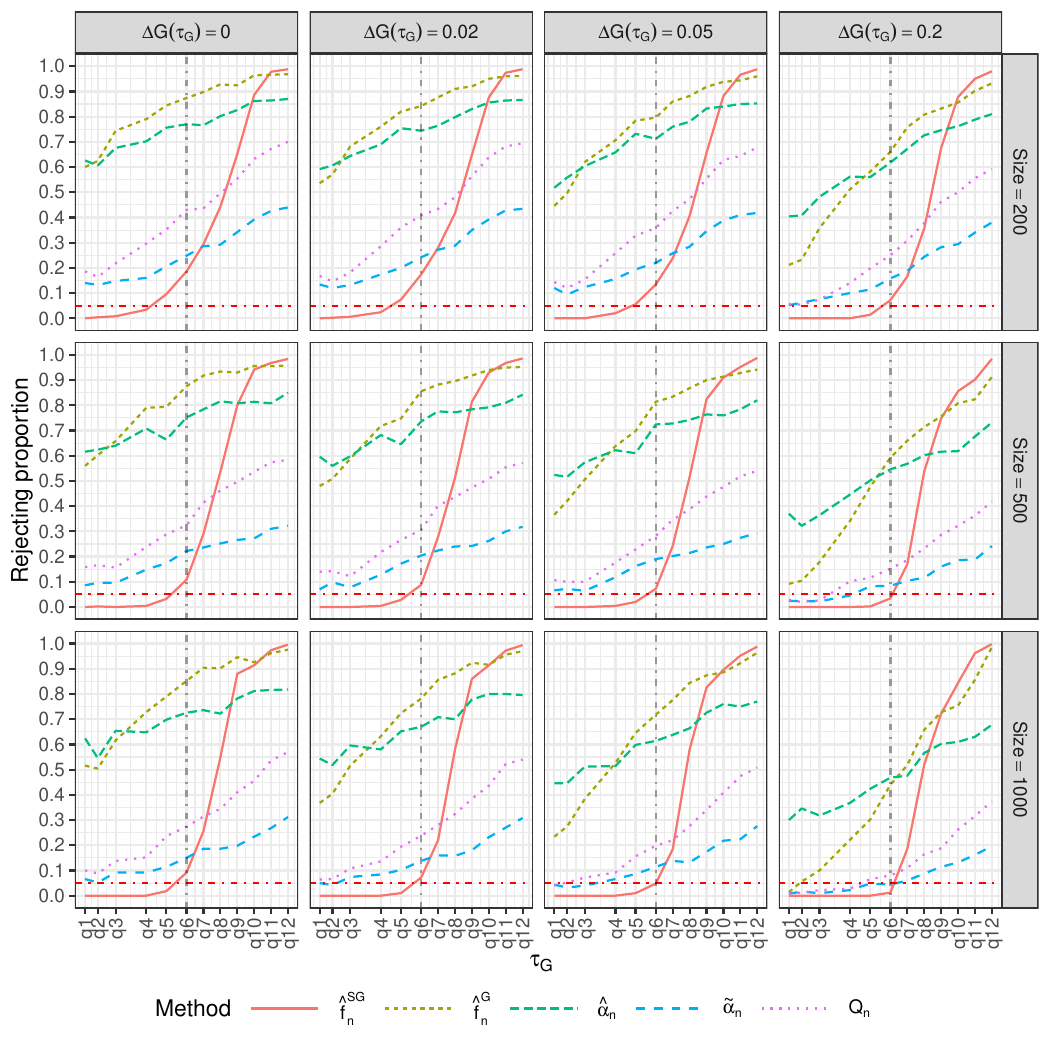}
		\caption{{Rejection rate of the null hypothesis of insufficient follow-up for different methods in} Setting \ref*{enum:sim_exp_unif} when $p=0.6$ (uncured fraction).\label{supp_fig:rej_prop_exp_unif_p_06}}
	\end{center}
\end{figure}

\begin{figure}[H]
	\begin{center}
		\includegraphics[width=\textwidth]{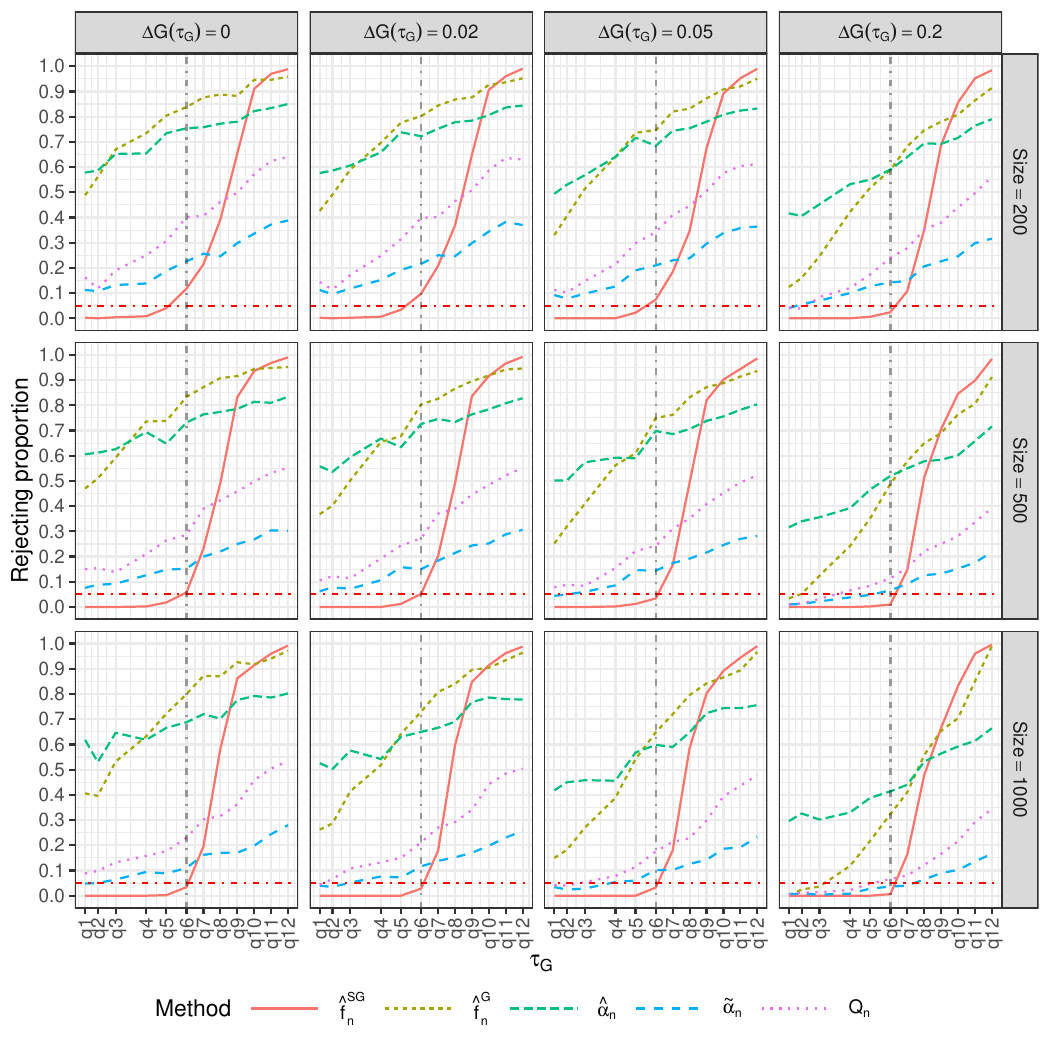}
		\caption{{Rejection rate of the null hypothesis of insufficient follow-up for different methods in} Setting \ref*{enum:sim_exp_unif} when $p=0.8$ (uncured fraction).\label{supp_fig:rej_prop_exp_unif_p_08}}
	\end{center}
\end{figure}

\clearpage
\subsubsection*{Sensitivity of $\tau$}
	\begin{figure}[H]
	\begin{center}
		\includegraphics[width=\textwidth]{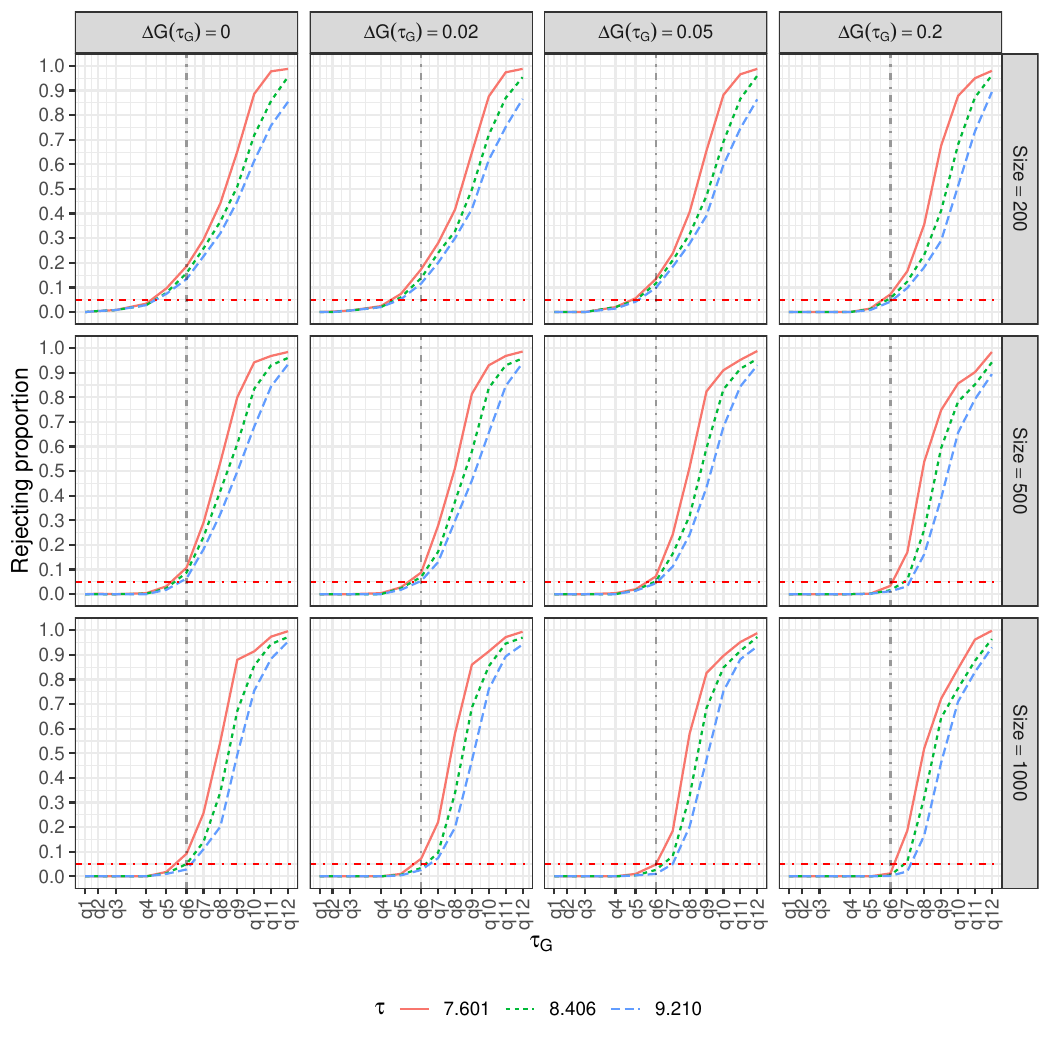}
		\caption{{Rejection rate of the null hypothesis of insufficient follow-up for the test based on $\hat{f}_{nh}^{SG}$ with different $\tau$} in Setting \ref*{enum:sim_exp_unif} when $p=0.6$ (uncured fraction).\label{supp_fig:rej_prop_exp_unif_p_06_tau}}
	\end{center}
\end{figure}

\clearpage
\subsubsection*{Comparison with RECeUS}
Figure \ref*{supp_fig:rej_prop_exp_unif_p_06_RECeUS} shows the rejection rate {of insufficient follow-up} against $\tau_{G}$ for Setting \ref*{enum:sim_exp_unif}, when $n=500$, $p=0.6$, and $\Delta G(\tau_G)=0.02$. The method RECeUS-AIC by \textcite{SO2023} is included for comparison in addition to $\hat{f}_{nh}^{SG}$, $\hat{f}_{n}^{G}$, $Q_n$. 
For the RECeUS-AIC procedure, we use the same thresholds specified in Section~2.1 of \textcite{SO2023} (2.5\% for $\hat{\pi}_n$ and 5\% for $\hat{r}_n$).
In summary, $\hat{f}_{nh}^{SG}$ performs better in terms of empirical level, while the RECeUS-AIC procedure demonstrates higher power in the region where $\tau_G$ lies between $q_6$ and $q_{10}$. When comparing the $Q_n$ test with RECeUS-AIC, the latter controls the level more effectively when $\tau_G$ is between $q_1$ and $q_3$, but not in the region from $q_3$ to $q_6$, although RECeUS-AIC exhibits higher empirical power than the $Q_n$ test.
\begin{figure}[H]
	\begin{center}
		\includegraphics[width=0.9\textwidth]{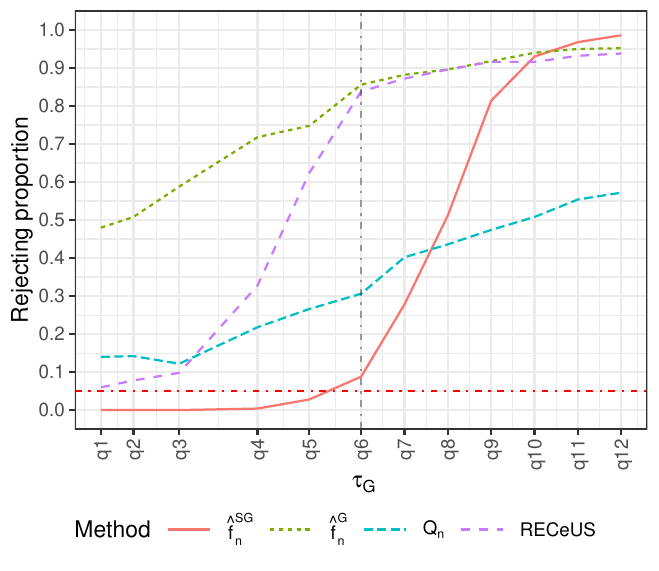}
		\caption{{Rejection rate of the null hypothesis of insufficient follow-up for different methods ($\hat{f}_{nh}^{SG}$, $\hat{f}_{n}^{G}$, $Q_n$ and RECeUS-AIC) in} Setting \ref*{enum:sim_exp_unif} when $n=500$, $\Delta G(\tau_G)=0.02$, and $p=0.6$ (uncured fraction).\label{supp_fig:rej_prop_exp_unif_p_06_RECeUS}}
	\end{center}
\end{figure}

\clearpage
\subsubsection{Setting \ref*{enum:sim_exp_exp_trunc}}
Figures \ref*{supp_fig:rej_prop_exp_exp_trunc_p_02}--\ref*{supp_fig:rej_prop_exp_exp_trunc_p_08} depict the rejection rate {of insufficient follow-up} against $\tau_{G}$ for Setting \ref*{enum:sim_exp_exp_trunc}, each figure with different $p$. In Setting \ref*{enum:sim_exp_exp_trunc}, the uncured subjects have an exponential distribution with rate $\lambda\in\{0.4,1,5\}$, while the censoring time has an exponential distribution with the rate fixed at 0.5. The censoring rate decreases as $\lambda$ increases, which affects the performance of the testing procedure. {In particular, $\lambda=0.4$ corresponds to a very high censoring rate since the density of the censoring variable decreases quicker than that of the event times and all methods have problems in controlling the level of the test.} We  observe that each method performs better, in terms of empirical level and power, when $\lambda$ is larger. 
\begin{figure}[H]
	\begin{center}
		\includegraphics[width=\textwidth]{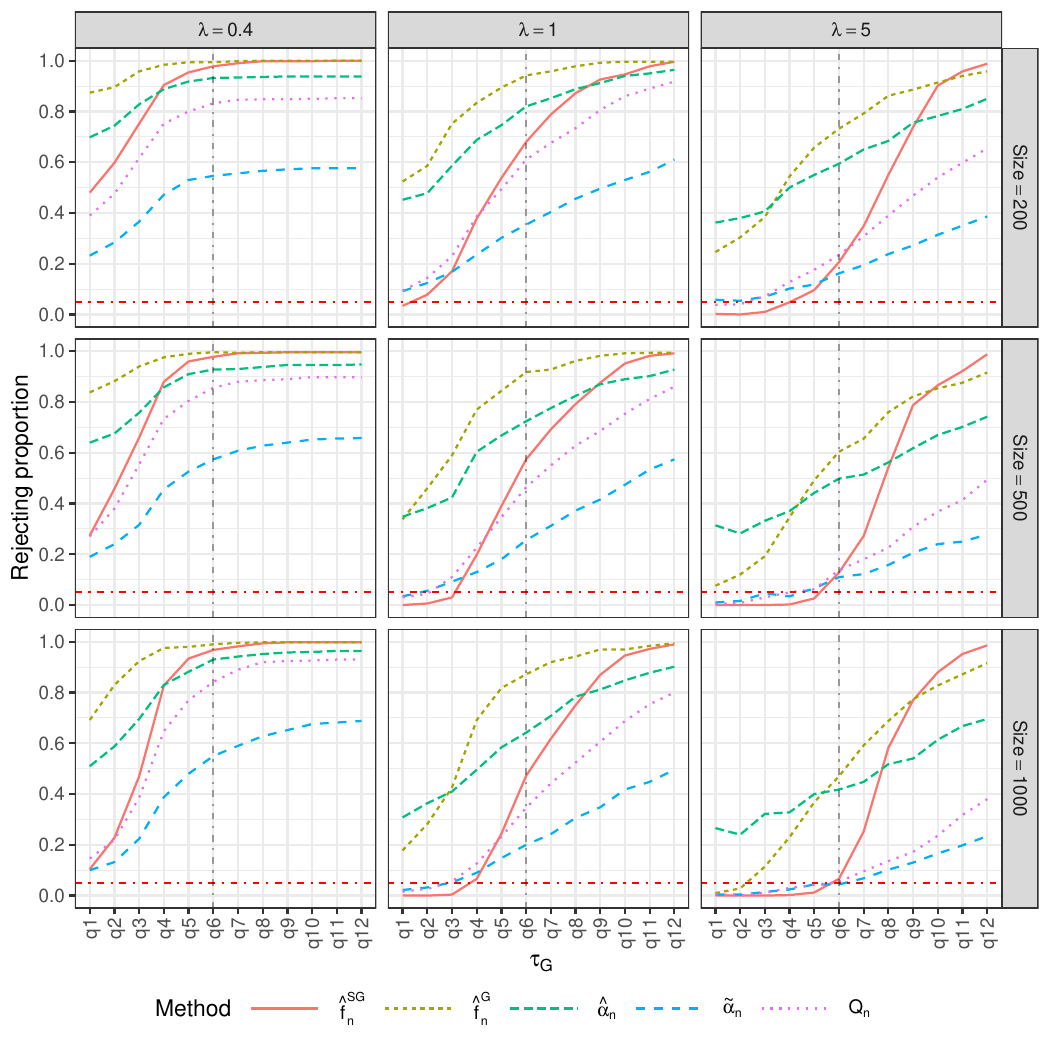}
		\caption{Rejection rate for Setting \ref*{enum:sim_exp_exp_trunc} when $p=0.2$ (uncured fraction).\label{supp_fig:rej_prop_exp_exp_trunc_p_02}}
	\end{center}
\end{figure}

\begin{figure}[H]
	\begin{center}
		\includegraphics[width=\textwidth]{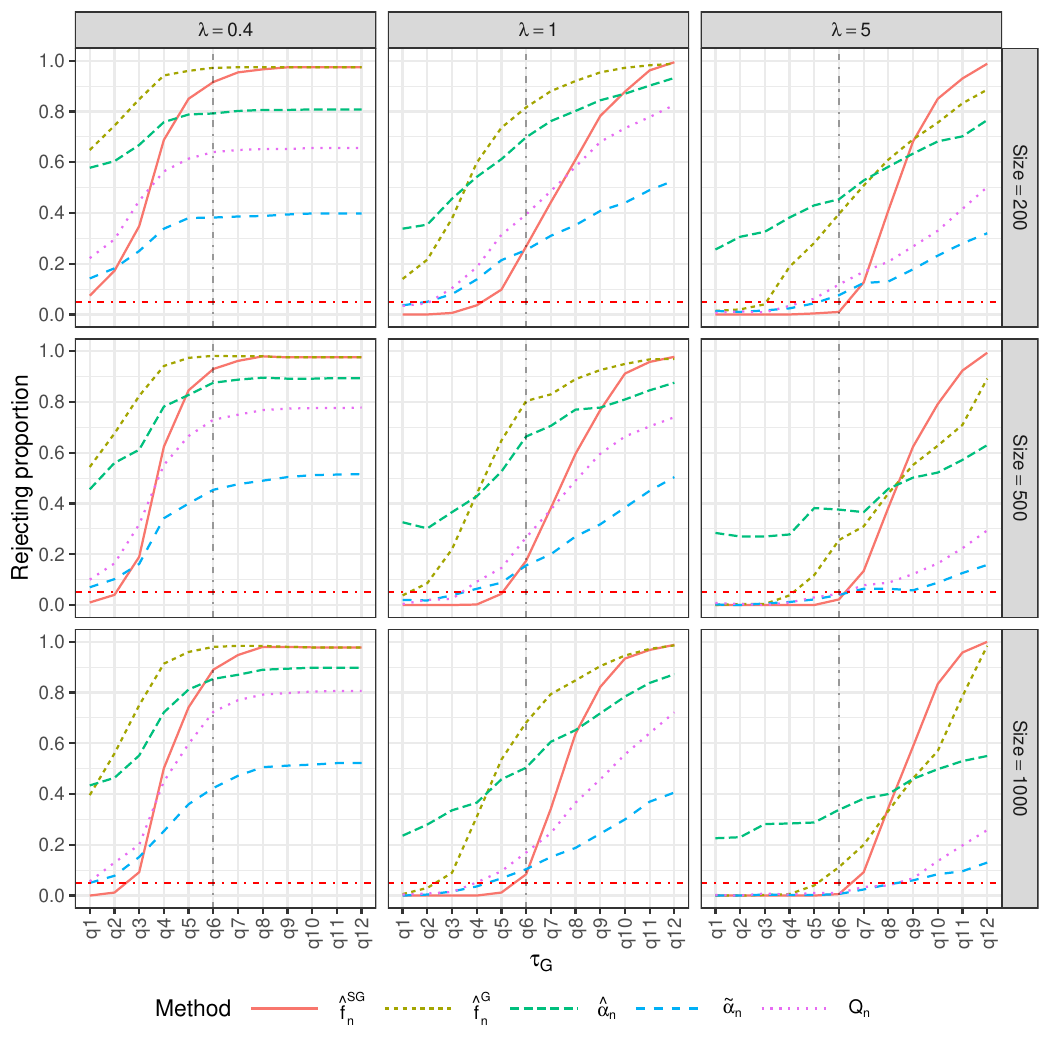}
		\caption{Rejection rate for Setting \ref*{enum:sim_exp_exp_trunc} when $p=0.6$ (uncured fraction).\label{supp_fig:rej_prop_exp_exp_trunc_p_06}}
	\end{center}
\end{figure}

\begin{figure}[H]
	\begin{center}
		\includegraphics[width=\textwidth]{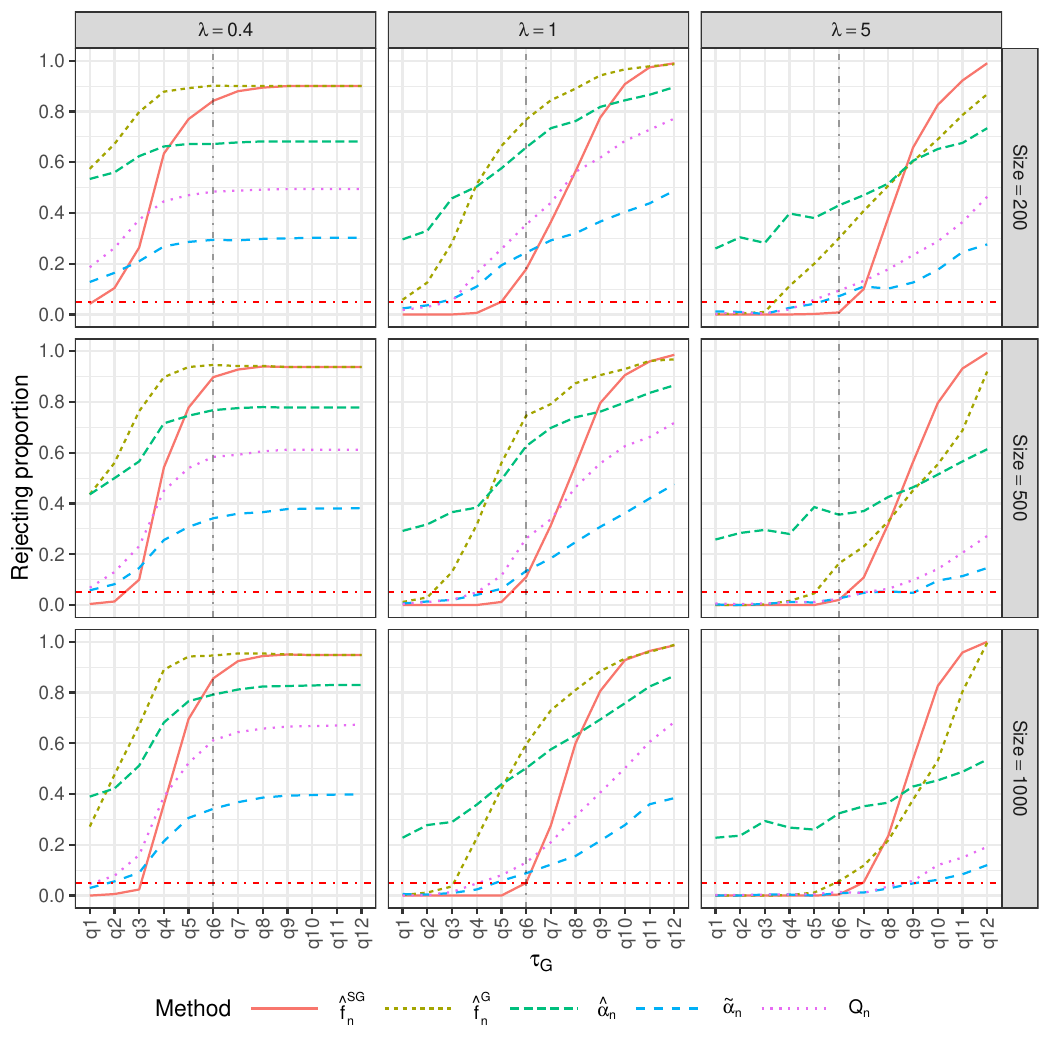}
		\caption{Rejection rate for Setting \ref*{enum:sim_exp_exp_trunc} when $p=0.8$ (uncured fraction).\label{supp_fig:rej_prop_exp_exp_trunc_p_08}}
	\end{center}
\end{figure}

\clearpage
\subsubsection{Setting \ref*{enum:sim_wb_unif}}
Figures \ref*{supp_fig:rej_prop_wb_unif_trunc_p_02}--\ref*{supp_fig:rej_prop_wb_unif_trunc_p_08} depict the rejection rate {of insufficient follow-up} against $\tau_{G}$ for Setting \ref*{enum:sim_wb_unif} when $p$ is 0.2, 0.6, and 0.8.
In this setting, the density $f_u$ decreases faster and we observe a worse behaviour in terms of empirical level, particularly when $p$, $n$ and $\Delta G(\tau_G)$ are smaller.
\begin{figure}[H]
	\begin{center}
		\includegraphics[width=\textwidth]{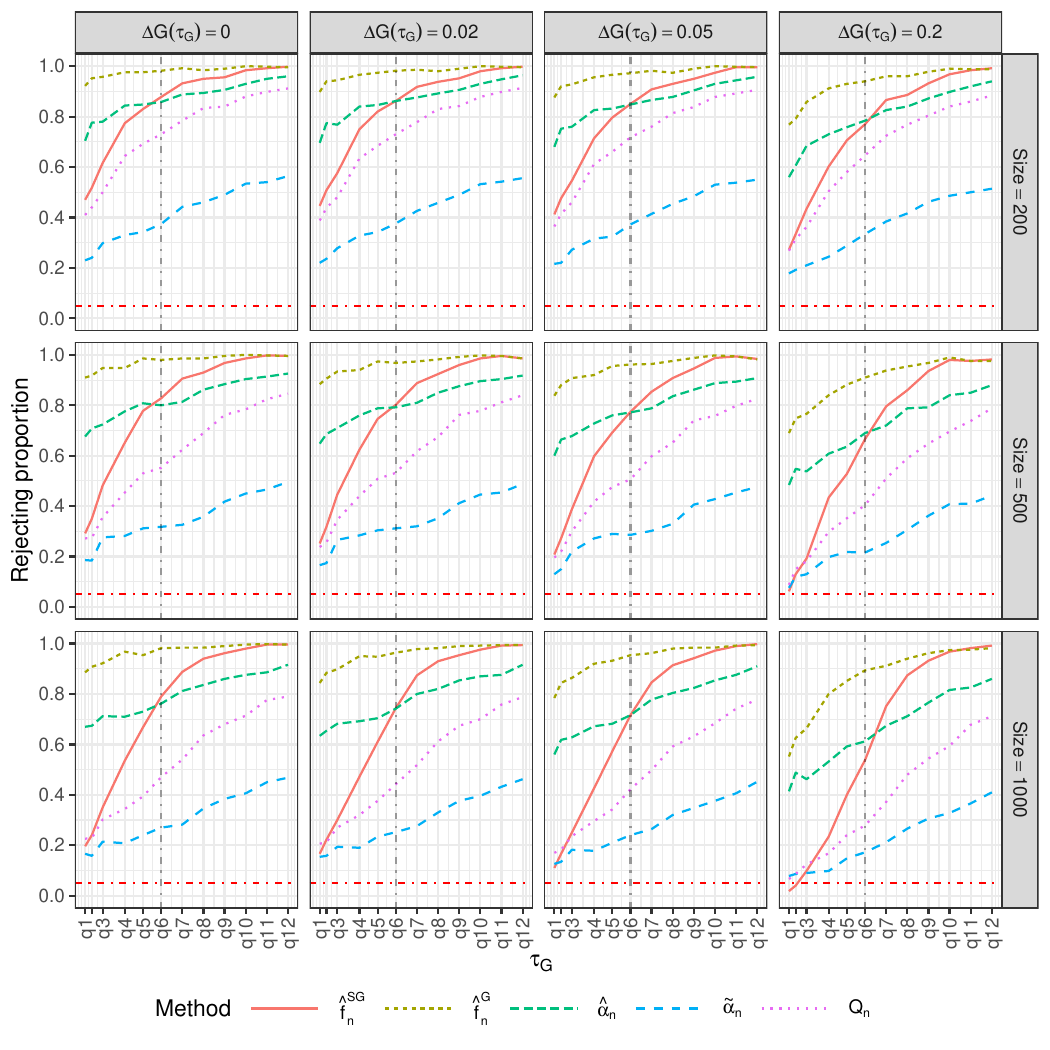}
		\caption{Rejection rate for Setting \ref*{enum:sim_wb_unif} when $p=0.2$ (uncured fraction).\label{supp_fig:rej_prop_wb_unif_trunc_p_02}}
	\end{center}
\end{figure}

\begin{figure}[H]
	\begin{center}
		\includegraphics[width=\textwidth]{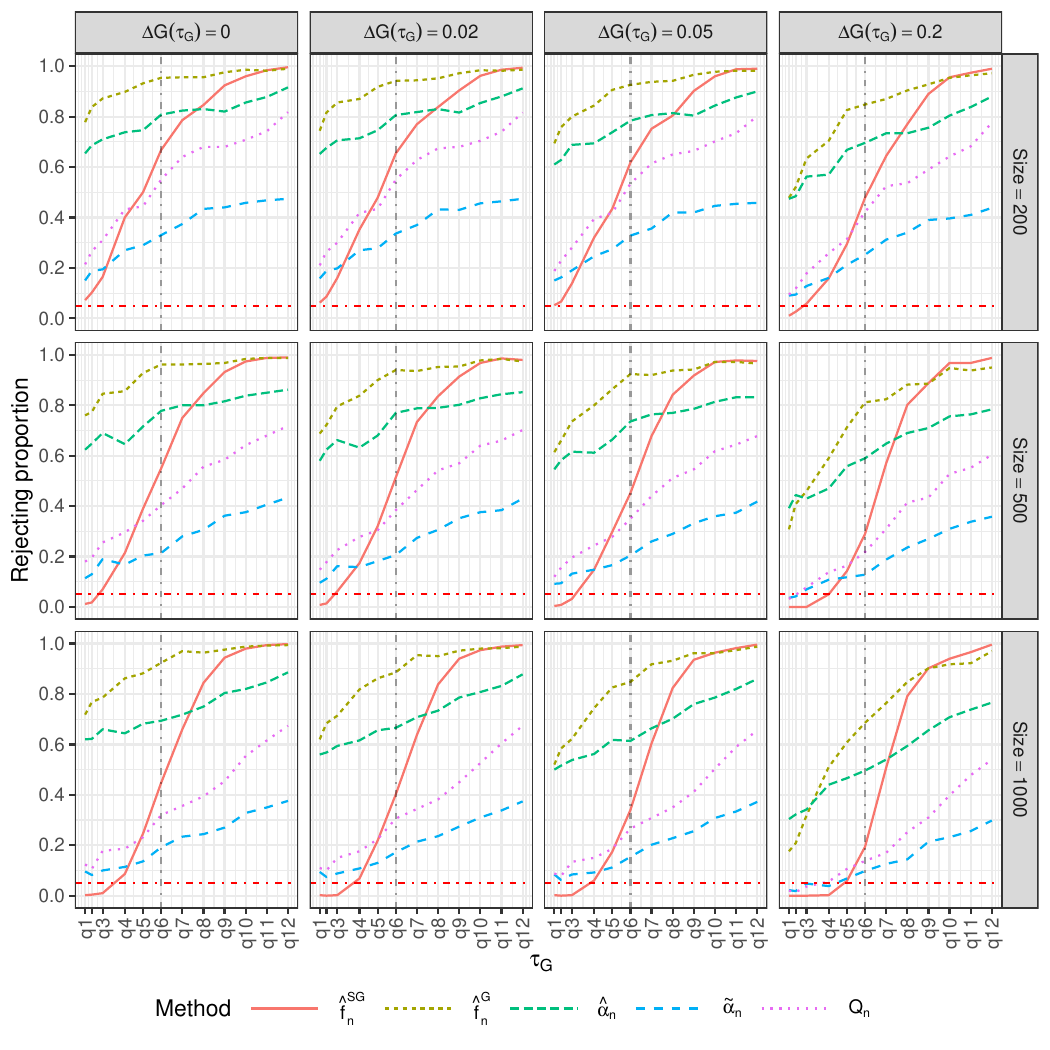}
		\caption{Rejection rate for Setting \ref*{enum:sim_wb_unif} when $p=0.6$ (uncured fraction).\label{supp_fig:rej_prop_wb_unif_trunc_p_06}}
	\end{center}
\end{figure}

\begin{figure}[H]
	\begin{center}
		\includegraphics[width=\textwidth]{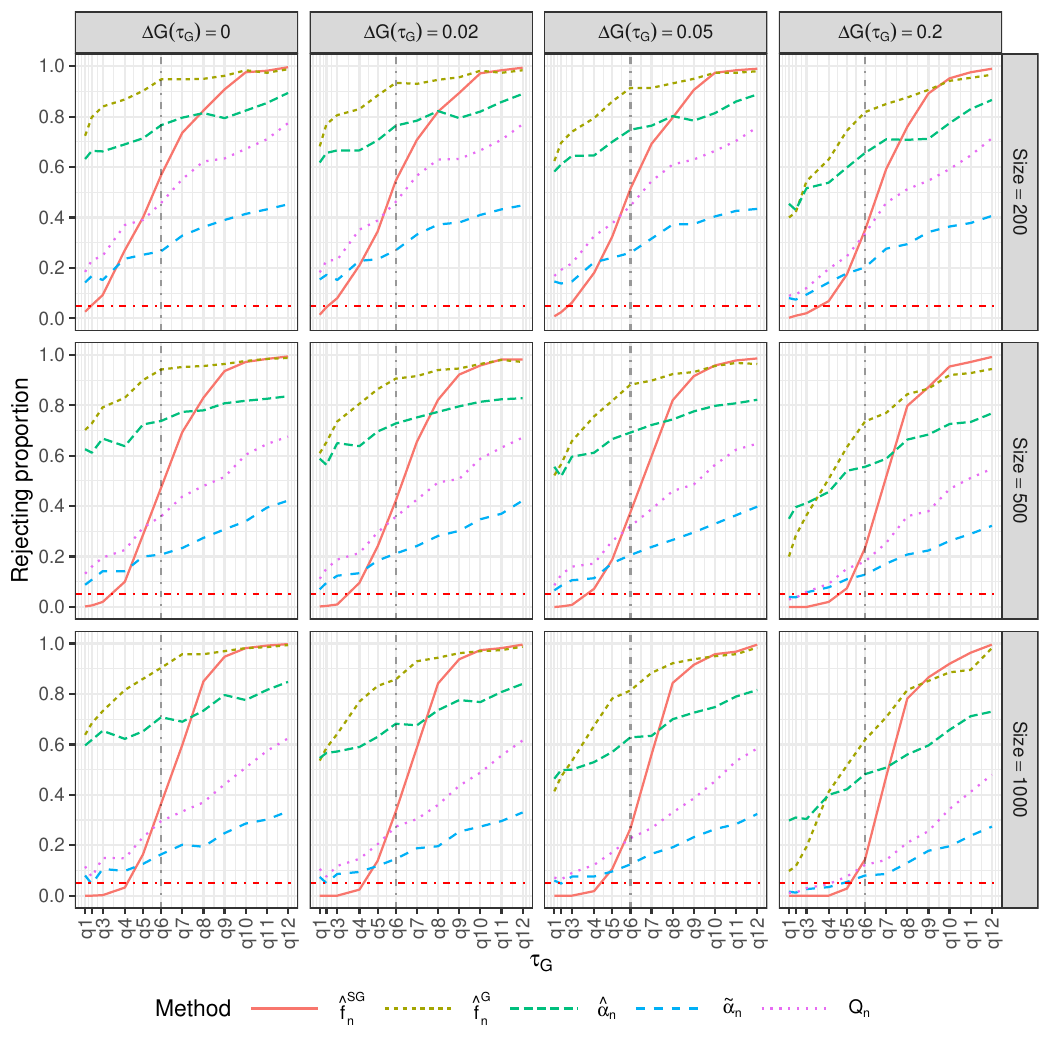}
		\caption{Rejection rate for Setting \ref*{enum:sim_wb_unif} when $p=0.8$ (uncured fraction).\label{supp_fig:rej_prop_wb_unif_trunc_p_08}}
	\end{center}
\end{figure}

\clearpage
\subsubsection*{Changing $\epsilon$ from 0.01 to 0.005}
Figure \ref*{supp_fig:rej_prop_wb_unif_trunc_p_06_eps} shows the rejection rates {of insufficient follow-up} for $\hat{f}_{nh}^{SG}$, when $\epsilon=0.005$ is used for $\tilde{H}_{0}:\tau_{G}\leq q_{1-\epsilon}$, meaning that we consider the follow-up as insufficient when $\tau_{G}$ is less than or equal to the 99.5\% quantile of $F_{u}$. Under such a stricter characterization of insufficient follow-up, the {rejection proportion is lower compared} to the results for $\epsilon=0.01$. Note that the parameter $\tau$ is fixed to the 99.95\% quantile of $F_{u}$ for both cases.
\begin{figure}[H]
	\begin{center}
		\includegraphics[width=\textwidth]{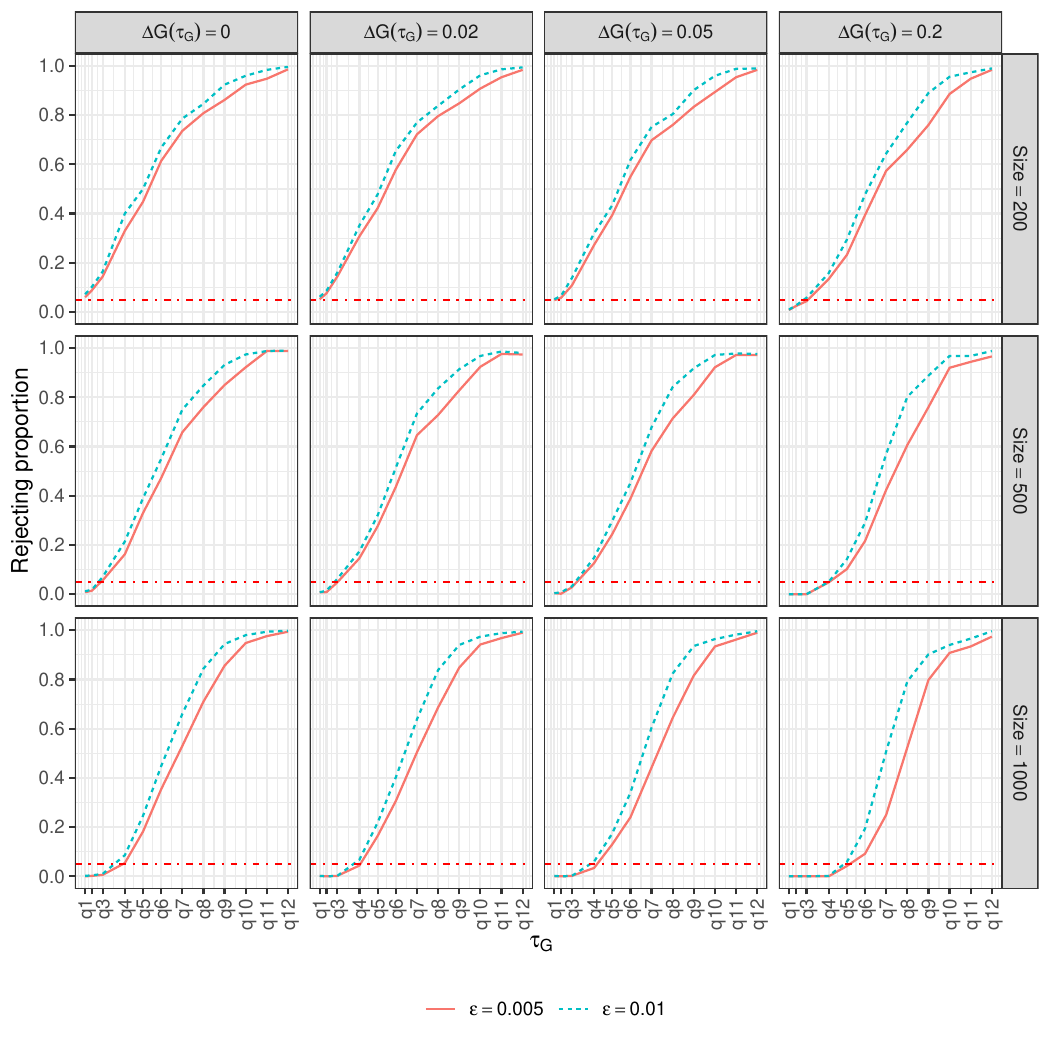}
		\caption{Rejection rate  {of insufficient follow-up for the test based on} $\hat{f}_{nh}^{SG}$ in Setting \ref*{enum:sim_wb_unif} when $p=0.6$ (uncured fraction). $q_6$ corresponds to $q_{0.99}$; and $q_{0.995}$ locates between $q_{7}$ and $q_{8}$.\label{supp_fig:rej_prop_wb_unif_trunc_p_06_eps}}
	\end{center}
\end{figure}

\clearpage
\subsubsection{Setting \ref*{enum:sim_texp_unif}}
\begin{figure}[H]
	\begin{center}
		\includegraphics[width=\textwidth]{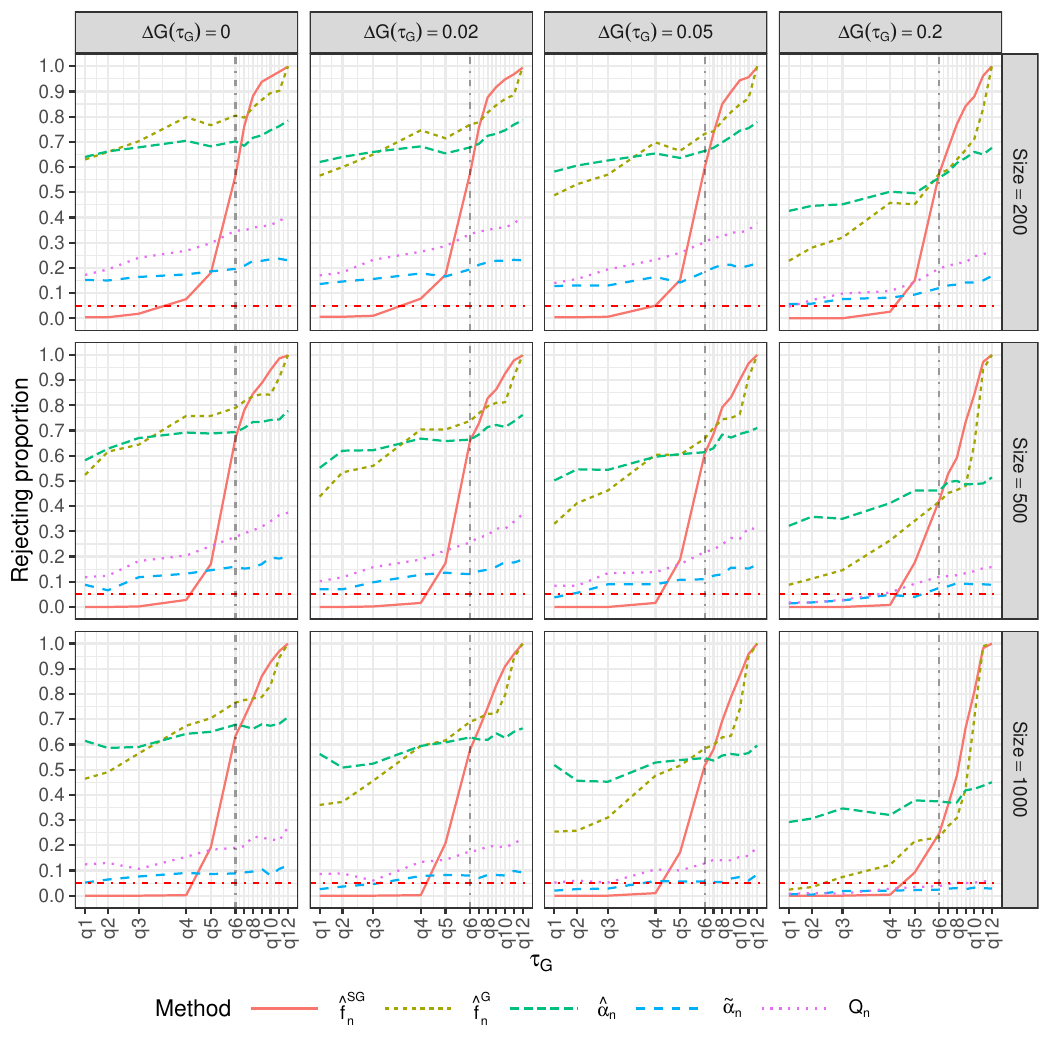}
		\caption{Rejection rate {of insufficient follow-up for different methods in}  Setting \ref*{enum:sim_texp_unif} when $p=0.6$ (uncured fraction).}
	\end{center}
\end{figure}

\clearpage
\subsubsection{Setting \ref*{enum:sim_texp_exp_trunc}}
Figure \ref*{supp_fig:rej_prop_texp_exp_trunc_p_06} shows the rejection rate {of insufficient follow-up} against $\tau_{G}$ for Setting \ref*{enum:sim_texp_exp_trunc}. The uncured subjects have a truncated exponential distribution with parameter $5$ and the censoring time follows an exponential distribution with rate $\lambda_{C}\in\{0.5,3\}$ and truncated at $\tau_{G}$. The censoring becomes heavier when $\lambda_{C}$ increases. $\hat{f}_{nh}^{SG}$ has a steeper rejection curve among the investigated methods, although the rejection rate of $\hat{f}_{nh}^{SG}$ is higher than the significance level when $\tau_{G}$ is between $q_{4}$ and $q_{6}$, and $\lambda_{C}=3$. {The $Q_n$ test shows very little power to detect sufficient follow-up, particularly when the censoring rate is higher.}
\begin{figure}[H]
	\begin{center}
		\includegraphics[width=\textwidth]{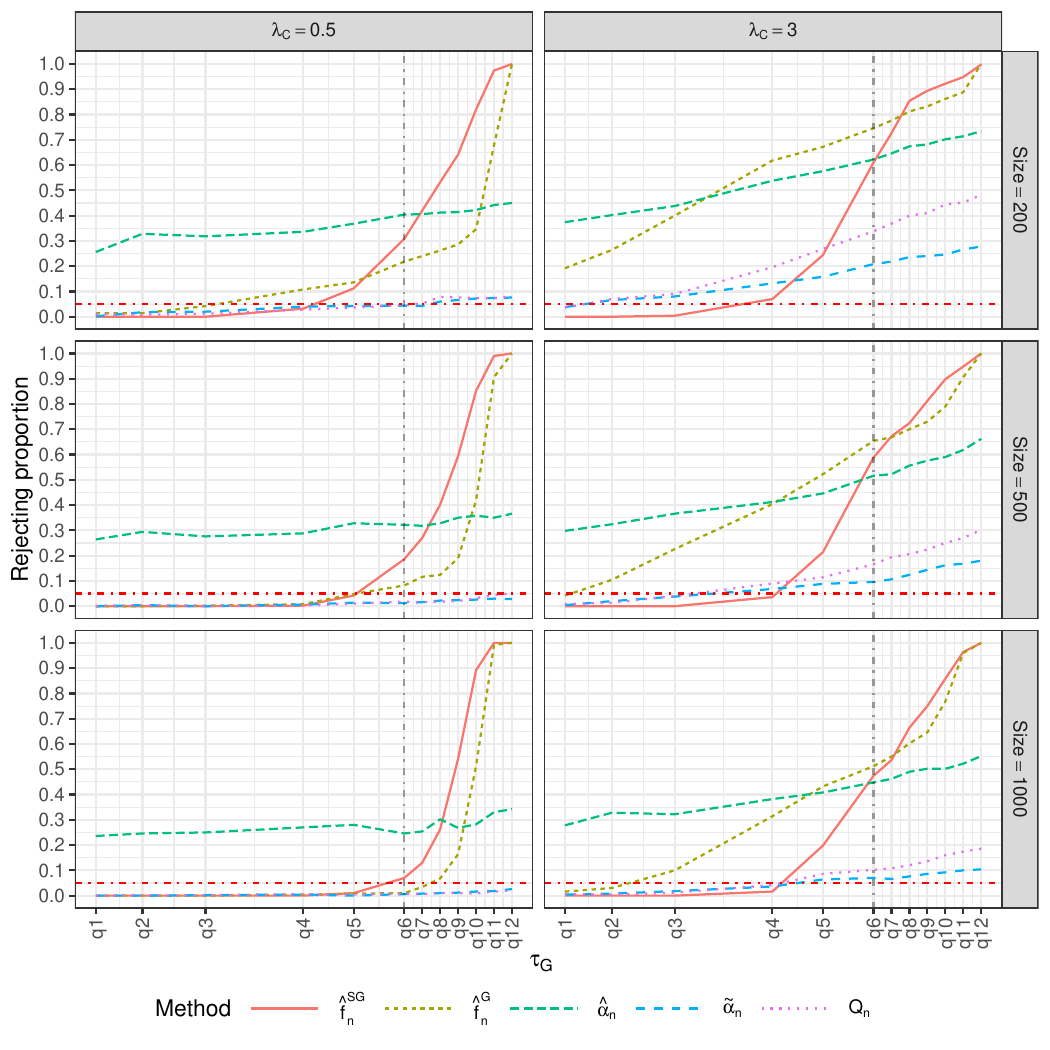}
		\caption{Rejection rate  Setting \ref*{enum:sim_texp_exp_trunc} when $p=0.6$ (uncured fraction).\label{supp_fig:rej_prop_texp_exp_trunc_p_06}}
	\end{center}
\end{figure}

\clearpage
\subsubsection{Setting \ref*{enum:sim_exp_exp}}

In Setting \ref*{enum:sim_exp_exp}, the uncured subjects have an exponential distribution with rate of 1 and the censoring time follows an exponential distribution with rate of 0.5. {The censoring rates, when $p=0.2$, 0.6 and 0.8, are around 86\%, 60\% and 46\%, respectively.} {In particular} $\tau_{F_{u}}=\infty$ and $\tau_{G}=\infty$, and this setting is regarded as sufficient follow-up under $\check{H}_{0}: \tau_{F_{u}} \leq \tau_{G}$ {but it is not sufficient follow-up under our formulation $\tilde{H}_{0}: q_{1-\epsilon}\geq\tau_{G}$}. Table \ref*{supp_tab:rej_prop_exp_exp} shows {the rejection proportion of null hypotheses considered by the different tests.} For the $T_{n}$-test, {for which $\check{H}_0$ is true}, the empirical level is around the significance level. For {computation of our test statistics we use the maximum observed survival time $y_{(n)}$ instead of $\tau_G$ for each generated dataset, meaning that essentially, we consider the follow-up as insufficient when $q_{1-\epsilon}\geq y_{(n)}$. {For the simulation with $p=0.8$ and sample size of 200, there are 5 out of 500 replications that  $q_{1-\epsilon}< y_{(n)}$. This is not observed for the remaining cases.} The parameter $\tau$ is set to $y_{(n)}+\delta$, where $\delta\in\{1,2\}$, which is considered reasonable given the range of values of $y_{(n)}$}. The empirical power of the proposed methods using $\hat{f}_{nh}^{SG}$ and $\hat{f}_{n}^{G}$ are close to 1, among all different $p$'s and sample sizes.
\begin{table}[h]
	{
		\centering
		\caption{Simulation result (in terms of rejection rate of the null hypotheses) for Setting \ref*{enum:sim_exp_exp}. $\tau$ is set to $y_{(n)}+\delta$ for the tests based on  $\hat{f}_{n}^{G}$ and $\hat{f}_{nh}^{SG}$.\label{supp_tab:rej_prop_exp_exp}}
		{
			\renewcommand{\arraystretch}{1.2}
			\small
			\setlength{\tabcolsep}{5pt}
			\begin{tabular}{cl@{\extracolsep{6pt}}c|cc|ccc|c}

&                       &          & \multicolumn{2}{c|}{$\tilde{H}_{0}: q_{1-\epsilon}\geq\tau_{G}$}                                            &\multicolumn{3}{c|}{${H}_{0}: \tau_{F_u}\geq\tau_{G}$}   & \shortstack{$\check{H}_{0}: \tau_{F_{u}}\leq\tau_{G}$} \\ 

$p$                  & size                  & $\delta$ & $\hat{f}_{n}^{G}$ & $\hat{f}_{nh}^{SG}$ & $\alpha_{n}$           & $\tilde{\alpha}_{n}$   & $Q_n$                  & $T_n$                                        \\ \hline
\multirow{6}{*}{0.2} & \multirow{2}{*}{200}  & 1        & 1                 & 1                  & \multirow{2}{*}{0.998} & \multirow{2}{*}{0.814} & \multirow{2}{*}{0.996} & \multirow{2}{*}{0}                       \\
&                       & 2        & 1                 & 1                  &                        &                        &                        &                                              \\ \cline{2-9} 
& \multirow{2}{*}{500}  & 1        & 1                 & 1                  & \multirow{2}{*}{0.994} & \multirow{2}{*}{0.854} & \multirow{2}{*}{0.994} & \multirow{2}{*}{0.002}                       \\
&                       & 2        & 1                 & 1                  &                        &                        &                        &                                              \\ \cline{2-9} 
& \multirow{2}{*}{1000} & 1        & 1                 & 1                  & \multirow{2}{*}{1}     & \multirow{2}{*}{0.898} & \multirow{2}{*}{1}     & \multirow{2}{*}{0.002}                       \\
&                       & 2        & 1                 & 1                  &                        &                        &                        &                                              \\ \hline
\multirow{6}{*}{0.6} & \multirow{2}{*}{200}  & 1        & 0.992             & 0.988              & \multirow{2}{*}{0.976} & \multirow{2}{*}{0.734} & \multirow{2}{*}{0.956} & \multirow{2}{*}{0}                       \\
&                       & 2        & 0.992             & 0.984              &                        &                        &                        &                                              \\ \cline{2-9} 
& \multirow{2}{*}{500}  & 1        & 1                 & 1                  & \multirow{2}{*}{0.992} & \multirow{2}{*}{0.79}  & \multirow{2}{*}{0.98}  & \multirow{2}{*}{0}                           \\
&                       & 2        & 1                 & 0.998              &                        &                        &                        &                                              \\ \cline{2-9} 
& \multirow{2}{*}{1000} & 1        & 1                 & 1                  & \multirow{2}{*}{0.994} & \multirow{2}{*}{0.876} & \multirow{2}{*}{0.99}  & \multirow{2}{*}{0}                           \\
&                       & 2        & 1                 & 1                  &                        &                        &                        &                                              \\ \hline
\multirow{6}{*}{0.8} & \multirow{2}{*}{200}  & 1        & 0.974             & 0.942              & \multirow{2}{*}{0.948} & \multirow{2}{*}{0.624} & \multirow{2}{*}{0.876} & \multirow{2}{*}{0.008}                       \\
&                       & 2        & 0.974             & 0.928              &                        &                        &                        &                                              \\ \cline{2-9} 
& \multirow{2}{*}{500}  & 1        & 0.996             & 0.994              & \multirow{2}{*}{0.97}  & \multirow{2}{*}{0.728} & \multirow{2}{*}{0.94}  & \multirow{2}{*}{0.014}                       \\
&                       & 2        & 0.996             & 0.986              &                        &                        &                        &                                              \\ \cline{2-9} 
& \multirow{2}{*}{1000} & 1        & 0.998             & 0.998              & \multirow{2}{*}{0.98}  & \multirow{2}{*}{0.81}  & \multirow{2}{*}{0.966} & \multirow{2}{*}{0.008}                        \\
&                       & 2        & 0.998             & 0.998              &                        &                        &                        &                                              \\ \hline
\end{tabular}
		}\par
	}
\end{table}
\clearpage

\subsubsection{Setting \ref*{enum:sim_lnorm_unif}}
In Setting \ref*{enum:sim_lnorm_unif}, the uncured fraction $p$ is 0.6, the censoring time follows the uniform distribution as in Setting~1 with $\tau_G=6.5$ and $\Delta G(\tau_G-)=0.02$. The uncured survival time follows a mixture of two log-normal distributions with the following distribution function:
\[
F_u(t) = 0.7\Phi\left(\frac{\log(t)-\mu_1}{\sigma_1}\right) + 0.3\Phi\left(\frac{\log(t)-\mu_2}{\sigma_2}\right),
\]
where $\mu_1=0$, $\sigma_1=1$, $\mu_2=\log8$, $\sigma_2=0.3$, and $\Phi$ is the cumulative distribution function of the standard normal distribution. The mean of the second log-normal distribution is $\exp(\mu_2+\sigma_2^2/2)\approx8.3682>\tau_G$. The density $f_u$ is non-monotone and therefore the non-increasing assumption of $f_u$ for the proposed method is not satisfied. This setting is regarded as insufficient follow-up under $\tilde{H}_0:q_{1-\epsilon}\geq\tau_G$ with $\epsilon=0.01$ or $q_{1-\epsilon}\approx14.44$. The censoring rate of this setting is around 67.06\%. The simulation study is carried out using a sample size of 500 and 500 replications. The parameter $\tau$ for the proposed method is set to the 99.95\% quantile of $F_u$, which is about 24.64. Table~\ref*{supp_tab:sim_lnorm_unif} shows the rejection rate of insufficient follow-up for different methods. For the proposed method using $\hat{f}_{nh}^{SG}$ and the $Q_n$ test, the rejection rate is below the nominal level of 5\%. 
\begin{table}[h]
	{
		\centering
		\caption{Rejection rate of insufficient follow-up for different methods in Setting \ref*{enum:sim_lnorm_unif}. \label{supp_tab:sim_lnorm_unif}}
		{
			\renewcommand{\arraystretch}{1.2}
			\small
			\setlength{\tabcolsep}{5pt}
			\begin{tabular}{c|c}
				Method & Rejection rate \\ \hline
				$\hat{f}_{n}^{G}$ &  0.390              \\
				$\hat{f}_{nh}^{SG}$ & 0.008               \\
				$Q_n$ &               0.038 \\ \hline
			\end{tabular}
		}\par
	}
\end{table}


\section{Technical lemmas and proofs}
Recall that $\hat{F}_{n}$ is the Kaplan--Meier estimator of the distribution function $F$. Assume that $F$ is continuous and G is right-continuous. \textcite{MR1988} showed that the process $\hat{F}_{n}(t)-F(t)$ can be approximated by a Gaussian process on the time interval $[0, t_{0}]$ with $t_{0}<\tau_{H}$. Note that the right extreme of the observed event time $\tau_{H}$ can be expressed as $\tau_{H}=\tau_{F}\wedge\tau_{G}$. Note that $\tau_{F}=\infty$ in the presence of cured subjects. We therefore only consider the time up to $\tau_{G}$. With an additional assumption on the censoring distribution that $G$ has a jump at its right extreme $\tau_{G}$, such result can be extended to the time interval $[0, \tau_{G}]$. This can be shown using the construction in Remark 3 of \textcite{MR1988}.
\begin{lemma}
	\label{lemma:strong_approx}
	Suppose the censoring distribution $G$ has a jump at $\tau_{G}$, i.e. $\Delta G(\tau_{G})=1-G(\tau_{G}-)>0$. The we have, for $x>0$,
	\begin{equation*}
		\prob[\sbracket]{
			\sup_{t \leq \tau_{G}}{n}
			\left\vert
			\hat{F}_{n}(t)-F(t) - n^{-1/2}\cbracket{1-F(t)}W\circ L(t)
			\right\vert
			> K_{1}\log n + x
		} < K_{2}e^{-K_{3}x},
	\end{equation*}
	where $K_{1}$, $K_{2}$ and $K_{3}$ are positive constants, $W$ is a Brownian motion, and 
	\[
	L(t) = \int_{0}^{t}
	\frac{\dd F(u)}{\rbracket{1-G(u-)}\rbracket{1-F(u)}^{2}}.
	\]
\end{lemma}

With Lemma \ref*{lemma:strong_approx}, we apply the arguments in \textcite{LM2017} to show the asymptotic normality of $\hat{f}_{nh}^{SG}(\tau_{G})$ below. We denote the least concave majorant of the Kaplan--Meier estimator by $\hat{F}_{n}^{G}$.
\begin{proof}[Proof of Theorem \ref*{thm:sg_normality}]
	To show the asymptotic normality of $\hat{f}_{nh}^{SG}(\tau_{G})$, we first decompose $\hat{f}_{nh}^{SG}(\tau_{G}) - f(\tau_{G})$ into three parts:
	\begin{align*}
		\begin{split}
			\hat{f}_{nh}^{SG}(\tau_{G}) - f(\tau_{G})
			&=
			\int_{\tau_{G} - h}^{\tau_{G}}\frac{1}{h}k_{B,\tau_{G}}\rbracket{\frac{\tau_{G}-u}{h}}\dd{F(u)}-f(\tau_{G})\\
			&\quad+
			\int_{\tau_{G} - h}^{\tau_{G}}\frac{1}{h}k_{B,\tau_{G}}\rbracket{\frac{\tau_{G}-u}{h}}\dd{(\hat{F}_{n}-F)(u)}\\
			&\quad+
			\int_{\tau_{G} - h}^{\tau_{G}}\frac{1}{h}k_{B,\tau_{G}}\rbracket{\frac{\tau_{G}-u}{h}}\dd{(\hat{F}_{n}^{G}-\hat{F}_{n})(u)}\\
			&=(I)+(II)+(III).
		\end{split}
	\end{align*}
	For $(I)$, since $f_{u}$ is twice continuously differentiable and by the properties of the boundary kernel, we have
	\begin{align*}
		\begin{split}
			n^{2/5}(I)
			&=n^{2/5}
			\int_{\tau_{G} - h}^{\tau_{G}}\frac{1}{h}k_{B,\tau_{G}}\rbracket{\frac{\tau_{G}-u}{h}}\cbracket{f(u)-f(\tau_{G})}\dd{u}\\
			&=n^{2/5}
			\int_{0}^{1}k_{B,\tau_{G}}(v)\cbracket{f(\tau_{G}-hv)-f(\tau_{G})}\dd{v}\\
			&=n^{2/5}
			\int_{0}^{1}k_{B,\tau_{G}}(v)\cbracket{-f^{\prime}(\tau_{G})hv+\frac{1}{2}f^{\prime\prime}(\xi_{n})h^{2}v^{2}}\dd{v}\\
			&\converge
			\frac{1}{2}c^{2}f^{\prime\prime}(\tau_{G})\int_{0}^{1}v^{2}k_{B,\tau_{G}}(v)\dd{v},\quad\text{as}~n\converge\infty,
		\end{split}
	\end{align*}
	where $0 < \tau_{G}-\xi_{n}<hv<h\converge 0$ as $n\converge\infty$.
	
	For $(III)$, with the strong approximation in Lemma \ref*{lemma:strong_approx}, we can apply the argument in the proof of Lemma 4.3 in \textcite{LM2017} to show that
	\begin{equation}
		\label{eq:sup_dist_LCM_KME_ord}
		\sup_{t\in[0,\tau_{G}]}\left\vert
		\hat{F}_{n}^{G}(t)-\hat{F}_{n}(t)
		\right\vert
		=O_{P}\rbracket{\frac{\log n}{n}}^{2/3}.
	\end{equation}
	The idea is to verify the four conditions in \textcite{DL2014} to establish such result. Therefore, $n^{2/5}(III)$ converges to zero in probability.
	Specifically, using the integration by parts and a change of variable, we have
		\[
			\begin{split}
				(III) 
				&= \frac{1}{h}k_{B,\tau_{G}}\rbracket{\frac{\tau_{G}-u}{h}}(\hat{F}_{n}^{G}-\hat{F}_{n})(u)\bigg\vert_{\tau_G-h}^{\tau_G}
				-\frac{1}{h}\int_{0}^{1}
					(\hat{F}_{n}^{G}-\hat{F}_{n})(\tau_G-hv)
					k_{B,\tau_{G}}^\prime(v)
				\dd{v}\\
				&=O_P\left(\frac{1}{h}\left(\frac{\log n}{n}\right)^{2/3}\right),
			\end{split}
		\]
		where the last equality follows from the result in \eqref{eq:sup_dist_LCM_KME_ord} along with the boundedness of $k$ and $k'$, and consequently of $k_{B,\tau_{G}}$ and $k_{B,\tau_{G}}^\prime$. Therefore, $n^{2/5}(III)=o_P(1)$ follows from the assumption that $hn^{1/5}\to c\in(0, \infty)$ as $n\to\infty$.

		For $(II)$, we have
		\[
			n^{2/5}(II)=\frac{1}{\sqrt{hn^{1/5}}}\int_{0}^{1}k_{B,\tau_{G}}(v)\dd{\hat{W}_{n}(v)},
		\]
		where, for $v\in[0,1]$,
		\begin{equation}
			\begin{split}
				\label{eq:W_split}
				\hat{W}_{n}(v)
				&=
				\sqrt{\frac{n}{h}}\left\lbrace
				\hat{F}_n(\tau_G-hv)-F(\tau_G-hv)-n^{-1/2}(1-F(\tau_G-hv))W\circ L(\tau_G-hv)
				\right\rbrace\\
				&\quad
				-\sqrt{\frac{n}{h}}\left\lbrace
				\hat{F}_n(\tau_G)-F(\tau_G)-n^{-1/2}(1-F(\tau_G))W\circ L(\tau_G)
				\right\rbrace\\
				&\quad
				+\sqrt{\frac{1}{h}}(1-F(\tau_G))\left\lbrace
				W\circ L(\tau_G-hv)-W\circ L(\tau_G)
				\right\rbrace\\
				&\quad
				+\sqrt{\frac{1}{h}}(F(\tau_G)-F(\tau_G-hv))\left\lbrace
				W\circ L(\tau_G-hv)
				\right\rbrace.
			\end{split}
		\end{equation}
		Using the strong approximation in Lemma \ref*{lemma:strong_approx},  together with the arguments from the proof of Theorem 4.4 in \textcite{LM2017}, it can be shown that the first two terms on the right-hand side of \eqref{eq:W_split} converges to 0 in probability, uniformly in $v$. The last term also converges to 0 in probability, uniformly in $v$, by applying the argument from the same proof in \textcite{LM2017}, which relies on the maximal inequality for Brownian motion. Therefore, $n^{2/5}(II)$ is dominated by the term
	\[
	\frac{1}{\sqrt{hn^{1/5}}}\int_{0}^{1}k_{B,\tau_{G}}(v)\dd{W_{n}(v)},
	\]
	where $W_{n}(v)=h^{-1/2}\rbracket{1-F(\tau_{G})}\cbracket{W\circ L(\tau_{G}-hv) - W\circ L(\tau_{G})}$, for $v\in[0,1]$. By scaling, time reversal, and symmetry of a  Brownian motion, we have
	\begin{equation}
		\label{eq:W_tilde_dist}
		\frac{1}{\sqrt{h}}\cbracket{W\circ L(\tau_{G}-hv) - W\circ L(\tau_{G})}\overset{d}{=}
		\tilde{W}\rbracket{\frac{L(\tau_{G})-L(\tau_{G}-hv)}{h}},
	\end{equation}
	where $\tilde{W}$ is a one-sided Brownian motion. By the uniform continuity of the one-sided Brownian motion on the compact interval $[0,1]$, we have
	\begin{equation}
		\label{eq:W_tilde_unif_conv}
		\sup_{v\in[0,1]}\left\vert
		\tilde{W}\rbracket{\frac{L(\tau_{G})-L(\tau_{G}-hv)}{h}} - \tilde{W}\rbracket{L^{\prime}(\tau_{G})v}
		\right\vert\converge[\mathbb{P}]0,
	\end{equation}
	where $L^{\prime}(\tau_{G})=\frac{f(\tau_{G})}{\rbracket{1-G(\tau_{G}-)}\rbracket{1-F(\tau_{G})}^{2}}$. 
	Using integration by parts, we have
		\[
			\frac{1}{\sqrt{hn^{1/5}}}\int_{0}^{1}k_{B,\tau_{G}}(v)\dd{W_{n}(v)}
			=
			\frac{1}{\sqrt{hn^{1/5}}}\left\lbrace
			k_{B,\tau_{G}}(v)W_n(v)\bigg\vert_0^1
			-\int_0^1W_n(v)\dd{k_{B,\tau_G}(v)}
			\right\rbrace.
		\]
		Using \eqref{eq:W_tilde_dist} followed by \eqref{eq:W_tilde_unif_conv}, it can be shown that the term on the right-hand side in the above display converges in distribution 
		to
		\[
			\begin{split}
				&\frac{1-F(\tau_G)}{\sqrt{c}}\left\lbrace
				k_{B,\tau_{G}}(v)\tilde{W}\rbracket{L^{\prime}(\tau_{G})v}\bigg\vert_0^1
				-\int_0^1\tilde{W}\rbracket{L^{\prime}(\tau_{G})v}\dd{k_{B,\tau_G}(v)}
				\right\rbrace\\
				&=\rbracket{1-F(\tau_{G})}\sqrt{L^{\prime}(\tau_{G})/c}\int_{0}^{1}k_{B,\tau_{G}}(v)\dd{\tilde{W}(v)},
			\end{split}
		\]
		where we used the rescaling of a Brownian motion and integration by parts. 
		Therefore we have
		\[
			\begin{split}
				\frac{1}{\sqrt{hn^{1/5}}}\int_{0}^{1}k_{B,\tau_{G}}(v)\dd{W_{n}(v)}
				&\converge[d]
				\rbracket{1-F(\tau_{G})}\sqrt{L^{\prime}(\tau_{G})/c}\int_{0}^{1}k_{B,\tau_{G}}(v)\dd{\tilde{W}(v)}\\
				&\sim
				N\rbracket{0, \frac{f(\tau_{G})}{c\lbrack1-G(\tau_{G}-)\rbrack}\int_{0}^{1}k_{B,\tau_{G}}(v)^{2}\dd{v}},
			\end{split}
		\]
	where we use that the integral with respect to the Brownian motion is a Normal distributed variable (see for example \textcite{durrett2018stochastic}).
\end{proof}
\begin{proof}[Proof of Proposition~\ref*{prop:level_grenander}]
	Let $\tau_G\leq q_{1-\epsilon}$. We have 
		\begin{align*}
			\begin{split}
				&\prob[\sbracket]{
					\hat{f}_{n}^{G}(\tau_{G} - cn^{-a}) 
					\leq 
					\frac{\epsilon \hat{F}_{n}(\tau_{G})}{\tau - \tau_{G}} - A_{1}^{-1}n^{-(1-a)/2}Q_{1 - \alpha}^{G}
				}\\
				&=
				\mathbb{P}\left[
				f(\tau_{G}) - \hat{f}_{n}^{G}(\tau_{G} - cn^{-a})
				\geq 
				f(\tau_{G}) -\frac{\epsilon \lbrace\hat{F}_{n}(\tau_{G}) - F(\tau_{G})\rbrace}{\tau - \tau_{G}} - \frac{\epsilon F(\tau_{G})}{\tau-\tau_{G}}\right.\\
				&\qquad\qquad\qquad\qquad\qquad\qquad\qquad
				+ A_{1}^{-1}n^{-(1-a)/2}Q_{1 - \alpha}^{G}
				\Bigg]\\%
				&\leq
				\mathbb{P}\left[
				A_{1}n^{(1-a)/2}
				\lbrace
				f(\tau_{G}) - \hat{f}_{n}^{G}(\tau_{G} - cn^{-a})
				\rbrace
				\right.\\
				&\qquad\qquad\geq 
				Q_{1 - \alpha}^{G}+\frac{A_{1}n^{(1-a)/2}p\eta }{\tau-\tau_{G}}- \frac{\epsilon A_{1}n^{(1-a)/2}\lbrace\hat{F}_{n}(\tau_{G})- F(\tau_{G})\rbrace}{\tau - \tau_{G}}\Bigg].
			\end{split}
		\end{align*}
		Here the last inequality follows from the fact that $\tilde{H}_{0}: q_{1 - \epsilon} \geq \tau_{G}$ entails $f(\tau_{G})\geq \frac{p(\epsilon-\eta)}{\tau-\tau_{G}}$, and $F(\tau_G)\leq p$. By Theorem~\ref*{theo:grenander} and Lemma~\ref*{lemma:strong_approx}, the probability on the right hand side converges to $\alpha$, which concludes the proof. Note that the upper bound on the rejection probability is still asymptotically bounded by $\alpha$ even if $A_1$ is replaced by a consistent estimator $\hat{A}_1$ since one would only have some additional terms involving $\hat{A}_1-A_1$, which would converge to zero.
\end{proof}

\begin{proof}[Proof of Proposition~\ref*{prop:level_SG}]
	Let $\tau_G\leq q_{1-\epsilon}$. We have
	\begin{align*}
		\begin{split}
			&\prob{\hat{f}_{nh}^{SG}(\tau_{G}) \leq \frac{\epsilon \hat{F}_{n}(\tau_{G})}{\tau - \tau_{G}} + n^{-2/5}Q_{\alpha}^{SG}}\\
			&=
			\prob{\hat{f}_{nh}^{SG}(\tau_{G}) - f(\tau_{G}) \leq \frac{\epsilon \lbrace\hat{F}_{n}(\tau_{G}) - F(\tau_{G})\rbrace}{\tau - \tau_{G}} + \frac{\epsilon F(\tau_{G})}{\tau-\tau_{G}} - f(\tau_{G}) + n^{-2/5}Q_{\alpha}^{SG}}\\
			&\leq
			\prob{\hat{f}_{nh}^{SG}(\tau_{G}) - f(\tau_{G}) \leq \frac{\epsilon \lbrace\hat{F}_{n}(\tau_{G}) - F(\tau_{G})\rbrace}{\tau - \tau_{G}} + \frac{p\eta}{\tau-\tau_{G}} + n^{-2/5}Q_{\alpha}^{SG}}\\
			&=
			\prob{n^{2/5}\lbrace\hat{f}_{nh}^{SG}(\tau_{G}) - f(\tau_{G})\rbrace \leq \frac{\epsilon n^{2/5}\lbrace\hat{F}_{n}(\tau_{G}) - F(\tau_{G})\rbrace}{\tau - \tau_{G}} + \frac{n^{2/5}p\eta }{\tau-\tau_{G}} + Q_{\alpha}^{SG}}.
		\end{split}
	\end{align*}
	Here the inequality follows from the fact that $\tilde{H}_{0}: q_{1 - \epsilon} \geq \tau_{G}$ entails $f(\tau_{G})\geq \frac{p(\epsilon-\eta)}{\tau-\tau_{G}}$, and {$F(\tau_G)\leq p$. By Theorem~\ref*{thm:sg_normality} and Lemma~\ref*{lemma:strong_approx}, the probability on the right hand side converges to $\alpha$, which concludes the proof.} Note that the upper bound on the rejection probability is still asymptotically bounded by $\alpha$ even if $Q_{\alpha}^{SG}$ is replaced by a consistent estimator $\hat{Q}_{\alpha}^{SG}$ since one would only have an additional term  $\hat{Q}_{\alpha}^{SG}-Q_{\alpha}^{SG}$, which would converge to zero.
\end{proof}

\section{Algorithm}
\begin{algorithm}[H]
	\caption{Bootstrapping procedure \label{algo:bootstrap_algo}}
	\begin{algorithmic}[1]
		\Require
		\Statex
		\begin{description}[itemsep=0mm]
			\item[Orignal sample data]
			$\cbracket{\rbracket{y_{i},\delta_{i}}: i=1,\cdots,n}$
			\item[Level of the test] $\alpha$
			\item[Number of bootstrap iterations] $B$
		\end{description}
		\State
		Obtain a smooth KME of $F$ by
		\[
		\tilde{F}_{nh_{0}}(t) = 
		\int_{(t - h_{0}) \vee 0}^{(t + h_{0}) \wedge \tau_{G}}
		\frac{1}{h_{0}}k^{(t)}\rbracket{
			\frac{t - v}{h_{0}}
		}\hat{F}_{n}^{G}(v)
		\dd{v},\quad t\in\lbrack 0, \tau_{G}\rbrack,
		\]
		where $\hat{F}_{n}^{G}$ is the least concave majorant (LCM) of the KME $\hat{F}_{n}$.\label{algo:bs_skm}
		\State {Compute the derivative $\tilde{f}_{nh_{0}}(\tau_{G})$ of $\tilde{F}_{nh_{0}}$ at $\tau_G$. }
		\State
		Estimate $G$ by the reversed KME, denoted by $\hat{G}_{n}$, using $\set{(y_{i}, \delta_{i})}{i=1,\cdots,n}$.
		\For{$b=1,\cdots,B$}
		\State
		Draw $(c_{b,1}^{\ast}, \cdots, c_{b,n}^{\ast})$ from $\hat{G}_{n}$;
		\State
		Draw $(t_{b,1}^{\ast}, \cdots, t_{b,n}^{\ast})$ from $\tilde{F}_{nh_{0}}$;
		\LineComment{Set $t_{b,i}^{\ast} = \infty$ if $u_{b,i} > \sup_{t}\tilde{F}_{nh_{0}}(t)$, where $u_{b,i}$ is the generated standard uniform random variate.}
		\State
		Construct $\lbrace{({y_{b,i}^{\ast},\delta_{b,i}^{\ast}}), i=1,\cdots,n}\rbrace$, where $y_{b,i}^{\ast}=t_{b,i}^{\ast} \wedge c_{b,i}^{\ast}$ and $\delta_{b,i}^{\ast} = \indicator{t_{b,i}^{\ast} \leq c_{b,i}^{\ast}}$.
		\State
		Obtain the bootstrap estimate $\hat{f}_{nh,b}^{SG^\ast}(\tau_{G})$ using $\lbrace{({y_{b,i}^{\ast},\delta_{b,i}^{\ast}}): i=1,\cdots,n}\rbrace$.
		\EndFor
		\State
		Approximate the critical value of the test by the ${\alpha}$-quantile of $\lbrace\hat{f}_{nh,b}^{SG^\ast}(\tau_{G}) - \tilde{f}_{nh_{0}}(\tau_{G}): b=1,\cdots,B\rbrace$
	\end{algorithmic}
\end{algorithm}

\clearpage
\section{Assumptions of $Q_n$ and $T_n$ test}
The assumptions needed for the $Q_n$ test to obtain the asymptotic distribution of $n\hat{q}_n$ when $\tau_G<\tau_{F_{u}}$ are:
\begin{enumerate}[label=(\alph*)]
	\item	$1-G(\tau_G-x)=a_G(1+o(1))x^\gamma$, as $x\downarrow0$, where $a_G$ and $\gamma$ are positive constants;
	\item	$F_u$ has a density $f_u$ in a neighborhood of $\tau_G$, which are positive and continuous at $\tau_G$.
\end{enumerate}
The assumptions needed for the $T_n$ test are:
\begin{enumerate}[label=(\alph*)]
	\item	$F_u$ is continuous at $\tau_{F_{u}}$ if $\tau_{F_{u}}<\infty$;
	\item	$\int_{0}^{\tau_{F_u}}\frac{\dd{F_u(s)}}{1-G(s-)}<\infty$;
	\item	$\lim_{n\converge\infty}n(1-G(\tau_G-x/\sqrt{n}))=\infty$ for all $x>0$;
	\item	$F_u$ has a finite second derivative $F_u''$ in $(\tau_{F_u}-\epsilon, \tau_{F_u})$ for some $\epsilon>0$ and satisfies
	\[
		\lim_{t\converge\tau_{F_{u}}}\frac{F_u''(t)(1-F_u(t))}{(F_u'(t))^2}=-1.
	\] 
\end{enumerate}

	\clearpage
	\begingroup
	\setlength\bibitemsep{\parsep}
	\setlength{\bibhang}{\leftmargini}
	\printbibliography[title={References for Supplementary Material}]
	\endgroup
	\end{refsection}
\end{document}